\documentclass[11pt]{article}
\usepackage{fullpage}

\usepackage[dvipsnames,svgnames]{xcolor}
\usepackage[bookmarks, colorlinks]{hyperref}
\hypersetup{%
    pdfborder = {0 0 0},
    bookmarksdepth = section,
    citecolor = DarkGreen,
    linkcolor = DarkRed,
    urlcolor = DarkGreen,
  }

\usepackage{amssymb}
 \usepackage{booktabs}
\usepackage{multirow}
\usepackage{amsmath}
\usepackage{amsthm}
\usepackage{graphicx,color,colordvi}
\usepackage{bbm}
\usepackage{cleveref}
\usepackage{stmaryrd}
\usepackage{underoverlap}
\usepackage[T1]{fontenc}
\usepackage[utf8]{inputenc}
\usepackage[shortlabels]{enumitem}
\usepackage{fullpage}
\usepackage[blocks]{authblk}
\usepackage{mathtools}
\usepackage{pgfplots}

\DeclareMathOperator{\interior}{int}
\DeclareMathOperator{\linspan}{span}
\newcommand{\sys}{\mathcal{S}}
\newcommand{\env}{\mathcal{E}}
\newcommand{\bZ}{\mathbb{Z}}
\newcommand{\init}{^\textnormal{i}}
\DeclarePairedDelimiter{\ceil}{\lceil}{\rceil}

\def\DJ{{\hbox{D\kern-.8em\raise.15ex\hbox{--}\kern.35em}}}

\def\openone{\leavevmode\hbox{\small1\kern-3.8pt\normalsize1}}

\def\CC{\mathbb{C}}
\def\bC{\mathbb{C}}

\def\NN{\mathbb{N}}

\def\11{\mathbb{I}}

\newcommand{\inv}{^{-1}}
\newtheorem{theorem}{Theorem}[section]
\newtheorem{lemma}[theorem]{Lemma}
\newtheorem{proposition}[theorem]{Proposition}
\newtheorem{corollary}[theorem]{Corollary}

\theoremstyle{definition}
\newtheorem{definition}[theorem]{Definition}
\newtheorem{example}[theorem]{Example}

\newcommand{\braket}[1]{ \langle #1 \rangle}

\def\reff#1{(\ref{#1})}
\def\eps{\varepsilon}

\newcommand{\tr}{\mathop{\rm Tr}\nolimits}
\newcommand{\re}{\mathop{\rm Re}\nolimits}

\newcommand{\spec}{{\rm sp}}

\newcommand{\rank}{\mathop{\rm rank}\nolimits}

\usepackage{pst-node}
\usepackage{tikz-cd} 

\newcommand{\bra}[1]{\langle#1|}
\newcommand{\ket}[1]{|#1\rangle}

\newcommand{\cB}{{\cal B}}
\newcommand{\cD}{{\cal D}}

\newcommand{\cF}{{\mathcal{F}}}

\newcommand{\cN}{{\cal N}}
\newcommand{\cT}{{\cal T}}
\newcommand{\cH}{{\cal H}}
\newcommand{\cK}{{\cal K}}
\newcommand{\cJ}{{\cal J}}

\newcommand{\cP}{\mathcal{P}}

\newcommand{\cL}{{\cal L}}

\newcommand{\cR}{{\cal R}}

\newcommand{\R}{{\mathbb{R}}}

\newcommand{\M}{{\mathbb{M}}}

\def\e{\mathrm{e}}
\newcommand{\iu}{\mathrm{i}\mkern1mu}

\newcommand{\lb}{\left(}
\newcommand{\rb}{\right)}

\usepackage{graphicx}
\usepackage{setspace}
\usepackage{verbatim}
\usepackage{subfig}

\usepackage{mathrsfs}

\theoremstyle{definition}

\theoremstyle{remark}
\newtheorem{remark}{Remark}
\newtheorem{remarks}[remark]{Remarks}
\newtheorem*{remarks*}{Remarks}

\numberwithin{equation}{section}

\newcommand{\EEB}{\operatorname{EEB}}

\newcommand{\EB}{\operatorname{EB}}
\newcommand{\LEA}{\operatorname{LEA}}
\newcommand{\AEB}{\operatorname{AEB}}
\newcommand{\AES}{\operatorname{AES}}
\newcommand{\ES}{\operatorname{ES}}
\newcommand{\spr}{\operatorname{spr}}

\newcommand{\SEP}{\operatorname{SEP}}
\newcommand{\PPT}{\operatorname{PPT}}

\newcommand{\Pas}{P_{\operatorname{asym}}}

\crefname{proposition}{prop.}{prop.}
\crefname{theorem}{thm.}{thm.}
\crefname{corollary}{cor.}{cor.}

\newcommand{\p}{(}
\DeclareRobustCommand\openone{\leavevmode\hbox{\small1\normalsize\kern-.33em1}}

\newcommand{\one}{\mathbb{I}}
\newcommand{\id}{{\rm{id}}}
\newcommand{\be}{\begin{equation}}
    \newcommand{\ee}{\end{equation}}
\newcommand{\bea}{\begin{eqnarray}}
    \newcommand{\eea}{\end{eqnarray}}
\newcommand{\beas}{\begin{eqnarray*}}
    \newcommand{\eeas}{\end{eqnarray*}}
\def\cDH{\cD(\cH)}

\newcommand{\ketbra}[2]{|#1\rangle\langle#2|}

\newcommand{\daniel}[1]{\color{orange}Daniel: #1\color{black}}

\begin{document}
\title{Eventually entanglement breaking Markovian dynamics: \\structure and characteristic times}
\author[1]{Eric P. Hanson}
\affil[1]{\small DAMTP, Centre for Mathematical Sciences, University of Cambridge, UK}

\author[2,4]{Cambyse Rouz\'{e}}
\affil[2]{\small Statistical Laboratory, Centre for Mathematical Sciences, University of Cambridge, UK}

\author[3]{Daniel Stilck Fran\c{c}a}

\affil[3]{QMATH, Department of Mathematical Sciences, University of Copenhagen, Universitetsparken 5, 2100 Copenhagen, Denmark}

\affil[4]{Department of Mathematics, Technische Universit\"at M\"unchen, 85748 Garching, Germany}

\date{February 21, 2019}

\maketitle

\begin{abstract}
    We investigate entanglement breaking times of Markovian evolutions in discrete and continuous time. In continuous time, we characterize which Markovian evolutions are \emph{eventually entanglement breaking}, that is, evolutions for which there is a finite time after which any entanglement initially present has been destroyed by the noisy evolution. In the discrete time framework, we consider the \emph{entanglement breaking index}, that is, the number of times a quantum channel has to be composed with itself before it becomes entanglement breaking. The PPT$^2$ conjecture is that every PPT quantum channel has an entanglement breaking index of at most 2; we prove that every faithful PPT quantum channel has a finite entanglement breaking index, and more generally, any faithful PPT CP map whose Hilbert-Schmidt adjoint is also faithful is eventually entanglement breaking. We also provide a method to obtain concrete bounds on this index for any faithful quantum channel. To obtain these estimates, we use a notion of \emph{robustness of separability} to obtain bounds on the radius of the largest separable ball around faithful product states. We also extend the framework
    of Poincar\'e inequalities for nonprimitive semigroups to the discrete setting to quantify the convergence of quantum semigroups in discrete time,  which is of independent interest.

\end{abstract}

\section{Introduction}\label{intro}

Distributing entangled quantum states using noisy quantum channels reliably and efficiently is one
of the fundamental challenges in quantum information theory, both from an experimental and theoretical point of  view.
Entanglement breaking channels, i.e. quantum channels that only output separable states when acting on one half
of a bipartite quantum state, are useless for such
non-classical communication protocols. If half of an entangled quantum state passes through a noisy channel several times, at what point does it lose its entanglement with the other half? In this work, we will prove several upper and lower bounds on this time in terms of properties of the channel; in other words, we establish bounds on the entanglement-breaking time of the channel.
This question arises naturally in the context of quantum repeaters~\cite{Bauml_2015,Christandl_2017} and such bounds limit their power to
implement non-classical communication protocols.

Mathematically, the situation can be modeled by means of continuous or discrete time \textit{quantum Markov semigroups} (QMS). The classification of the ability of QMS to preserve entanglement  both at finite time and asymptotically has recently received considerable attention.
The authors of~\cite{lami_entanglement-saving_2016} (see also \cite{kennedy2017composition})
took a more qualitative and asymptotic point of view, completely characterizing the class of discrete time QMS that do not become entanglement breaking asymptotically. Additionally, \cite{rahaman2018eventually} showed that certain classes of channels lead to eventually entanglement breaking QMS.

We further contribute to this classification by providing a simple characterization of the finite time and asymptotic entanglement behavior of continuous time QMS based on the structure of their \textit{decoherence-free} (DF) subalgebra, building upon the work of~\cite{lami_entanglement-saving_2016}. In particular, we show that the class of eventually entanglement breaking continuous-time QMS is precisely the one of primitive QMS, that is of those which possess a unique full-rank invariant state. On the other hand, continuous time QMS that break entanglement asymptotically are precisely those possessing a non-trivial, yet commutative DF subalgebra. Finally, QMS with a non-commutative DF subalgebra never break entanglement. We also manage to obtain quantitative lower and upper bounds on the entanglement-breaking time of a QMS: the former are obtained by exploring the connection between the spectrum of a quantum channel and whether its Choi matrix is of positive partial transpose (PPT). On the other hand, upper bounds are derived by exploiting the rapid-decoherence property of Markovian evolutions, which will provide us with a method to estimate how far the Choi matrix of the QMS at a given time is to its asymptotics. This property is usually obtained through functional inequalities for the underlying quantum channel and we relate our results to the most widely used in the literature, like the Poincar\' e inequality~\cite{TKRV10} or the logarithmic Sobolev inequality~\cite{kastotemme,bardet2017estimating,bardet2018hypercontractivity}. This combined with some knowledge about the geometry of the convex set of separable states \cite{gurvits2002} allows us to obtain quantitative bounds on the entanglement-breaking times.

The situation becomes more complicated in the discrete time setting. The asymptotic picture was developed by \cite{lami_entanglement-saving_2016}, where discrete time quantum Markov semigroups which preserve entanglement asymptotically were characterized. We continue this path by characterizing faithful quantum channels which are entanglement-breaking in finite time. To do so, we first study the ability of irreducible quantum channels~\cite{evans_spectral_1978,fagnola_irreducible_2009} to preserve entanglement. We show that there are irreducible quantum channels that only become entanglement breaking asymptotically and, in the case of quantum channels with an invariant state of full rank we obtain a converse, relating quantum channels that do not preserve entanglement asymptotically to irreducible channels. Moreover, again under the \textit{faithfulness} assumption, i.e. that of the existence of an invariant state of full rank, we show that quantum PPT channels become entanglement breaking after a finite number of iterations by exploring results on the structure of the spectrum of irreducible maps. We then extend this result to faithful completely positive maps by a suitable similarity transformation.
This generalizes the results of~\cite{rahaman2018eventually}, where the authors showed the statement for doubly stochastic channels.

However, in light of the PPT squared conjecture~\cite{Bauml_2015}, it is also desirable to obtain quantitative bounds to when a channel becomes entanglement breaking. So far, the conjecture was only proved for low-dimensional cases~\cite{Christandl2018,Chen2019} or for some particular families of quantum channels~\cite{Collins2018,kennedy2017composition}.
In~\cite{Christandl2018} the authors obtain upper bounds on the number of iterations in terms
of the Schmidt number of a channel. Here instead, we once again adopt an approach based on functional inequalities for discrete time QMS. Unfortunately, there is not a lot of work on functional inequalities for nonprimitive evolutions in discrete time~\cite{Szehr2015,Carbone2019}, which is particularly important in our setting. To start amending this gap in the literature and increase the class of examples our techniques apply to, we generalize the framework of Poincar\' e inequalities to discrete-time nonprimitive QMS. We  believe these techniques are of independent interest and will find applications elsewhere. Lower bounds are again found by use of the PPT criterion.

The lower and upper bounds we obtain with our techniques are tight up to constants for some classes of examples.
Moreover, we use similar techniques to consider the similar problems of when a
pair of QMS becomes entanglement annihilating~\cite{Morav_kov__2010}, how often one has to apply a doubly stochastic primitive quantum channel until it becomes mixed unitary and when the output of quantum Gibbs samplers are approximate quantum Markov networks.

\paragraph{Structure of the paper and contributions}
\begin{itemize}
    \item In \Cref{sec:cts}, we investigate the entanglement breaking properties of continuous time quantum Markov semigroups. We prove the set of primitive QMS coincides with the set of eventually entanglement breaking continuous time QMS in \Cref{propprimitivecontinuoustime}, and establish upper and lower bounds on the entanglement-breaking times in \Cref{tEB_upper} and \Cref{tEB_lower}. To establish the upper bounds, we use a notion of ``robustness of separability'' in \Cref{ballsseparable} and use strong decoherence bounds which we discuss in \Cref{strongdeco}. The lower bounds use the PPT criterion.
    \item In \Cref{sec3}, we investigate the structural properties of discrete time QMS. We identify classes of eventually entanglement breaking channels that are dense in the set of quantum channels in \Cref{thm:EEB_dense}, and discuss an application of this to the PPT$^2$ conjecture in \Cref{cor:PPT2-reduction-prim}. We discuss irreducible evolutions in \Cref{sec:irred}, prove they are asymptotically entanglement-breaking in the remarks following \Cref{prop:decomp-irred}, and establish necessary conditions for an irreducible map to be PPT, as well as sufficient conditions for an irreducible map to be entanglement-breaking, in \Cref{theoremEB}.  We then prove that PPT quantum channels with a full-rank invariant state are eventually entanglement-breaking in \Cref{cor:pptfulleeb}. We use these results to characterize faithful eventually entanglement breaking discrete time QMS in \Cref{prop:charEEB-discrete}. In \Cref{convergencedynamics}, we discuss the ``phase subspace'' of asymptotically entanglement-breaking channels, and in particular characterize irreducible quantum channels in terms of their phase subspace. In \Cref{sec:CPnonTP} we relate completely positive non-trace-preserving maps to quantum channels, and in \Cref{thm:CP_PPT_EB_discrete} establish weak conditions under which a completely positive PPT map is eventually entanglement breaking.
    \item In \Cref{sec:finite-time-discrete}, we establish finite-time properties of discrete time quantum Markov semigroups. In \Cref{strongmixing}, we investigate the so-called decoherence-free subalgebra of discrete-time evolution, and the contraction properties of the evolution with respect to this subalgebra, which is more subtle than in the continuous-time case. In \Cref{prop:discrete-Poincare}, we establish a discrete-time Poincar\'e inequality, which we use in \Cref{EBTdiscrete} to establish entanglement-breaking times, including for PPT channels with a full rank invariant state (\Cref{prop:PPT-EB_time}). We also provide a method to compute entanglement-breaking times for any faithful quantum channel in \Cref{rem:prim-eb-bound}. In \Cref{sec:tMU} we apply these techniques to establishing times for doubly-stochastic discrete time evolutions to become mixed unitary. We also discuss another application of continuous-time decoherence bounds, namely to approximate quantum Markov networks, in \Cref{sec:aQMN}.

\end{itemize}
\subsection{Notation, basic definitions and preliminaries}\label{sec:preliminaries}

\paragraph{States and norms}
Let $(\cH,\langle .|.\rangle)$ be a finite dimensional Hilbert space of dimension $d_\cH$.
We denote by $\cB(\cH)$ the Banach space of bounded operators on $\cH$, by $\cB_{\text{sa}}(\cH)$ the
subspace of self-adjoint operators on $\cH$, i.e. $\cB_{\text{sa}}(\cH)=\left\{X=\cB(\cH);\ X=X^\dagger\right\}$,
and by $\cB_{\text{sa}}^+(\cH)$ the cone of positive semidefinite operators on $\cH$,
where the adjoint of an operator $Y$ is written as $Y^\dagger$.
The identity operator on $\cH$ is denoted by $\mathbb{I}_\cH$, dropping the index $\cH$ when it is unnecessary.
Similarly, we will denote by $\id_{\cH}$, or simply $\id$, the identity superoperator on $\cB(\cH)$.
We denote by $\mathcal{D}(\cH)$ the set of positive semidefinite, trace one operators on $\cH$,
also called \emph{density operators}, and by $\cD_+(\cH)$ the subset of full-rank density operators. We denote by $|\Omega\rangle $ the maximally entangled state on $\cH\otimes \cH$: given any orthonormal basis $|i\rangle $ of $\cH$,
$$|\Omega\rangle :=\frac{1}{\sqrt{d_\cH}}\, \sum_{i=1}^{d_\cH} |i\rangle\otimes |i\rangle.$$
Given a bipartite system $\cH\otimes \mathcal{K}$, we denote by $\operatorname{SEP}\p \cH:\mathcal{K})$
the convex subset of separable states in $\cD\p\cH\otimes \mathcal{K})$.

For $p\ge 1$, the Schatten $p-$norm of an operator $A\in\cB(\cH)$ is denoted by $\|A\|_{p}:=(\tr|A|^p)^{\frac{1}{p}}$, where $|A|=\sqrt{A^\dagger A}$, and we denote by $\cT_p(\cH)$ the corresponding Schatten class. As usual, the operator norm is denoted by
$\|A\|_{\infty}$. We simply denote by $\|\Phi\|_{p\to q}$ the operator norm of a superoperator $\Phi$ from Schatten class $p$ to Schatten class $q$. Moreover, given a density matrix $\sigma\in\cD_+(\cH)$, we define the $\sigma-$weighted $p-$norm to be given by
$\|A\|_{p,\sigma}=\tr(|\sigma^\frac{1}{2p}A\sigma^\frac{1}{2p}|^p)^{\frac{1}{p}}$. In the case $p=2$, the $\|.\|_{2,\sigma}$ norm derives from an inner product $\langle A,\,B\rangle_\sigma:=\tr(\sigma^{\frac{1}{2}} A^\dagger\sigma^{\frac{1}{2}} B)$. The norms of superoperators between two such spaces are denoted by $\|\Phi\|_{p,\sigma\to q,\omega}$.
These norms provide a natural framework
to study the convergence of semigroups and we refer to e.g.~\cite{OLKIEWICZ1999246} for a review of some of their properties.

\paragraph{Quantum channels, Markovian evolutions and their spectrum}

Next, a \textit{quantum channel} is a completely positive, trace preserving map $\Phi:\cB(\cH) \to \cB(\cH)$. Given a linear map $\Phi:\cB(\cH)\to\cB(\cH)$, its  spectrum is denoted by $\spec(\Phi)$. We will exploit extensively the special structure of the spectrum of quantum channels and its connection to the semigroup's asymptotic behaviour, which we will now review in detail.
$\Phi^*$ corresponds to the dual map with respect to the Hilbert Schmidt inner product $\langle A,\,B\rangle_{\operatorname{HS}}:=\tr(A^\dagger B)$. Like any linear operator on $\cB(\cH)$, a quantum channel $ \Phi$ admits a Jordan decomposition:
\begin{equation}\label{Jordandecomp}
    \Phi=\Phi_P+\Phi_Q\,,~~~\Phi_P=\sum_{k:\,|\lambda_k|=1}\,\lambda_k P_k\,,~~~\Phi_Q=\sum_{k:\,|\lambda_k|<1}\lambda_k P_k+N_k\,,
\end{equation}
where $\lambda_k$ are the eigenvalues of $\Phi$, $P_k$ the associated (not necessarily orthogonal) eigenprojections, and $N_k^{d_k}=0$, where $d_k:=\tr(P_k)$, so that $\sum_{k=1}^K P_k=\mathbb{I}$. Note that we split the Jordan decomposition into two parts, $\Phi_P$ and $\Phi_Q$ and implicitly used the fact that all eigenvalues of quantum channels are contained in the unit disk of the complex plane. The operator $\Phi_P$ corresponds to the eigenvalues of modulus $1$, referred to as the \textit{peripheral spectrum}.
For any $k\in[K]$, $\lambda_kP_k+N_k$ constitutes the $k$-th Jordan block of $\Phi$. Here, we have used that the peripheral eigenvalues of a quantum channel are semi-simple, so there is no associated nilpotent part \cite{wolf2012quantum}.

In particular, since $\Phi$ is hermiticity preserving (and in particular positive), the eigenvalues of $\Phi$ either are real, or come in conjugate pairs. Since, moreover, $\Phi$ is positive unital ($\Phi(\mathbb{I})=\mathbb{I}$) or trace preserving ($\tr\Phi(A)=\tr(A)$ for all $A\in\cB(\cH)$), $ 1\in\spec(\Phi)$, all the other eigenvalues of $\Phi$ lie in the unit disc of the complex plane, and the eigenvalues lying on the peripheral spectrum are associated to one-dimensional Jordan blocks.

Given a quantum channel $\Phi:\cB(\cH)\to \cB(\cH)$, the sequence $\{\Phi^n\}_{n\in\NN}$ is called a \textit{discrete time quantum Markov semigroup} (discrete time QMS). Here, semigroup refers simply to the property that $\Phi^{n+m}=\Phi^n\circ \Phi^m$. Analogously, a \textit{continuous time quantum Markov process} (continuous time QMS) corresponds to a family $(\Phi_t)_{t\ge 0}$ of quantum channels $\Phi_t:\cB(\cH)\to\cB(\cH)$ that satisfies the following conditions: $\Phi_0 = \id$, $\Phi_{t+s} = \Phi_t\circ \Phi_s$ for any $s,t\in\R_0^+$, and $\Phi_t$ depends continuously on $t$. Any continuous time QMS can be written as $\Phi_t = e^{t\mathcal{L}}$ for a generator $\mathcal{L}:\cB(\cH) \to \cB(\cH)$.

We primarily consider discrete or continuous time QMS which are \emph{faithful}, that is, those that have a full-rank invariant state.
Let us recall basic ergodic properties of these evolutions.
The simplest case is that of \textit{primitive evolutions}. We call a quantum dynamical semigroup (or the Liouvillian generator) \textit{primitive} if it has a unique full rank fixed point $\sigma$. In this case for any initial state $\rho\in\cD(\cH)$
we have $\rho_t = e^{t\mathcal{L}}(\rho)\rightarrow \sigma$ as $t\rightarrow\infty$ (see \cite[Theorem 14]{ergodicchiribella}). We refer to \cite{ergodicchiribella,wolf2012quantum} for other characterizations of primitive channels and sufficient conditions for primitivity.

A notion closely related to primitivity is that of irreducibility. A positive linear map $\Phi:\cB(\cH)\to\cB(\cH)$ is said to be \textit{irreducible} if, for any orthogonal projection $P\in\cB(\cH)$, $\Phi(P\cB(\cH)P)\subset P\cB(\cH)P$ implies that $P=0$ or $P=\mathbb{I}$. For a positive trace-preserving map, this property is equivalent to the existence of a unique invariant state $\sigma > 0$ such that for every $\omega\in \cD(\cH)$, we have
\begin{equation} \label{eq:irred-erogdic-average}
    \lim_{N\to\infty}\frac{1}{N}\sum_{n=0}^{N-1} \Phi^n(\omega) = \sigma,
\end{equation}
by \cite[Cor. 6.3]{wolf2012quantum}. In the case of an irreducible, completely positive (in fact, for a Schwarz) map, it is known from Perron-Frobenius theory that the peripheral eigenvalues $\lambda_k$ in (\ref{Jordandecomp}) are non-degenerate and equal to $\phi^k$, where $\phi := \exp(2\iu \pi/z)$, for some fixed $z\le d_\cH$ (with $z=1$ if and only if the channel is primitive). Another useful criterion is that for a quantum channel, irreducibility is equivalent to $1$ being a non-degenerate eigenvalue with the corresponding eigenvector (being proportional to) a faithful quantum state. This makes irreducibility a good candidate for a property that implies that the quantum channel is EEB, although we will show that is generally not the case.

It is then natural to introduce the \textit{phase subspace} $\tilde{\cN}(\Phi)$  to study evolutions that are not necessarily primitive. It is defined as
\begin{align*}
    \tilde{\cN}(\Phi)=\text{span}\{X\in \cB(\cH):\exists \phi\in \R \text{ s.t. }\Phi(X)=e^{i\phi}X\},
\end{align*}
i.e. the linear span of the peripheral points and denote by $P$ the projection onto it.

More generally, given a \textit{faithful} continuous-time quantum Markov semigroup $(\Phi_t)_{t\ge 0}$ (i.e. for any $t\ge 0$, $\Phi_t(\sigma)=\sigma$ for a given full-rank state $\sigma$), there exists a completely positive, trace preserving map $E_\cN$, such that for any $\rho\in\cD(\cH)$, $\Phi_t(\rho)-\Phi_t\circ E_\cN(\rho)\to 0$ as $t\to\infty$ (see e.g. \cite{carbone_decoherence_2013}) and so that there exists a self-adjoint operator $H$ such that, $\Phi_t\circ E_\cN(\rho)=\e^{iHt}E_\cN(\rho)\e^{-iHt}$. In words, the evolution asymptotically behaves like a unitary evolution. On the other hand,  observables under the action adjoint semigroup $(\Phi_t^*)_{t\ge 0}$  converge towards a matrix subalgebra $\cN((\Phi_t^*)_{t\ge 0})$ of $\cB(\cH)$ of the following form
\begin{align}
    \cN((\Phi_t^*)_{t\ge 0}):=\bigoplus_{i\in\cJ}\,\cB(\cH_i)\otimes \mathbb{I}_{\cK_i}\,,\,\,\,\,\,\,\,\,\,\,\cH:=\bigoplus_{i\in\cJ}\,\cH_i\otimes\cK_i\,.
\end{align}
where the evolution is unitary. The completely positive map $E_\cN^*$ is a projection onto  $\cN((\Phi_t^*)_{t\ge 0})$. Moreover, $E_\cN^*$ is a \textit{conditional expectation onto $\cN((\Phi_t^*)_{t\ge 0})$}, meaning it is a positive linear unital map $E_{\cN}^*: \cB(\cH) \to \cN((\Phi_t^*)_{t\ge 0})$ such that $E_{\cN}^*(YXZ) = Y E_{\cN}^*(X) Z$ for all $Y,Z \in \cN((\Phi_t^*)_{t\ge 0})$ and $X \in \cB(\cH)$.
In the case of a non-necessarily faithful quantum channel $\Phi:\cB(\cH)\to \cB(\cH)$, the phase subspace is similarly known to possess the following structure (Theorem 6.16 of \cite{wolf2012quantum}, Theorem 8 of \cite{wolf2010inverse}): there exists a decomposition of $\cH$ as $\cH=\bigoplus_{j=1}^K\cH_j\otimes \cK_j\oplus\cK_0$ such that
\begin{align} \label{eq:decohere-decomp}
    \tilde{\cN}(\Phi):=\bigoplus_{i=1}^K\,\cB(\cH_i)\otimes \tau_i\oplus 0_{\cK_0}\,,~~~P(\rho\oplus 0_{\cK_0})=\sum_{i=1}^K\tr_{\cK_i}(p_i\rho p_i)\otimes \tau_i\,,
\end{align}
where $p_i$ is the orthogonal projector onto the $i$-th subspace, for some fixed states $\tau_i\in\cD_+(\cK_i)$. Moreover, there exist unitaries $U_i \in \cH_i$, and a permutation $\pi \in S_K$ which permutes within subsets of $\{1,\dotsc,K\}$ for which the corresponding $\cH_i$'s have equal dimension such that for any element $X\in \tilde \cN(\Phi)$, decomposed as $X=\bigoplus_{i=1}^K\,X_i \otimes \tau_i\oplus 0_{\cK_0}$ according to \eqref{eq:decohere-decomp}, we have
\begin{equation} \label{eq:Phi-on-decohere-decomp}
    \Phi(X) = \bigoplus_{i=1}^K\, U_i X_{\pi(i)} U_i^\dagger \otimes \tau_i\oplus 0_{\cK_0}.
\end{equation}
Note that the space $\cK_0\not=\{0\}$ if and only if the quantum channel is not faithful.
\paragraph{Entanglement loss}

A completely positive map $\Phi:\cB(\cH)\to\cB(\cH)$ is called \textit{entanglement breaking} if $\Phi\otimes\id_{\cH}(\rho)$ is a separable state for all input states $\rho\in\cD\p\cH\otimes \cH)$. This is  equivalent to the Choi matrix of the channel being separable~(see \cite{Horodecki2003} for the proof in the trace-preserving case; standard arguments can extend the result to the general case), where the Choi matrix of $\Phi$ is defined as
\begin{align}\label{choijamil}
    J(\Phi):=   d_\cH\,(\Phi\otimes\id_{\cH})(\ketbra{\Omega}{\Omega}).
\end{align}
The class of entanglement breaking channels on $\cH$ is denoted by $\operatorname{EB}(\cH)$. The map $\Phi\mapsto J(\Phi)$ is a linear isometry between the Hilbert space of linear operators on $\cB(\cH)$ equipped with the Schatten 2-norm and the Hilbert space of linear operators on $\cH\otimes \cH$, equipped also with the Schatten 2-norm:
\begin{equation}\label{eq:J-preserves-2norm}
    \|\Phi\|_2 := \sqrt{\tr \Phi^*\Phi}  = \| J(\Phi)\|_2 := \sqrt{\tr[ J(\Phi)^*J(\Phi)]},
\end{equation}
as the $2$-norm is invariant under coordinate permutations.
The map $\Phi:\cB(\cH) \to \cB(\cH)$ is of \textit{positive partial transpose} (PPT) if $(\cT\circ \Phi)\otimes\id_{\cH}$ is a positive operator, where $\cT$ is the partial transpose w.r.t. to some basis. The class of PPT channels on $\cH$ is called $\operatorname{PPT}(\cH)$. As the set of separable bipartite states is a subset of the set of bipartite states with a positive partial transpose, it follows that $\operatorname{EB}(\cH)\subset\operatorname{PPT}(\cH)$.

More generally, we also consider quantum Markovian evolutions on a bipartite Hilbert space $\cH_A\otimes \cH_B$. We call a bipartite quantum channel $\cT : \cB(\cH_{A}\otimes \cH_{B}) \to \cB(\cH_{A}\otimes \cH_{B})$
\emph{entanglement annihilating} if its output is separable (on any density matrix input).
Similarly, a quantum channel on a single-partite Hilbert space $\Phi: \cB(\cH)\to \cB(\cH)$ is called 2-\textit{locally entanglement annihilating} if $\Phi\otimes \Phi$ is entanglement annihilating. The class of bipartite, entanglement annihilating channels on $\cH_A\otimes\cH_B$ is denoted by $\operatorname{EA}(\cH_A,\cH_B)$, whereas the class of  $2$-locally entanglement annihilating channels on $\cH$ is denoted by $\operatorname{LEA}_2(\cH)$. Note that any quantum channel that is entanglement breaking is also $2$-locally entanglement annihilating: $\operatorname{EB}(\cH)\subset\operatorname{LEA}_2(\cH)$.

In this paper, we study the entanglement properties of quantum Markovian evolutions in discrete and continuous time. A discrete time QMS $\{\Phi^n\}_{n\in\NN}$, resp. a continuous time QMS $(\Phi_t)_{t\ge 0}$, is said to be \textit{eventually entanglement breaking} (EEB) if there exists $n_0\in\NN$, resp. $t_0\ge0$, such that for any $n\ge n_0$, resp. $t\ge t_0$, $\Phi^n$, resp. $\Phi_t$, is entanglement breaking. The class of eventually entanglement breaking Markovian evolutions is denoted by $\operatorname{EEB}(\cH)$, leaving the choice of framework of discrete or continuous time implicit. We also say a quantum channel $\Phi$ is eventually entanglement breaking if the discrete QMS $\{\Phi^n\}_{n=1}^\infty$ is eventually entanglement breaking. On the other hand, Markovian evolutions which are not entanglement breaking at any finite time are called \emph{entanglement saving}, using language introduced by Lami and Giovannetti \cite{lami_entanglement-saving_2016}; the class of entanglement saving Markovian evolutions is denoted by $\operatorname{ES}(\cH)$. Thus, the set of all Markovian evolutions (in either discrete time or in continuous time) decomposes into two disjoint classes,
\begin{equation} \label{eq:EEB-ES-decomp}
    \EEB(\cH) \sqcup \ES(\cH).
\end{equation}
Lami and Giovannetti also introduce the notion of \emph{asymptotically entanglement saving} evolutions in the discrete-time case. They showed that every discrete time QMS has at least one limit point, and either all of the limit points of a discrete time QMS $\{\Phi^n\}_{n=1}^\infty$ are entanglement breaking, or none of them are. They term the latter case as asymptotically entanglement saving, and we denote the set of asymptotically entanglement saving evolutions on $\cH$ as $\AES(\cH)$. In analogy, we call the former case by \emph{asymptotically entanglement breaking}, denoted $\AEB(\cH)$. Thus, the set of discrete time QMS on $\cH$ decomposes into the disjoint classes
\begin{equation}\label{eq:AES-AEB-decomp}
    \AES(\cH)\sqcup \AEB(\cH).
\end{equation}
It is interesting to compare \eqref{eq:EEB-ES-decomp} and \eqref{eq:AES-AEB-decomp}. A discrete time QMS $\{\Phi^n\}_{n=1}^\infty$ is AES if the limit points of the sequence $\{J(\Phi^n)\}_{n=1}^\infty$ are all entangled. Since $J(\Phi^{n+1}) = \Phi\otimes \id (J(\Phi^n))$, if $J(\Phi^{n+1})$ is entangled, $J(\Phi^n)$ must be as well. In particular, if $\{\Phi^n\}_{n=1}^\infty\in \AES(\cH)$, then $J(\Phi^n)$ is entangled for every $n$, and the discrete time QMS is entanglement saving. So we see $\AES(\cH)\subset \ES(\cH)$. However, a priori, an entanglement saving QMS could be asymptotically entanglement breaking: at any finite $n$, $J(\Phi^n)$ could be entangled, but in the limit, $J(\Phi^n)$ could be in the set of separable states (though necessarily on the boundary). We therefore define $\EB_\infty(\cH) = \AEB(\cH) \cap \ES(\cH)$, the set of discrete time QMS which are asymptotically entanglement breaking, but not entanglement breaking for any finite $n$. With this notion, we may relate \eqref{eq:EEB-ES-decomp} and \eqref{eq:AES-AEB-decomp}. We have the disjoint decomposition of the set of all discrete time QMS, denoted $\operatorname{dQMS}(\cH)$, satisfies
\begin{equation} \label{eq:QMC-full-decomp}
    \operatorname{dQMS}(\cH) = \UOLoverbrace{\EEB(\cH) \operatorname{\sqcup}}[\EB_\infty(\cH)]^{\AEB(\cH)}  \UOLunderbrace{ \sqcup \AES(\cH)}_{\ES(\cH)}.
\end{equation}
\begin{remark}
    The isomorphism between bipartite states and CP maps given by the Choi matrix and the equivalence between EB maps and separable states discussed before allows us to directly translate results on the entanglement loss of CP maps to statements about the separability of a bipartite state. We will mostly state our results in the picture of CPTP maps (see also \Cref{sec:CPnonTP} for a method to remove the TP assumption), but it should be straightforward to derive the corresponding statements for bipartite states.
\end{remark}

Some of the definitions introduced in this chapter, along with a preview of some of the results, are depicted diagrammatically in \Cref{fig:diagram}.

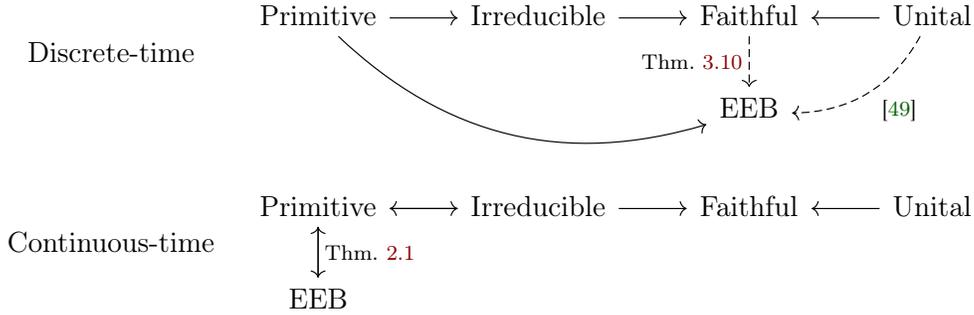
\begin{figure}
    \centering
    \begin{tabular}{cc}
        Discrete-time   & \begin{tikzcd}
            \text{Primitive} \rar \arrow[drr, bend right] & \text{Irreducible} \rar & \text{Faithful} \arrow[d, dashed, "\text{ \Cref{cor:pptfulleeb}}"'] & \lar \text{Unital} \arrow[ld, bend left, dashed, "\text{ \cite{rahaman2018eventually}}"]\\
            & & \text{EEB}
        \end{tikzcd} \\[1.5cm]
        Continuous-time &
        \begin{tikzcd}
            \text{Primitive} \rar \dar & \lar \text{Irreducible} \rar & \text{Faithful} & \lar \text{Unital} \\
            \text{EEB} \arrow[u, "\text{\Cref{propprimitivecontinuoustime}}"']
        \end{tikzcd}
    \end{tabular}
    \caption{Relations between classes of quantum Markov semigroups, in which arrows represent subsets; e.g., primitive discrete-time quantum Markov semigroups are a subset of irreducible discrete-time quantum Markov semigroups. The dashed arrows indicate relations which only hold for quantum Markov semigroups $\{\Phi^n\}_{n=1}^\infty$ associated to a PPT channel $\Phi$. The arrows without annotations are standard and are discussed in the text.}
    \label{fig:diagram}
\end{figure}

Let us now consider an example to illustrate these definitions and show that all sets arise naturally in physical systems.

\begin{example} \label{ex:RWA}
    Consider a discrete time quantum Markov semigroup $\{\Phi^n\}_{n\in\NN}$ associated to a \emph{repeated interaction system}, in which a system $\sys$ interacts with a chain of identical \emph{probes} $\env_k$, one at a time, for a duration $\tau$. During the interaction, the dynamics of the system are modeled by a Hamiltonian evolution, and at the times $(k \tau)_{k\geq 1}$, the evolution forms a semigroup. In this example, the system and each probe are 2-level systems, with associated Hilbert spaces $\cH_\sys = \cH_\env = \bC^2$. We define Hamiltonians $h_\sys = E a^*a$ and $h_\env = E_0 b^*b$, where $a/a^*$, resp. $b/b^*$, are the annihilation/creation operators for $\sys$, resp. $\env$, and $E\in \R_{>0}$ (resp. $E_0\in \R_{>0}$) corresponds to the energy of the excited state of $\sys$ (resp. $\env$). We can express these operators in the (ground state, excited state) basis of each system by
    \[
        a = b= \begin{pmatrix}
            0 & 1 \\ 0 & 0
        \end{pmatrix}, \qquad a^* = b^* = \begin{pmatrix}
            0 & 0 \\ 1 & 0
        \end{pmatrix}, \qquad a^*a = b^*b = \begin{pmatrix}
            0 & 0 \\ 0 & 1
        \end{pmatrix}_.
    \]
    We consider the initial state of each probe to be a \emph{thermal state},
    \[
        \xi_\beta  = \frac{\exp(-\beta h_\env)}{\tr[\exp(-\beta h_\env)]} = \begin{pmatrix}
            \frac{1}{1 + \e^{-\beta E_0}} & 0 \\ 0 & \frac{\e^{-\beta E_0}}{1+\e^{-\beta E_0}}
        \end{pmatrix}
    \]
    where $\beta \in [0,\infty]$ represents the inverse temperature (setting Boltzmann's constant to one). In the case of zero-temperature ($\beta=\infty$), we take
    \[
        \xi_\infty = \lim_{\beta\to\infty} \xi_\beta = \begin{pmatrix}
            1 & 0 \\ 0 & 0
        \end{pmatrix}.
    \]
    We consider an interaction modeling the two systems coupling through their dipoles, in the \emph{rotating wave approximation}. In this setting, the system and each probe interact via the potential $\lambda v_\text{RW} \in \cB(\cH_\sys \otimes \cH_\env)$, where $\lambda \geq 0$ is a coupling constant, and
    \[
        v_\text{RW} = \frac{u_1}{2}(a^* \otimes b + a \otimes b^*)
    \]
    where $u_1$ is a constant, which we take to be equal to $1$ with units of energy.
    This is a common approximation in the regime $|E - E_0 | \ll \min\{E,E_0\}$ and $\lambda \ll |E_0|$.

    The system begins in some state $\rho\init$, couples with the first probe (in thermal state $\xi_\beta$), and evolves for a time $\tau >0$ according to the unitary operator
    \[
        U := \exp(-i\tau (h_\sys \otimes \id + \id \otimes h_\env + \lambda v_\text{RW})).
    \]
    That is, $\rho\init\otimes \xi_\beta$ evolves to $U\,( \rho\init \otimes \xi_\beta )\,U^*$. Then we trace out the probe to obtain
    \[
        \rho_1 := \tr_\env( U\,( \rho\init \otimes \xi_\beta )\,U^*).
    \]
    This process is repeated, and at the end of the $k$th step, the system is in the state
    \[
        \rho_k =  \tr_\env( U \,(\rho_{k-1} \otimes \xi_\beta )\,U^*) = \Phi^k(\rho\init)
    \]
    where $\Phi$ is the quantum channel given by
    \[
        \Phi(\eta) := \tr_\env( U(\, \eta \otimes \xi_\beta\,) U^*) = \Phi^k(\rho\init).
    \]
    What class in the decomposition \eqref{eq:QMC-full-decomp} does the discrete time QMS $\{\Phi^n\}_{n\in\NN}$ lie in? To answer that question, we first compute the eigenvectors and eigenvalues of $\Phi$, yielding
    \[
        \Phi(\rho_{\beta^*}) = \rho_{\beta^*}, \qquad \Phi(a) = \gamma a, \qquad \Phi(a^*) = \bar \gamma a^*, \qquad \Phi(\sigma_z) = |\gamma|^2 \sigma_z
    \]
    where $\sigma_z = \begin{pmatrix}
            1 & 0 \\ 0 & -1
        \end{pmatrix}$ is the Pauli-$z$ matrix,
    \[
        \gamma := \e^{-\frac{1}{2}i \tau (E_0 + E)} \left(\cos \frac{\tau \nu}{2} + i \frac{(E_0 - E)}{\nu} \sin \frac{\tau\nu}{2}\right)
    \]
    for $\nu := \sqrt{ (E_0 -E)^2 + \lambda^2}$,
    and, defining the rescaled inverse temperature $\beta^* := \frac{E_0}{E}\beta$,
    \[
        \rho_{\beta^*} := \begin{pmatrix}
            \frac{1}{1 + \e^{-\beta^* E}} & 0 \\ 0 & \frac{\e^{-\beta^* E}}{1+\e^{-\beta^* E}}
        \end{pmatrix}=  \xi_\beta
    \]
    is a thermal state on $\cH_\sys$ represented by the same matrix as $\xi_\beta$. In the case $\beta=\infty$, we set $\rho_{\beta^*} := \begin{pmatrix}
            1 & 0 \\ 0 & 0
        \end{pmatrix}= \xi_\infty$. Note
    \[
        |\gamma| = \sqrt{ 1 - \frac{\lambda^2}{\nu^2}\sin^2 \left( \frac{\nu \tau}{2} \right) } \leq 1
    \]
    as required, since $\Phi$ is a CPTP map.
    In \cite[Prop. 23]{lami_entanglement-saving_2016} the authors show that a qubit channel is AES if and only if it is a unitary channel, which in turn is equivalent to it having determinant $1$.
    Thus, we have $\{\Phi^n\}_{n\in\NN} \in \AES(\cH)$ iff $|\gamma|=1$.

    To analyze whether or not $\Phi^n$ is entanglement-breaking, it suffices to check if $J(\Phi^n)$ is PPT, as $\Phi$ is a qubit channel. To that end, define the Gibbs factor $g = \exp(-\beta E_0)$ and  partition function $Z = 1 + g$. Then we have
    \[
        \rho_{\beta^*} +g Z\inv \sigma_z= \begin{pmatrix}
            1 & 0 \\ 0 & 0
        \end{pmatrix}, \qquad \rho_{\beta^*}  -  Z\inv \sigma_z = \begin{pmatrix}
            0 & 0 \\ 0 & 1
        \end{pmatrix}_.
    \]
    Thus, taking $(\ket{0},\ket{1})$ to be the (ground state, excited state) basis for each system,
    \begin{align*}
        (\id \otimes \mathcal{T}) J(\Phi^n) & = \sum_{i,j=0}^1 \Phi^n(\ket{i}\bra{j})\otimes \ket{j}\bra{i} \\
                                            & = \begin{pmatrix}
            \Phi^n(\rho_{\beta^*} + Z\inv \sigma_z) & \Phi^n(a^*)                                \\
            \Phi^n(a)                               & \Phi^n(\rho_{\beta^*}  - g Z\inv \sigma_z)
        \end{pmatrix}                                  \\
                                            & = \begin{pmatrix}
            \rho_{\beta^*}+|\gamma|^{2n} Z\inv\sigma_z) & \bar \gamma^n a^*                              \\
            \gamma^n a                                  & \rho_{\beta^*}- |\gamma|^{2n}g Z\inv \sigma_z)
        \end{pmatrix}                                  \\
                                            & = \begin{pmatrix}
            Z\inv( g+ |\gamma|^{2n}) & 0                          & 0                          & 0                        \\
            0                        & Z\inv ( 1 - |\gamma|^{2n}) & \bar \gamma^n              & 0                        \\
            0                        & \gamma^n                   & Z\inv g(1 - |\gamma|^{2n}) & 0                        \\
            0                        & 0                          & 0                          & Z\inv(1+g |\gamma|^{2n})
        \end{pmatrix}
    \end{align*}
    which has eigenvalues
    \[
        Z\inv( g+ |\gamma|^{2n}), \quad Z\inv(1+g |\gamma|^{2n}), \quad \frac{1 - |\gamma|^{2n}}{2} \pm \frac{1}{2}\sqrt{ (1 - |\gamma|^{2n})^2 \left( \frac{g - g\inv}{g + g\inv} \right)^2 + 4 |\gamma|^{2n} }
    \]
    for $g\in (0,1]$, and
    \[
        1, \quad 1, \quad \pm |\gamma|^{2n}
    \]
    for $g=0$.
    In either case, the eigenvalues only depend on the independent parameters $|\gamma| \in [0,1]$ and $g \in [0,1]$. Since $(\id \otimes \mathcal{T}) J(\Phi^n) \geq 0$ is equivalent to $\Phi^n \in \EB(\cH)$, we find
    \begin{itemize}
        \item If $|\gamma| = 1$, $\{\Phi^n\}_{n\in\NN} \in \AES(\cH)$. This occurs when $\nu\tau\in 2\pi\bZ$; in this case, $\Phi$ is a unitary channel.
        \item If $|\gamma| = 0$, then $\Phi\in \EB(\cH)$. This occurs when $E=E_0$, and $\nu\tau\in\pi + 2\pi \bZ$. In this case, $\Phi =\rho_{\beta^*}\tr[\cdot]$.
        \item If $|\gamma| \in (0,1)$ and $g=0$, $\{\Phi^n\}_{n\in\NN} \in \EB_\infty(\cH)$.  In this case, $\beta=\infty$ and $\Phi$ has a unique peripheral eigenvalue, namely $1$, with non-full-rank invariant state, $\ket{1}\bra{1}$.
        \item If $|\gamma|\in (0,1)$ and $g\in (0,1]$,  then $\{\Phi^n\}_{n=1}^\infty \in \EEB(\cH)$, and in particular, the minimal $n$ such that $\Phi^n \in \EB(\cH)$ is given by $n =\max(1, \ceil{ \frac{1}{2}\frac{\log B(g)}{\log|\gamma|}})$ where
              \[
                  B(g) := \frac{1 + 4 g^2 + g^4 - (1 + g^2) \sqrt{1 + 6 g^2 + g^4}}{2 g^2} \in (0, 3-2 \sqrt{2}] \quad \text{for}\quad g\in (0,1].
              \]
              In this case, $\Phi$ is primitive, with faithful invariant state $\rho_{\beta^*}$.
    \end{itemize}
\end{example}

\section{Continuous time quantum Markov semigroups} \label{sec:cts}

In this section, we exclusively investigate the simpler case of continuous time quantum Markov processes. It turns out that in this case, there is a simple characterization of $\EEB(\cH)$:

\begin{theorem}\label{propprimitivecontinuoustime}
    A continuous time quantum Markov semigroup $(\Phi_t)_{t\ge 0}$ is eventually entanglement breaking if and only if it is primitive.
\end{theorem}

In order to prove the above result, we first need the following straightforward extension of Theorem 1 of \cite{gurvits2002}.
\begin{lemma}\label{lemma1}
    Let $\sigma,\,\omega > 0$. Then  $\omega\otimes \sigma + \Delta$ is separable for any Hermitian operator $\Delta$ such that $\|\Delta\|_2\le \lambda_{\min}(\sigma)\,\lambda_{\min}(\omega)$, where $\lambda_{\min}(\sigma)$, resp. $\lambda_{\min}(\omega)$, stands for the smallest eigenvalue of $\sigma$, resp. $\omega$.
\end{lemma}
\begin{proof}
    Theorem 1 of \cite{gurvits2002} shows that if $\|\tilde \Delta\|_2 \leq 1$ then $I\otimes I + \tilde \Delta \in \SEP$. Since
    \[
    \omega\otimes \sigma + \Delta \in \SEP \iff I\otimes I + \omega^{-1/2}\otimes \sigma^{-1/2} \Delta \omega^{-1/2}\otimes \sigma^{-1/2} \in \SEP
    \]
    the proof is concluded by the fact that $\|\omega^{-1/2}\otimes \sigma^{-1/2} \,\Delta\, \omega^{-1/2}\otimes \sigma^{-1/2}\|_2 \leq \|\omega\inv\|_\infty \|\sigma\inv\|_\infty\|\Delta\|_2$.
\end{proof}

Now, we proceed to prove \Cref{propprimitivecontinuoustime}.

\begin{proof}~[\Cref{propprimitivecontinuoustime}]
    First assume that $(\Phi_t)_{t\ge 0}$ is primitive. Then there exists a full-rank state $\sigma\in\cD(\cH)$ such that, for any $\rho\in\cD(\cH)$, $\Phi_t(\rho)\to \sigma$ as $t\to\infty$. Therefore,
    \begin{align*}
        (\Phi_t\otimes\id)(|\Omega\rangle\langle\Omega|)\to \sigma\otimes d_\cH^{-1}\,\mathbb{I}\,,
    \end{align*}
    since $\Phi_t$ is trace-preserving for each $t$. The result follows by \Cref{lemma1}, which implies the existence of $t_0>0$ such that for any $t\ge t_0$, $(\Phi_t\otimes\id)(|\Omega\rangle\langle\Omega|)$ is in a nonempty ball around $\sigma\otimes d_{\cH}^{-1}\,\mathbb{I}$ consisting of separable states.
    Conversely, assume that $(\Phi_t)_{t\ge 0}\in\operatorname{EEB}(\cH)$, i.e. there exists $n\in\NN$ such that $\Phi_n=\Phi_1^n$ is entanglement breaking.
    Note that for any $t\ge 0$ the map $\Phi_t$ is invertible.
    In \cite[Theorem 11]{lami_entanglement-saving_2016}, the authors show that for invertible maps, being AES is equivalent to having more than one eigenvalue in the peripheral spectrum. As we assume $\Phi_t$ to be EEB, we conclude that $\Phi_t$ has a unique stationary state.
    In Proposition 7.5 of \cite{wolf2012quantum} this is shown to imply that the evolution is primitive, which concludes the proof.
    \qed
\end{proof}

\medskip

We note that the results of~\cite{Yoshida2018} show that imposing the stronger notion that the underlying Markovian semigroup is asymptotically decoupling, i.e. that all outputs are product states in the limit $t\to\infty$ is equivalent to the quantum channel being mixing, i.e. converges asymptotically to a unique, not necessarily faithful, quantum state. Their results also hold in discrete time.

Furthermore, for non-primitive continuous time quantum Markov semigroups, one can provide a characterization of $\AES(\cH)$ in a similar fashion as what is done in Theorem 24 of \cite{lami_entanglement-saving_2016} in the discrete time case.
We can then show:

\begin{proposition} \label{prop:cts-AES}
    Let $(\Phi_t)_{t\ge 0}$ be a faithful continuous time quantum Markov semigroup. Then $(\Phi_t)_{t\ge 0}\in \operatorname{AES}(\cH)$ if and only if its decoherence-free subalgebra $    \cN((\Phi_t^*)_{t\ge 0})$ is noncommutative, which means that there exists $i\in\cJ$ such that $d_{\cH_i}>1$.
\end{proposition}
\begin{proof}
    We simply need to show that $E_\cN\in\operatorname{EB}(\cH)$ if and only if $d_{\cH_i}=1$ for all $i$. If $d_{\cH_i}=1$ for all $i$,it follows from the decomposition of the decoherence free subalgebra given in Eq. \eqref{eq:decohere-decomp} that there exist quantum states $\sigma_i,\tau_i$ s.t. for all $\rho\in\cDH{}$ it holds that:
    \begin{align*}
        \lim_{t\to\infty}\Phi_t(\rho)=\sum\limits_{i}\tr\lb\sigma_i \rho\rb\tau_i.
    \end{align*}
    In \cite{HPR03} the authors show that such "measure and prepare" quantum channels are entanglement breaking.
    If now there exists $i$ such that $d_{\cH_i}>1$, choose as input state
    $\rho=|\psi\rangle\langle\psi|\otimes \tau_i\in\mathcal{D}\left(\mathcal{H}_i\otimes\mathcal{H}\right) $, where $|\psi\rangle=(d_{\cH_i})^{-1/2}\sum_{j=1}^{d_{\cH_i}}|j\rangle|j\rangle$ is the maximally entangled state on $\cH_i\otimes \cH_i$, and the result follows from the fact that $\id_{\cH_i}\otimes E_\cN(\rho)$ is entangled.
    \qed
\end{proof}

\begin{corollary}
    Let $(\Phi_t)_{t\geq 0}$ be a faithful continuous time quantum Markov semigroup. Then $(\Phi_t)_{t\geq 0}\in \EB_\infty(\cH)$ if and only if its decoherence-free subalgebra $\cN((\Phi_t^*)_{t\ge 0})$ is commutative and non-trivial, meaning  $d_{\cH_i}=1$ for all $i\in \cJ$, and $|\cJ|>1$.
\end{corollary}
\begin{remark}
    The dephasing semigroup provides a simple example of an element in $\EB_\infty(\cH)$.
\end{remark}
\begin{proof}
    Since \Cref{prop:cts-AES} characterizes when $(\Phi_t)_{t\geq 0} \in \AEB$, it remains to exclude eventually entanglement breaking maps. However, the semigroup is eventually entanglement breaking if and only if it is primitive, by \Cref{propprimitivecontinuoustime}, and is primitive if and only if $|\cJ|=1$, which completes the proof.
\end{proof}

\medskip

Proposition \ref{propprimitivecontinuoustime} justifies the introduction of the following characteristic times in the continuous time primitive case: let $(\Phi_t)_{t\ge 0}$ be a primitive continuous time quantum Markov semigroup with invariant state $\sigma\in\cD_+\p \cH)$. The entanglement breaking time $t_{\EB}(\Phi)$ of $(\Phi_t)_{t\ge 0}$, analogous to the entanglement breaking index in~\cite{Lami2015eb}, is defined as follows:
\begin{equation*}
    t_{\EB}\p (\Phi_t)_{t\ge 0})\overset{\operatorname{def}}{=}\inf\left\{t\ge 0:\Phi_t \in \EB\p\cH)\right\}\,.
\end{equation*}
Similarly, given a continuous time quantum Markov semigroup $(\Gamma_t)_{t\ge 0}$ over a bipartite Hilbert space $\cH_A\otimes\cH_B$, we define the \textit{entanglement annihilation time} $t_{\text{EA}}(\Gamma)$ as follows
\begin{align*}
    t_{\text{EA}} ((\Gamma_t)_{t\ge 0})\overset{\operatorname{def}}{=}\inf\left\{t\ge 0: \Gamma_t\in\operatorname{EA}(\cH_A,\cH_B)\right\}.
\end{align*}
In the case when $\Gamma_t=\Phi_t\otimes\Phi_t$, for $\Phi_t:\,\cB(\cH)\to\cB(\cH)$, this time is called the $2$-\textit{local entanglement annihilation time}, and is denoted by
\begin{align*}
    t_{\text{LEA}_2}((\Phi_t)_{t\ge 0})\overset{\operatorname{def}}{=}\inf\left\{t\ge 0:\Phi_t\in\operatorname{LEA}_2(\cH)\right\}\,.
\end{align*}
The quantity $t_{\text{LEA}_2}$ can be seen as a quantitative version of the notion of asymptotic decoupling for Markovian quantum dynamics~\cite{Yoshida2018}.
Entanglement breaking, entanglement annihilation, and $2$-local entanglement annihilation times of quantum Markov semigroups in discrete time can be similarly defined.
In the next two subsections, we provide bounds on $t_{\text{EB}}$, $t_{\text{EA}}$ and $t_{\text{LEA}_2}$: the upper bounds found in \Cref{upperbounds} use the strong decoherence property of Markovian evolutions together with estimates on the radius of open balls around any full-rank product state. On the other hand, lower bounds found in
\Cref{lowerbounds} mainly use the inclusion $\operatorname{EB}(\cH)\subset\operatorname{PPT}(\cH)$.

\subsection{Upper bounds on entanglement loss via decoherence}\label{upperbounds}

First, we briefly review in \Cref{strongdeco} the notion of strong decoherence of a quantum Markovian evolution which leads to the derivation of bounds on the time it takes for any state evolving according to a continuous time  quantum Markov semigroup to come $\eps$-close to equilibrium. As a second step, in \Cref{ballsseparable}, we get quantitative bounds on the radius of balls surrounding any full-rank separable state on a bipartite Hilbert space $\cH_A\otimes \cH_B$. Upper bounds on entanglement loss times follow by simply choosing $\eps$ as the radius of the separable ball around the adequate state found in \Cref{ballsseparable}. This is done in \Cref{mergethings}.

\subsubsection{Strong decoherence}\label{strongdeco}

At the beginning of this section, we briefly recalled the convergence of quantum Markovian evolutions towards their decoherence-free subalgebra. Moreover, any finite dimensional, faithful continuous time quantum Markov semigroup $(\Phi_t)_{t\ge 0}$ satisfies the so-called \textit{strong decoherence property} (SD): there exist constants $K,\gamma>0$, possibly depending on $d_\cH$, such that for any initial state $\rho$\footnote{In the discrete time case, we define the same property by simply replacing the left hand side by $\| \Phi^n(\rho-P(\rho))\|_1$, where $P$ corresponds to the projection onto the phase subspace, and the right hand side by $K\gamma^{-n}$.}:
\begin{align}\label{SDP}\tag{$\operatorname{SD}$}
    \|\Phi_t(\rho-E_\cN(\rho))\|_1\le \,K\,\e^{-\gamma t}\,.
\end{align}

In the primitive case, good control over the constants $K$ and $\gamma$ can be achieved from so-called \textit{functional inequalities} (see e.g. \cite{OLKIEWICZ1999246,TKRV10,kastotemme}). These techniques were recently extended to the non-primitive case in \cite{bardet2017estimating,bardet2018hypercontractivity}. Some of them have also been adapted to the discrete case (see \cite{MullerHermes2018sandwichedrenyi,MHF16,TKRV10}). In this section, we briefly review these tools that we use in the next section in order to derive upper bounds on the various entanglement loss times previously defined.

\paragraph{Poincar\' e inequality:} Perhaps the simplest functional inequality is the \textit{Poincar\'{e} inequality} (or spectral gap inequality): In the case of a continuous time quantum Markov semigroup $(\Phi_t)_{t\ge 0}$ with associated generator $\cL$, its Poincar\'{e} constant is defined as \cite{bardet2017estimating}:
\begin{align}\label{spectralgap}
    \lambda(\cL^*):=\inf_{X\in\cB_{sa}(\cH)}\frac{-\langle X,\,\cL^*(X)\rangle_{\sigma_{\tr}}}{\|X-E^*_\cN(X)\|^2_{2,\sigma_{\tr}}}\,,
\end{align}
where $\sigma_{\tr}:=d_\cH^{-1}\,E_\cN(\mathbb{I})$, and $\cL^*$ is the generator of the dual semigroup $(\Phi_t^*)_{t\ge 0}$ acting on observables. The Poincar\'{e} constant turns out to be the spectral gap of the operator $\frac{\cL^*+\hat{\cL}}{2}$, where $\hat{\cL}$ is the adjoint of $\cL^*$ with respect to $\langle .,.\rangle_{\sigma_{\tr}}$, namely minus its second largest distinct eigenvalue. Moreover,
\begin{align}\label{poincareDF}
    \|\hat \Phi_t(X- \hat E_\cN(X))\|_{2,\sigma_{\tr}}\le \e^{-\lambda(\cL^*) t}\,  \|X-\hat E_\cN(X)\|_{2,\sigma_{\tr}}\,.
\end{align}
Strong decoherence in the form of \Cref{SDP} with $K={\|\sigma_{\tr}^{-1}\|_\infty}^{1/2}$ and $\gamma=\lambda$ follows from \Cref{poincareDF}, since $\|\Phi_t(\rho-E_\cN(\rho))\|_1\le \|\hat{\Phi}_t(X-\hat E_\cN(X))\|_{2,\sigma_{\tr}} $ \cite{Ruskai94}, where $X=\sigma_{\tr}^{-1/2}\rho\sigma_{\tr}^{-1/2}$, $\|X-\hat E_\cN(X)\|_{2,\sigma_{\tr}}\le {\|\sigma_{\tr}^{-1}\|_\infty}^{1/2}$, $\hat{\Phi}_t$ is the adjoint of $\Phi_t^*$ with respect to the inner product $\langle .,. \rangle_{\sigma_{\tr}}$, and $\hat E_\cN$ is the conditional expectation onto the decoherence-free subalgebra of $\hat \Phi_t$.

\paragraph{Modified logarithmic Sobolev inequality:} The prefactor $K=\|\sigma_{\tr}^{-1}\|_\infty$ obtained from the Poincar\'{e} method is known to be suboptimal in some situations. A stronger inequality that one can hope for is the so-called \textit{modified logarithmic Sobolev inequality} (MLSI): given a continuous time quantum Markov semigroup $(\Phi_t)_{t\ge 0}$, its associated modified logarithmic Sobolev constant $\alpha_1(\cL)$ is defined as follows \cite{bardet2017estimating}:
\begin{align*}
    \alpha_1(\cL):=\inf_{\rho\in\cD(\cH)}\frac{\operatorname{EP}_{\cL}(\rho)}{D(\rho\| E_\cN(\rho))}\,,
\end{align*}
where $\operatorname{EP}_\cL(\rho):=-\tr(\cL(\rho)(\ln\rho-\ln\sigma))$ is the so-called \textit{entropy production} of the semigroup $(\Phi_t)_{t\ge 0}$. Its name is justified by the fact that $\frac{d}{dt}D(\Phi_t(\rho)\|E_\cN(\rho))=-\operatorname{EP}_\cL(\rho)$. This directly leads to the following exponential decay in relative entropy:
\begin{align}
    D(\Phi_t(\rho)\|\Phi_t\circ E_\cN(\rho))\le \e^{-\alpha_1(\cL)t}D(\rho\|E_\cN(\rho))\,.
\end{align}
This together with the bound $D(\rho\|E_\cN(\rho))\le \ln\|\sigma_{\tr}^{-1}\|_\infty$ and the quantum Pinsker inequality implies \reff{SDP} with $K=(2\ln (\|\sigma_{\tr}^{-1}\|_\infty))^{1/2}$ and $\gamma=\alpha_1(\cL)$. In particular, this new prefactor $K$ constitutes a considerable improvement over the one derived from the Poincar\'{e} method. Similar decoherence times were recently obtained via \textit{decoherence free hypercontractivity} in \cite{bardet2018hypercontractivity} and other similar techniques have recently been developed to obtain convergence bounds similar to the ones obtained by MLSI inequalities~\cite{gao2018fisher,Bardet2019}.

\paragraph{Estimates in diamond norm} By Proposition \ref{propprimitivecontinuoustime}, we know that in continuous time, it only makes sense to talk about entanglement-breaking times for primitive  quantum Markov semigroups. As we will see in the next subsection, these times can be derived from the strong decoherence property of a primitive semigroup $(\Phi_t)_{t\ge 0}$ on $\cB(\cH)$, of invariant state $\sigma$, when tensorized with the identity map on a reference system $\cH_R\simeq \cH$. The resulting semigroup $(\Phi_t\otimes \id_{\cH_R})_{t\ge 0}$ is faithful and non-primitive, and its associated conditional expectation $E_\cN$ takes the form $\sigma\otimes \tr_{\cH}(\cdot)$. In this case, for any $\rho\in\cD(\cH\otimes\cH_R)$:
\begin{equation}\label{eq:diamond-norm-bound}
    \|\Phi_t\otimes \id(\rho-\sigma\otimes \tr_{\cH_R}(\rho))\|_1\le \| \Phi_t-\sigma \tr(\cdot)\|_\diamond \le d_\cH^{1/2}\|   \Phi_t-\sigma \tr(\cdot)\|_{1\to 1}\,,
\end{equation}
where $\|\cdot\|_{\diamond}$ denotes the diamond norm. Hence, the strong decoherence bound on $(\Phi_t\otimes \id_{\cH_R})_{t\ge 0}$ can be simply derived from a strong mixing bound for the primitive evolution  $(\Phi_t)_{t\ge 0}$ at the cost of a multiplicative factor $d_\cH^{1/2}$. The same factor would appear from the spectral gap estimate, since in that case $\sigma_{\tr}=\sigma\otimes \frac{\mathbb{I}}{d_{\cH}}$, so that $K=\|\sigma^{-1}\|_\infty^{1/2}d_\cH^{1/2}$.

However, in finite dimensions, it was shown in \cite{bardet2017estimating} that the modified logarithmic Sobolev constant of $(\Phi_t\otimes \id_{\cH_R})_{t\ge 0}$ is positive and in some situations can be of the same order as the one of $(\Phi_t)_{t\ge 0}$ (see also \cite{gao2018fisher}). In this case, $K=(2\ln( d_\cH \|\sigma^{-1}\|_\infty))^{1/2}$ constitutes an improvement over the constant $K'=d_\cH^{1/2}(2\ln( \|\sigma^{-1}\|_\infty))^{1/2}$ that one would get from the diamond norm estimate, after using the strong mixing bound provided by MLSI for $(\Phi_t)_{t\ge 0}$. Other related forms of convergence measures for $(\Phi_t\otimes \id_{\cH_R})_{t\ge 0}$ have also been investigated in~\cite{Gao2019}.

\subsubsection{Separable balls around separable states}\label{ballsseparable}

In this subsection, we consider the problem of finding separable balls around separable states. We restrict our discussion to full rank separable states, as separable states that are not faithful lie on the boundary of the set of separable states \cite{lami_entanglement-saving_2016}.
Note, however, that there are faithful separable states that lie on the boundary of the set of separable states and the boundary
of the set of separable quantum states is still a subject of active research~\cite{2015RSPSA.47150102C,2018arXiv180706223K}.
Thus, one way of quantifying how much the state lies in the interior of the set of separable states is the following measure of ``robustness of separability'', inspired by
the robustness of entanglement introduced in~\cite{vidal1999robustness} (see also Chapter 9 of \cite{aubrun2017}).
\begin{definition}[Robustness of separability]
    Let $\rho_{AB}>0$ be separable on the bipartite Hilbert space $\cH\equiv \cH_{AB}$. We define its robustness of separability w.r.t. the maximally mixed state, $R(\rho_{AB})$, as
    \begin{align*}
        R(\rho_{AB}) = \sup\, \left\{ \lambda \in [0,1] :  \exists\, \rho'_{AB} \text{ separable s.t. }\rho_{AB} = \lambda\, \frac{\mathbb{I}_A\otimes \mathbb{I}_B}{d_\cH} + (1-\lambda)\,\rho'_{AB}\right\}.
    \end{align*}
\end{definition}

\begin{proposition}[Properties of the robustness of separability]\label{robustnesssep}
    Let $\rho_{AB}\in \SEP(A:B)$. Then we have the following properties.
    \begin{enumerate}
        \item We have the bounds
              \begin{equation} \label{eq:R-UB}
                  0 \leq R(\rho_{AB}) \leq d_\cH\,\lambda_{\min}(\rho_{AB})\,,
              \end{equation}
              and equality holds in the second inequality for product states:
              \begin{equation} \label{eq:R-product}
                  R(\rho_A\otimes \rho_B) = d_\cH\,\lambda_{\min}(\rho_A)\,\lambda_{\min}(\rho_B)\,.
              \end{equation}
        \item $\rho_{AB} \in \interior \SEP(A:B)$ if and only if $R(\rho_{AB}) > 0$. Moreover, any state $\sigma_{AB}$ such that
              \begin{equation} \label{eq:robusteness_of_sep_bound}
                  \|\rho_{AB}-\sigma_{AB}\|_2 \leq \frac{R(\rho_{AB})}{d_\cH}
              \end{equation}
              is separable.
        \item  $R$ is quasi-concave on $\SEP(A:B)$: if $\rho_{AB} = \sum_j p_j \sigma^{(j)}_{AB}$ where $p_j>0$, $\sum_j p_j = 1$, then
              \[
                  R(\rho_{AB}) \geq \max_j R(\sigma^{(j)}_{AB}).
              \]
    \end{enumerate}

\end{proposition}
\begin{proof}
    The upper bound in \eqref{eq:R-UB} follows from the fact that $\rho_{AB} - \lambda \frac{\one_A\,\otimes\,\one_B}{d_\cH} \not \geq 0$ for $\lambda > d_Ad_B\lambda_{\min}(\rho_{AB})$. For product states, we may explicitly evaluate $R(\rho_A\otimes \rho_B)$ using the reformulation
    \[
        R(\rho_{AB}) = \sup \left\{ \lambda\in [0,1] :  \frac{1}{1-\lambda} \left(\rho_{AB} - \lambda \frac{\one_A\otimes\one_B}{d_Ad_B}\right)  \in \SEP(A:B)\right\}
    \]
    and by expanding $\rho_A$ in its eigenbasis, $\one_A$ in the same basis, $\rho_B$ in its eigenbasis, and $\one_B$ in the same basis.

    For the second point, we first note that if $R(\rho_{AB}) = 0$, then for any  $\lambda\in(0,1)$,
    \[
        \frac{1}{1-\lambda} \left(\rho_{AB} - \lambda \frac{\one_A\otimes\one_B}{d_Ad_B}\right) \not \in \SEP(A:B).
    \]
    This quantity is in the affine hull of $\SEP$, and can be made arbitrarily close to $\rho_{AB}$ by taking $\lambda$ small, which proves $\rho_{AB}$ is not in the relative interior of $\SEP$. The other implication follows from the bound \eqref{eq:robusteness_of_sep_bound}, which is proven as follows. By the definition of $R$ and the closedness of $\SEP(A:B)$, we may write
    \begin{equation}
        \rho_{AB}=R(\rho_{AB})\, \frac{\one_A\otimes\one_B}{d_Ad_B}+(1-R(\rho_{AB}))\rho_{AB}'
    \end{equation}
    for some \(\rho_{AB}'\) separable.
    Now consider another state \(\sigma_{AB}\) st. \(\|\rho_{AB}-\sigma_{AB}\|_2\leq R(\rho_{AB})d_\cH^{-1}\).
    Then
    \begin{align*}
        \sigma_{AB}=\rho_{AB}+(\sigma_{AB}-\rho_{AB})=R(\rho_{AB})\lb\frac{\one_A\otimes\one_B}{d_Ad_B}+\frac{1}{R(\rho_{AB})}\lb\rho_{AB}-\sigma_{AB}\rb\rb+(1-R(\rho_{AB}))\rho_{AB}'.
    \end{align*}
    As $\|\rho_{AB}-\sigma_{AB}\|_2 \leq \frac{R(\rho_{AB})}{d_{\cH}}$,
    \begin{align*}
        \frac{\one_A\otimes\one_B}{d_Ad_B}+\frac{1}{R(\rho_{AB})}\lb\rho_{AB}-\sigma_{AB}\rb
    \end{align*}
    is a separable state (cf. Theorem 1 of \cite{gurvits2002}), from which it follows that $\sigma_{AB}$ is separable as well, as a convex combination of separable states.

    For the last point, it suffices to prove the statement in the case of two states. If $\lambda = \min(R(\rho_{AB}), R(\sigma_{AB}))$, then
    \[
        \frac{1}{1-\lambda} \left(\rho_{AB} - \lambda \frac{\one_A\otimes\one_B}{d_Ad_B}\right) \in \SEP(A:B), \quad \text{and}\quad \frac{1}{1-\lambda} \left(\sigma_{AB} - \lambda \frac{\one_A\otimes\one_B}{d_Ad_B}\right) \in \SEP(A:B)
    \]
    and so for any $t\in [0,1]$,
    \[
        \frac{1}{1-\lambda} \left(t \rho_{AB} + (1-t)\sigma_{AB} - \lambda \frac{\one_A\otimes\one_B}{d_Ad_B}\right) \in \SEP(A:B)
    \]
    too, so $R(t \rho_{AB} + (1-t)\sigma_{AB}) \geq \lambda = \min(R(\rho_{AB}),R(\sigma_{AB}))$ as desired. \qed
\end{proof}
\begin{remark}
Together, \eqref{eq:robusteness_of_sep_bound} and \eqref{eq:R-product} recover \Cref{lemma1} in the case of density matrices. Note Proposition 3 of \cite{lami_entanglement-saving_2016} shows that $\rho_A\otimes\rho_B \in \interior \SEP(A:B)$ when $\rho_A$ and $\rho_B$ are full rank. \Cref{lemma1} strengthens this result by giving a quantitative bound:
    \[
        B_2(\lambda_{\min}(\rho_A)\,\lambda_{\min}(\rho_B), \rho_A\otimes\rho_B) \subset \SEP(A:B)
    \]
    where $B_p(r,\rho_{AB})$ is the closed ball in $p$-norm of radius $r$ around $\rho_{AB}$.

    Admittedly, it is not a priori clear how to obtain good lower bounds on the robustness of separability for general separable states $\rho_{AB}$ and we leave this for future work.
\end{remark}

\begin{remark}[Separable balls in relative entropy]
    In the previous section we commented on how to obtain better convergence estimates based on showing the exponential decay of other distance measures, like the relative entropy. It is natural to wonder if
    obtaining separable balls in the relative entropy and applying the convergence bound directly might lead to better results. Unfortunately, to the
    best of our knowledge, the radius of separable balls in the relative entropy is not known. However, such results would only
    lead to a constant improvement on the bounds for $t_{\text{EB}}$ we obtain with our methods. To see why this is the case, note for $\rho \in \cD$ and $\sigma \in \cD_+$,
    \begin{align*}
        \frac{1}{2}\|\rho-\sigma\|_1^2 \leq D\lb \rho\|\sigma\rb\leq \|\rho-\sigma\|_1\|\sigma^{-1}\|_\infty
    \end{align*}
    where the first inequality is Pinsker's inequality and second inequality  follows simply from~\cite[Theorem 1]{Audenaert_2011}.
    Thus, the radius of the largest ball in trace distance and
    relative entropy can only differ by square-roots and a factor which is polynomial in $\|\sigma^{-1}\|_\infty$.
    As we have a logarithmic dependence on the radius and on $\|\sigma^{-1}\|_\infty$ in our bounds (see e.g. \Cref{tEB_upper}), we could only hope to improve our bounds by a constant factor if we were able to derive optimal balls in relative entropy.

\end{remark}

\subsubsection{Upper bounds}\label{mergethings}

Here, we combine the tools gathered in the last two subsections, namely estimates on the radius of balls surrounding tensor product states, as well as the strong decoherence property, in order to estimate from above the entanglement loss in the different situations defined at the beginning of this section.

\begin{proposition}\label{tEB_upper}
    Let $\p\Phi_t)_{t\ge 0}$ be primitive with  full-rank invariant state $\sigma$ and generator $\cL$. Assuming that $\operatorname{(}$\ref{SDP}$\operatorname{)}$ holds for $(\Phi_t \otimes \id)_{t\geq 0}$: $$t_{\EB}\p (\Phi_t)_{t\ge 0})\le\frac{\log\,\p \,K\,d_\cH\|\sigma^{-1}\|_{\infty})}{\gamma}\,.$$

\end{proposition}
\begin{remark}
    Recall that strong decoherence always holds for some $K,\lambda>0$ in the finite-dimensional setting considered here.
\end{remark}
\begin{proof}
    By \Cref{lemma1} and the fact that $1$-norm upper bounds the $2$ norm, we know that the $\|\cdot\|_1$-norm around $\mathbb{I}/d_\cH\otimes\sigma$ of radius $\lambda_{\min}(\sigma)/d_\cH$ is included in $\SEP\p \cH,\cH)$. In the primitive case, $(\Phi_t\otimes \id)\circ E_{\cN}(\rho) =  \sigma \otimes \tr_A(\rho)$. It is therefore clear that for any $t$ such that $K\,\e^{-\lambda(\cL)\,t}\leq\lambda_{\min}(\sigma)/d_\cH$, $\id\otimes \Phi_t\p |\Omega\rangle\langle\Omega|)$ is separable, which implies that the channel $\Phi_t$ itself is entanglement breaking.
    \qed

\end{proof}

The proof of \Cref{tEB_upper2} can be adapted to get upper bounds on the entanglement annihilating time of a tensor product of semigroups. For example:

\begin{theorem}[Upper bound on $t_{\operatorname{EA}}$] \label{tEA_UB_QDS}
    Let $(\Phi_t)_{t\ge 0}$  (resp. $(\Psi_t)_{t\ge 0}$) be primitive and reversible on $\cB(\cH_A)$ with respect to the full-rank state $\sigma$ (resp. in $\cB(\cH_B)$ with respect to $\omega$). Assume that the spectral gaps are both lower bounded by $\lambda>0$. Then,
    \begin{align}\label{equ:eafromspgap}
        t_{\operatorname{EA}}((\Phi_t \otimes \Psi_t)_{t\ge 0})\leq\frac{3}{2}\,\frac{\log\lb\|\sigma^{-1}\|_\infty\|\omega^{-1}\|_\infty\rb}{\lambda}\,.
    \end{align}
\end{theorem}

\begin{proof}
    Since $\lambda$ is a lower bound on the spectral gap of $(\Phi_t\otimes\Psi_t)_{t\ge 0}$, it follows from equation~\eqref{eq:robusteness_of_sep_bound} that choosing $t$ as in the statement is enough to ensure that
    all the outputs of the semigroup are contained in the separable ball around $\sigma\otimes\omega$, which implies that $\Phi_t\otimes\Psi_t$
    is in $\operatorname{EA}(\cH_A,\cH_B)$.
    \qed
\end{proof}

\subsection{Lower bounds via the PPT criterion}\label{lowerbounds}

Here, we derive a lower bound on the time it takes a Markov semigroup $(\Phi_t)_{t\ge 0}$ to become entanglement breaking based on spectral data.
The idea is simply to use the useful fact that the set of PPT states includes the set of separable states.
We recall that a state $\rho\in\cD(\cH_A\otimes \cH_B)$ is said to have a
\textit{positive partial transpose} (PPT) if the operator $\id\otimes \mathcal{T}(\rho)$ is positive,
where the superoperator $\mathcal{T}$ denotes the transposition with respect to any basis (see Proposition 2.11 of \cite{aubrun2017}).
We will prove lower bounds for the semigroup to become $2$-locally entanglement annihilating, but note that also implies that it is not entanglement breaking.
\subsubsection*{Sufficient conditions for entanglement loss}

In the next lemma, given a channel $\Phi$ we find necessary conditions on $k$ for $\Phi^k$ to be $2$-locally entanglement annihilating.

\begin{lemma} \label{lem:1-shot-LEA2}
    Let $\Phi:\cB(\cH) \to \cB(\cH)$ be a quantum channel. If
    \begin{equation} \label{eq:2-norm-LEA}
        \|\Phi\|_2 > \sqrt{d_\cH}
    \end{equation}
    then $(\Phi\circ\cT)\otimes \Phi$ is not a positive map. In particular, $\Phi\otimes \Phi\notin \operatorname{LEA}_2(\cH)$, and hence $\Phi\not \in \EB(\cH)$.
\end{lemma}
\begin{remark}
    This is similar to the prior result that for a quantum channel $\Phi$,
    \begin{equation} \label{eq:1norm_not_EB}
        \|\Phi\|_1 > d_{\cH}\implies \Phi\not\in \EB(\cH)
    \end{equation}
    which is due to the reshuffling criterion \cite{CW03} (see, e.g. \cite[eq.~(47)]{lami_entanglement-saving_2016}).
\end{remark}
\begin{proof}
    We have
    \[
        (   \Phi \otimes \Phi)\circ(\id \otimes \cT) (|\Omega\rangle\langle \Omega|) = \Phi\otimes \Phi(F),
    \]
    where $F=\sum\limits_{i,j=1}^d\ketbra{i}{j}\otimes\ketbra{j}{i}$ is the flip operator.
    Therefore, any witness $X_{AB}\geq 0$ with
    \begin{equation} \label{eq:ANNcert}
        \tr[ X_{AB} (\Phi\otimes \Phi)(F)] < 0
    \end{equation}
    certifies that $\Phi\notin \operatorname{LEA}_2(\cH)$. We rewrite \eqref{eq:ANNcert} as
    \[
        \tr[ ( \Phi^*\otimes  \Phi^* )(X_{AB}) \,F] < 0.
    \]
    Taking $X_{AB} = \Pas := \frac{\mathbb{I}\otimes \mathbb{I} - F}{2}$, the condition becomes
    \begin{equation} \label{eq:ann-ieq-1}
        d_\cH = \tr[ \Phi^*(\mathbb{I})\otimes  \Phi^*(\mathbb{I}) F] < \tr[(  \Phi^*\otimes \Phi^* )(F) F].
    \end{equation}
    using that $\Phi^*$ is unital and $\tr[F]=d_\cH$. The right-hand side can be rewritten as
    \begin{align*}
        \tr[ \Phi\otimes \Phi (F) F] & = \sum_{ij} \tr[ \Phi(\ket{i}\bra{j}) \otimes \Phi (\ket{j}\bra{i})  F] \\
                                     & = \sum_{ij}\tr[ \Phi(\ket{i}\bra{j}) \Phi (\ket{j}\bra{i})]             \\
                                     & =\sum_{ij}\|\Phi (\ket{j}\bra{i})\|_2^2= \|\Phi\|_2^2
    \end{align*}
    using $\tr[(A\otimes B) F] = \tr[AB]$ and that the squared 2-norm of a block matrix is the sum of the squared 2-norms of each submatrix.
    Thus, if $\|\Phi\|_2^2 > d_\cH$,  then $\Phi\notin \operatorname{LEA}_2(\cH)$. \qed
\end{proof}

\begin{corollary}\label{cor:n-shot-LEA2}
    Let $\Phi:\cB(\cH) \to \cB(\cH)$ with $\det(\Phi)\not=0$\footnote{Here, $\det(\Phi)$ simply denotes the product of the eigenvalues of $\Phi$.}. If $\|\Phi^{-k}\|_{2\to 2}\leq d_\cH$, then $\Phi^k\notin \operatorname{LEA}_2(\cH)$.
\end{corollary}
\begin{proof}
    By \Cref{lem:1-shot-LEA2}, $\Phi^k\not\in \LEA_2(\cH)$ if
    \[
        \|\Phi^k\|_2 > \sqrt{d_\cH}.
    \]
    We write
    \[
        \|\Phi^k\|_2^2 = \sum_{ij}\braket{\Phi^k(\ket{i}\bra{j}),\Phi^k(\ket{i}\bra{j})}_\text{HS}
    \]
    Now note that
    \begin{align*}
        \inf\limits_{A\not=0}\frac{\|\lb \Phi^k\rb(A)\|_2}{\|A\|_2}=\inf_{B\not=0}\frac{\|B\|_2}{\|\lb \Phi^{-k}\rb(B)\|_2}=\frac{1}{\|\Phi^{-k}\|_{2\to 2}},
    \end{align*}
    which can be seen by taking $A=\Phi^{-k}(B)$. Thus, it follows that
    \begin{align*}
        \sum_{ij} \langle (\Phi^k)(\ket{j}\bra{i}), ( \Phi^k) (\ket{j}\bra{i}) \rangle_\text{HS}\geq\frac{1}{\|\Phi^{-k}\|_{2\to 2}}d_\cH^2
    \end{align*}
    and we obtain the claim.
    \qed
\end{proof}
\begin{remark}\label{EBinLEA2}
    Since $\operatorname{EB}(\cH)\subset \operatorname{LEA}_2(\cH)$, these also constitute conditions for $\Phi^k$ to be entanglement breaking.
\end{remark}

\subsubsection*{Lower bounds}
In the next proposition, we derive a lower bound on $t_{\operatorname{EB}}$ for a continuous time quantum Markov semigroup using \Cref{lem:1-shot-LEA2}.

\begin{proposition}[Lower bound for $t_{\operatorname{EB}}$ and $t_{\operatorname{LEA_2}}$] \label{tEB_lower}
    For any continuous time quantum Markov semigroup $(\Phi_t)_{t\ge 0}$ on $\cB(\cH)$,
    \begin{align*}
        t_{\operatorname{EB}}((\Phi_{t})_{t\ge 0})    & >t_0,            &
        t_{\operatorname{LEA}_2}((\Phi_{t})_{t\ge 0}) & >\frac{1}{2}t_0, & \text{for} &  & t_0 & :=\frac{\log (d_{\cH}+1)}{\max_{j} |\re (\lambda_j(\cL))|}
    \end{align*}
    where $\{\lambda_j\}_{j=1}^{d_\cH^2}$ are the eigenvalues of $\cL$, the generator of the $\operatorname{QMS}$. In the case that $(\Phi_t)_{t\geq 0}$ is reversible with respect to a faithful state $\sigma$, $\cL$ is self-adjoint with respect to $\braket{\,\cdot\,,\,\cdot\,}_\sigma$, and $\max_{j} |\re (\lambda_j(\cL))|= \|\cL\|_{2,\sigma \to 2,\sigma}$ is the largest eigenvalue (in modulus) of $\cL$.
\end{proposition}
\begin{proof}
    Since $\|\Phi_t\|_2 \geq \|\lambda(\Phi_t)\|_2$ where $\lambda(\Phi_t)$ is the vector of eigenvalues (with multiplicity), we have
    \begin{align*}
        \|\Phi_t\|_2^2 & \geq \sum_{i=1}^{d_\cH^2} |\e^{t \lambda_j(\cL)}|^2 =  \sum_{i=1}^{d_\cH^2} \e^{2t \re (\lambda_j(\cL))} =  \sum_{i=1}^{d_\cH^2} \e^{-2t |\re (\lambda_j(\cL))|} \\
                       & \geq 1 +  (d_\cH^2-1) \e^{-2t \max_{j} |\re (\lambda_j(\cL))|}
    \end{align*}
    using that $\Phi_t$ is the exponential of $-t\cL$ and that $(\Phi_t)_{t\geq 0}$ is trace-preserving, so $\cL$ must have a zero eigenvalue.
    By \Cref{lem:1-shot-LEA2} it follows that
    \begin{align*}
        1 +  (d_\cH^2-1) \e^{-2t \max_{j } |\re (\lambda_j(\cL))|} \geq d_{\cH}
    \end{align*}
    is a sufficient condition for the semigroup not to be entanglement annihilating at time $t$ and the bound $t_{\operatorname{LEA_2}} > \frac{1}{2}t_0$ follows after
    rearranging the terms. The bound $t_\text{EB} > t_0$ follows in the same way from \eqref{eq:1norm_not_EB}.
    \qed
\end{proof}
\medskip

\begin{example}[The depolarizing channel]
    The quantum depolarizing channel $\tilde \Phi_p$ of parameter $p$ is a unital quantum channel modeling isotropic noise,
    \begin{equation*}
        \tilde{\Phi}_p\p \rho):\,=\, \p1-p) \frac{\mathbb{I}}{d_\cH}\tr[\rho]+p\,\rho.
    \end{equation*}
    Its output can be interpreted as the time-evolved state under the so-called quantum depolarizing semigroup:
    \begin{equation*}
        \Phi_t\p \rho)=\tilde{\Phi}_{\e^{-t}}(\rho)=\p1-\e^{-t})~\frac{\mathbb{I}}{d_\cH}\tr[\rho]+\e^{-t}\,\rho.
    \end{equation*}
    This continuous time QMS is primitive, with invariant state $\frac{\mathbb{I}}{d_\cH}$, and has generator $\cL(\rho) =  \frac{\mathbb{I}}{d_\cH}\tr[\rho]-\rho$, which is the difference of a rank-1 projection and the identity map, and therefore has eigenvalues $0$ and $-1$. The spectral gap of $1$ provides the following strong decoherence via \eqref{poincareDF}:
    \begin{equation*}
        \|\id\otimes \Phi_t\p \rho)- \tr_2\p\rho)\otimes {\mathbb{I}}/{d_\cH}\|_1\le \sqrt{d_\cH}\,\e^{-t}.
    \end{equation*}
    \Cref{tEB_lower} and \Cref{tEB_upper} therefore yield lower and upper bounds on the entanglement-breaking time, namely
    \[
        \log(d_\cH + 1) \leq t_{\EB}\p (\Phi_t)_{t\ge 0})\le \tfrac{5}{2}\log(d_\cH).
    \]
    These bounds match up to constant factors. In fact, the lower bound is tight:
    \[
        t_{\EB}\p (\Phi_t)_{t\ge 0}) = \log(d_\cH+1)
    \]
    as calculated in, e.g. \cite{lami2016}. Likewise, considering time evolution under $\Phi_t\otimes \Phi_t$, \Cref{tEB_lower}  and \Cref{equ:eafromspgap} yield the bounds
    \[
        \frac{1}{2}\log(d_\cH + 1) \leq t_{\operatorname{LEA_2}}((\Phi_t)_{t\ge 0})\leq 3\,\log(d_\cH)
    \]
    which can be compared to the exact result,
    \[
        t_{\operatorname{LEA_2}}((\Phi_t)_{t\ge 0} )= \log \left( \frac{d_\cH + 1 + \sqrt{3}}{1 + \sqrt{3}} \right)
    \]
    from \cite{lami2016}.
\end{example}

\section{Structural properties of discrete time evolutions} \label{sec3}
In the previous section, we showed that the only continuous time quantum Markov semigroups which become entanglement breaking in finite time were primitive quantum Markov semigroups.  Now, we investigate the discrete time case, where the situation becomes much more subtle. First, while every primitive discrete time QMS is eventually entanglement breaking, not all eventually entanglement breaking discrete time QMS are primitive. For example, the quantum channel
\[
    \Phi(\rho) = \tr(\rho)\,\ket{0}\bra{0}
\]
has a rank-1 invariant state (and thus is not primitive), but the associated discrete time QMS $\{\Phi^n\}_{n=1}^\infty$ is  entanglement-breaking.

\begin{proposition}[Density of classes of quantum channels] \label{thm:EEB_dense} The sets of faithful quantum channels, primitive quantum channels, and eventually entanglement breaking quantum channels are dense in the set of all quantum channels. Likewise, the sets of faithful PPT quantum channels, primitive PPT
    quantum channels, and eventually entanglement breaking PPT quantum channels are dense in the set of all PPT quantum channels.
\end{proposition}
\begin{remark}
    C.f. Theorem V.2 of \cite{rahaman2018eventually} which proves that EEB unital channels are dense in the set of unital channels and the remarks after~ \cite[Theorem 11]{lami_entanglement-saving_2016}, where they show that EEB quantum channels are of full measure in the set of quantum channels.
\end{remark}
\begin{proof}
    First note that the set of faithful quantum channels is dense in the set of quantum channels.
    To see this, let $\Phi:\cB(\cH) \to \cB(\cH)$ be a quantum channel and for a state $\tau\in\cD\p \cH)$ define
    $\Psi_{\tau}:\cB(\cH) \to \cB(\cH)$ to be the quantum channel that acts as
    \begin{align*}
        \Psi_{\tau}(X)=\tr(X)\tau\,.
    \end{align*}
    Note that $\Psi_\tau^2=\Psi_\tau$.
    For $\epsilon\in (0,1)$, we define the channel $\Phi' :=(1-\epsilon)\Phi+\epsilon \,\Psi_{\frac{\mathbb{I}}{d}}$ which maps
    quantum states to strictly positive operators. This is known to be equivalent to the irreducibility of $\Phi'$~\cite[Theorem 6.2]{wolf2012quantum}, which in particular implies the existence of a stationary state $\sigma\in\cD_+\p \cH)$
        of full rank. As $\eps\to 0$, we have $\Phi'\to \Phi$, which shows the density of faithful channels.
    Next, we reuse the same $\Phi'$ and consider the quantum channel $\Phi'' :=(1-\epsilon)\Phi'+\epsilon \Psi_\sigma$. As $\eps\to 0$, $\Phi''\to \Phi$, and moreover for each $\eps\in(0,1)$, $\Phi''$ is eventually entanglement breaking. To see this, note that as $\Phi'$ has $\sigma$ as a stationary state, it follows that $\Phi'\Psi_\sigma=\Psi_\sigma \Phi'=\Psi_\sigma$ and, thus,
    \begin{align*}
        \lb \Phi''\rb^n=(1-\epsilon)^n\lb \Phi'\rb^n+(1-(1-\epsilon)^n)\Psi_\sigma\,.
    \end{align*}
    Clearly, $\lim\limits_{n \to\infty}\lb \Phi''\rb^n=\Psi_\sigma$, which is primitive.
    We have
    \begin{align*}
        \id\otimes \Psi_\sigma\p |\Omega\rangle\langle \Omega|)=\frac{\mathbb{I}}{d_\cH}\otimes\sigma.
    \end{align*}
    As we saw in \Cref{robustnesssep}, separable states of the form
    $\sigma_1\otimes\sigma_2$ with $\sigma_1,\sigma_2\in\cD_+\p \cH)$ full rank are in the interior of the set of separable states. Thus,
    the Choi matrix of $\lb \Phi''\rb^n$ converges to a separable state in the interior of the set of separable states and will be separable
    for some finite $n$. This also allows us to conclude the density of primitive channels, since they are in particular EEB.

    The corresponding results for PPT channels are obtained by noting that if $\Phi$ is PPT, then $\Phi'$ and $\Phi''$ are PPT as well.
    \qed
\end{proof}

One potentially useful observation is that to prove the PPT$^2$ conjecture, one may restrict to a dense set. More generally, define for $k\in \mathbb{N}$ the PPT$^k$ conjecture as the claim that any PPT quantum channel $\Phi$ has $\Phi^k\in \EB(\cH)$.
\begin{corollary} \label{cor:PPT2-reduction-prim}
    The $\operatorname{PPT}^k$ conjecture holds if and only if every $\operatorname{PPT}$ primitive quantum channel $\Phi$ has $\Phi^k\in \EB(\cH)$.
\end{corollary}
\begin{proof}
    Let $\Phi$ be a PPT quantum channel. By \Cref{thm:EEB_dense}, there is a sequence $\Phi_n\to \Phi$ with each $\Phi_n$ being PPT and primitive. By continuity of $x\mapsto x^k$, we have $\Phi^k = \lim_{n\to\infty} \Phi_n^k$. Since the set of entanglement-breaking channels is closed, $\Phi^k\in \EB(\cH)$.
\end{proof}

\medskip
As mentioned before, our goal in this paper is to estimate the time after which a quantum system undergoing a quantum Markovian evolution has lost all its entanglement. In order to better characterize discrete-time evolutions for which asking this question makes sense, we first need to leave aside those evolutions for which the phenomenon does not occur, that is, evolutions that either destroy entanglement after an infinite amount of time (EB$_\infty$), or even those of never-vanishing output entanglement (AES).

A big part of this question was already answered in the discrete time case by \cite{lami_entanglement-saving_2016}. In this paper (Theorem 21), the authors showed that, given a quantum channel $\Phi:\cB(\cH)\to\cB(\cH)$ with $\dim(\ker\Phi)<2(d_\cH-1)$, $\{\Phi^n\}_{n\in\NN}\in \operatorname{ES}(\cH)$ if and only if either it has a non-full rank positive fixed point, or the number of peripheral eigenvalues is strictly greater than $1$, which itself is equivalent to the existence of $1\le n\le d_\cH$ such that $\Phi^n$ has a non-full rank positive fixed point. In the same paper, the authors showed that if $\Phi$ has more than $d_\cH$ peripheral eigenvalues, then $\{\Phi^n\}_{n\in\NN}$ is asymptotically entanglement saving. These interesting results clearly show the link between the spectral properties of the quantum channel $\Phi$ and the entanglement properties of the corresponding discrete time quantum Markov semigroup $\{\Phi^n\}_{n\in\NN}$. In the next subsections, we further develop this intuition. First, we prove the following simple consequence of a result from \cite{lami_entanglement-saving_2016}.

\begin{lemma} \label{d-periph-ES}
    Let $\Phi$ be a quantum channel on $\cH$ with $d_\cH$ peripheral eigenvalues counted with multiplicity and at least one non-zero non-peripheral eigenvalue. Then $\{\Phi^n\}_{n\in\NN}\in\operatorname{ES}(\cH)$.
\end{lemma}

\begin{proof}
    For any $N$, $\Phi^N$ has $d_\cH$ peripheral eigenvalues and at least one non-peripheral eigenvalue non-zero. Thus, if $\{\lambda_k\}_{k=1}^{d^2}$ are the eigenvalues of $\Phi$ counted with multiplicity, we have
    \[
        \|\Phi^N\|_1 \geq \sum_{k=1}^{d^2} |\lambda_k^N| = d_\cH + \sum_{\lambda_k : |\lambda_k|<1} |\lambda_k^N| > d_\cH\,.
    \]
    The result follows from the fact that a quantum channel $\Psi:\cB(\cH)\to \cB(\cH)$ such that $\|\Psi\|_1>d_\cH$ is not entanglement breaking (see \cite{lami_entanglement-saving_2016}).  \qed
\end{proof}

\subsection{Irreducible evolutions} \label{sec:irred}

The previous result suggest looking at channels with less than $d_\cH$ peripheral eigenvalues counted with multiplicity in order to characterize EEB. Since primitive channels are eventually entanglement-breaking, a natural next step is to consider the wider class of irreducible channels introduced in \Cref{sec:preliminaries}.

    The study of the peripheral spectral properties of irreducible maps is the subject of the non-commutative Perron-Frobenius theory for irreducible completely positive maps; see \cite{evans_spectral_1978}, or \cite[Section 6.2]{wolf2012quantum}. See also \cite[Appendix A]{HJPR2} for a summary of this theory and extensions to deformations of irreducible CPTP maps. Together with the Jordan decomposition (see, e.g. \cite[Section 1.5.4]{Kato}), this theory provides a useful decomposition of irreducible quantum channels. In the next proposition, we recall this decomposition and provide a minimal set of quantities needed to construct such a map. This structural result will be used later to construct irreducible quantum channels that are not EEB. Moreover, it will allow us to show that irreducibility combined with the PPT property is enough to ensure that a quantum channel is EEB.

\begin{proposition} \label{prop:decomp-irred}
    Consider
    \begin{enumerate}
        \item A number $z\in \{1,\dotsc,d_\cH\}$,
        \item An orthogonal resolution of the identity $\{ p_n\}_{n=0}^{z-1}$, i.e., $\sum_{n=0}^{z-1}p_n = \one$ and $p_n^\dagger = p_n^2 = p_n$ for each $n$,
        \item A faithful state $\sigma$ such that $[\sigma,p_n]=0$ and $\tr[\sigma p_n] = \frac{1}{z}$, for each $n=0,\dotsc,z-1$,
        \item A linear map $ \Phi_Q$ such that:
              \begin{enumerate}
                  \item \label{it:Phi_Q-spr} $\spr ( \Phi_Q) < 1$
                  \item \label{it:Phi_Q-J} $J( \Phi_Q) \geq - z\left(\sigma \otimes \one\right)L_{1}$, where for $k=0,\dotsc,z-1$, we define $L_k := \sum_{n=0}^{z-1}  p_{n-k} \otimes  p_{n}$ where the subscripts are taken modulo $z$.
                  \item \label{it:Phi_Q-Pj} We have
                        \begin{equation} \label{eq:Phi_Q_kills_pn}
                            \Phi_Q( \sigma p_n) = \Phi_Q^*(p_n) = 0, \qquad \forall n = 0,1,\dotsc,z-1.
                        \end{equation}
              \end{enumerate}
    \end{enumerate}
    Let
    \begin{equation} \label{eq:Phi-irred}
        \Phi:= \sum_{n=0}^{z-1}  \theta^n P_n + \Phi_Q
    \end{equation}
    where $P_n(\cdot) = \tr[ u^{-n} \, \cdot\, ] u^n \sigma$ for $u := \sum_{k=0}^{z-1}  \theta^k  p_k$ and $ \theta := \exp(2\iu \pi / z)$.
    Then  $\Phi$ is an irreducible quantum channel. On the other hand, any irreducible quantum channel $\Phi$ can be decomposed as \eqref{eq:Phi-irred} for some choices of $z, \{p_n\}_{n=0}^{z-1}$, $\sigma$, and $\Phi_Q$ as in (1)--(4).
    Moreover, in either case, $\sigma$ is the unique fix point\footnote{up to a multiplicative constant} of $\Phi$; $P_n(\cdot)$ are its peripheral eigenprojections, associated to eigenvalues $\theta^n$ and eigenvectors $u^n \sigma$; and, for any $j,k =0,\dotsc,z-1$, we have the intertwining relations
    \begin{equation} \label{eq:irred_on_p_blocks}
        \Phi(p_j X p_k) = p_{j-1}\Phi(X) p_{k-1}, \quad \text{and}\quad \Phi^*(p_j X p_k) = p_{j+1}\Phi^*(X) p_{k+1}, \quad \forall X\in \cB(\cH)
    \end{equation}
    where the subscripts are interpreted modulo $z$. Additionally, for $\Phi_P:= \sum_{n=0}^{z-1}  \theta^n P_n$,
    \begin{equation} \label{eq:irred-choi-P}
        J(\Phi_P^k) = \hat J_k :=  z\left(\sigma \otimes \one\right) L_{k}= \sum_{m=0}^{z-1} \tr[p_m] \frac{p_{m-k}\sigma p_{m-k}}{\tr[ p_{m-k} \sigma]} \otimes \frac{p_m}{\tr[p_m]}.
    \end{equation}

\end{proposition}
\begin{proof}
    Let us note that \eqref{it:Phi_Q-Pj} is equivalent to the property that
    \begin{equation} \label{eq:Phi_Q-kills-Pj}
        \Phi_Q\circ P_j =  P_j \circ \Phi_Q = 0, \qquad \forall\, j=0,\dotsc,z-1.
    \end{equation}
    To see this, note that the generalized discrete Fourier transform
    \[
        \cF : \bigoplus_{j=0}^{z-1} \cB(\cH) \to \bigoplus_{j=0}^{z-1} \cB(\cH)
    \]
    given by $\cF( (X_0,\dotsc, X_{z-1} ) ) = (Y_0,\dotsc, Y_{z-1} )$ for $Y_n = \sum_{j=0}^{z-1} \theta^{nj}X_j$ is an invertible linear transformation, with inverse $\cF\inv( (Y_n)_{n=0}^{z-1} ) = \frac{1}{z}\cF( (Y_{z-\ell})_{\ell=0}^{z-1} )$. All indices are taken mod $z$. Next, using the definition of the $P_n$, \eqref{eq:Phi_Q-kills-Pj} is equivalent to
    \[
        \Phi_Q( u^{n} \sigma) = 0, \qquad  \Phi_Q^*(u^n)=0, \qquad \forall\, j=0,\dotsc,z-1.
    \]
    Since $\vec 0 = (\Phi_Q(u^n \sigma))_{n=0}^{z-1} = \cF( (\Phi_Q(p_j \sigma) )_{j=0}^{z-1} )$ and $\vec 0 = (\Phi_Q(u^{n}))_{n=0}^{z-1} = \cF( (\Phi_Q(p_j))_{j=0}^{z-1} )$, and $\cF$ has trivial kernel, \eqref{eq:Phi_Q-kills-Pj} implies \eqref{it:Phi_Q-Pj}. The converse follows similarly.
    \Cref{eq:irred_on_p_blocks} follows from a simple computation. The fact that an irreducible map can be decomposed as (\ref{eq:Phi-irred}) for some choices of $z$, $\{p_n\}_{n=0}^{z-1}$, $\sigma$ and $\Phi_Q$ as in (1)-(4b) and (\ref{eq:Phi_Q-kills-Pj}) is not new, and we refer to \cite[Section 6.2]{wolf2012quantum} for more details. We believe however that the forward implication is, and postpone its proof to \Cref{proof:prop:decomp-irred} for sake of clarity.

    \qed
\end{proof}

\medskip
\begin{remarks*}
    ~\begin{itemize}
        \item The matrix $\hat J_k$ is separable, and thus $\Phi_P^k$ is entanglement-breaking, for any $k\geq 1$.
        \item If $z >1$, then $\frac{1}{d_\cH}\hat J_k$ does not have full support. Thus, $\frac{1}{d_\cH}\hat J_k$ is on the boundary of the set of density matrices, and thus on the boundary of $\SEP$ and $\PPT$ as well. In fact, we can say something stronger than this: whenever $z>1$, there exist entangled density matrices arbitrarily close to each $\frac{1}{d_\cH} \hat J_k$, $k=0,1,\dotsc,z-1$. To see this, note that $\frac{1}{d_\cH}\hat J_k \in L_k \cD(\cH \otimes \cH ) L_k$. However, we can construct entangled states in $L_j \cD(\cH \otimes \cH ) L_j$ for any $j\neq k$. For instance, let $\ket{0} \in p_0 \cH$, $\ket{1} \in p_1 \cH$, $\ket{ j_0 } \in p_{-j} \cH$, and $\ket{j_1} \in p_{1-j} \cH$ be normalized vectors. Then
              \[
                  \ket{\Omega_j} :=\frac{1}{\sqrt{2}}\big( \ket{j_0}\otimes \ket{0} + \ket{j_1}\otimes \ket{1}  \big)
              \]
              is (local-unitarily equivalent to) a Bell state, and has
              \[
                  L_j \ket{\Omega_j} = \sum_{n=0}^{z-1} (p_{n-j} \otimes p_n )\ket{\Omega_j} = \ket{\Omega_j},
              \]
              and thus $L_j \ket{\Omega_j}\bra{\Omega_j} L_j \in L_j \cD(\cH \otimes \cH ) L_j$. Thus, for any $t\geq 0$,
              \[
                  (1-t)\frac{1}{d_\cH} \hat J_k + t \ket{\Omega_j}\bra{\Omega_j} = (1-t) \frac{1}{d_\cH} \hat J_k \Big|_{L_k} \,\oplus \,t\ket{\Omega_j}\bra{\Omega_j}\Big|_{L_j}
              \]
              is an entangled density matrix, and can be made arbitrarily close to  $\frac{1}{d_\cH} \hat J_k $ by sending $t\to 0$.

              Moreover, the limit points of $\{\frac{1}{d_\cH}J(\Phi^n)\}_{n=0}^\infty$ are exactly $\{\frac{1}{d_\cH}\hat J_k\}_{k=0}^{z-1}$, which follows from the mixing time results in the sequel (e.g. \Cref{prop:discrete-Poincare}). This shows the  analysis of \Cref{upperbounds} and \Cref{lowerbounds} in the primitive case does not carry over to the irreducible discrete-time case, because the aforementioned limit points are neither in the interior of the set of separable states nor the interior of the set of PPT states.
        \item The above proposition shows that the peripheral eigenvectors of irreducible channels $\Phi$ commute. In Theorem 32 of \cite{lami_entanglement-saving_2016}, the authors show that asymptotically entanglement saving channels are characterized by the fact that they possess at least two noncommuting phase points. This implies $\{\Phi^n\}_{n\in\NN} \in \AEB(\cH)$, which generalizes Corollary 6.1 of \cite{rahaman2018eventually} to the non-unital case. There, the authors show that a unital irreducible quantum channel is AEB if and only if its phase space is commutative.
        \item The intertwining property \eqref{eq:irred_on_p_blocks} holds for $\Phi$ and $\Phi_P$ (which itself is an irreducible map), and therefore for $\Phi_Q$. This implies $J(\Phi_Q) = L_1 J(\Phi_Q) L_1$, i.e., the Choi matrix of $\Phi_Q$ is supported on the same subspace as that of $\Phi_P$.
        \item Given a map $\Phi_Q$ which intertwines with $\{p_n\}_{n=0}^{z-1}$, a sufficient condition for $\Phi_Q\ge - z\left(\sigma \otimes \one\right)L_{1}\equiv -J(\Phi_P)$ is given by
              \begin{equation}\label{eq:Phi_Q-suff-cond-for-CPB}
                  \|\Phi_Q\|_2 =\|J(\Phi_Q)\|_2 \leq z \lambda_\text{min}(\sigma)\,,
              \end{equation}
              since in that case
              \[
                  J(\Phi_Q) \geq - \spr(\Phi_Q) L_1 \geq -\|J(\Phi_Q)\|_2 L_1 \geq  -z \lambda_\text{min}(\sigma) L_1 \geq -z L_1(\sigma\otimes\id)L_1 =- J(\Phi_P).
              \]
    \end{itemize}
\end{remarks*}

The intertwining property \eqref{eq:irred_on_p_blocks} of irreducible maps is very useful for understanding their entanglement breaking properties. In fact, a slight generalization of this property will prove useful.

\begin{definition}[$\big(\{p_i\}_{i=0}^{z-1}, \{\tilde p_i\}_{i=0}^{z-1}\big)$-block preserving]
    Given two orthogonal resolutions of the identity, $\{p_i\}_{i=0}^{z-1}$ and $\{\tilde p_i\}_{i=0}^{z-1}$, we say that a quantum channel $\Phi$ is $\big(\{p_i\}_{i=0}^{z-1}, \{\tilde p_i\}_{i=0}^{z-1}\big)$-block preserving if for all $i,j\in \{0,\dotsc,z-1\}$,
    \[
        \Phi( p_i \cB(\cH) p_j) \subset \tilde p_i \cB(\cH)\tilde p_j.
    \]
\end{definition}
From \eqref{eq:irred_on_p_blocks}, if $\Phi$ is irreducible, then $\Phi^k$ is $\big(\{p_i\}_{i=0}^{z-1}, \{ p_{i-k}\}_{i=0}^{z-1}\big)$-block preserving. Using this notion, the following result shows that PPT channels which are block preserving in the above sense must annihilate off-diagonal blocks. We note that results of a similar flavor were shown in~\cite{Cariello2016,Cariello2018}.
\begin{lemma} \label{lem:block preserving-PPT}
    If $\Phi$ is a $\operatorname{PPT}$ quantum channel and  is $\big(\{p_i\}_{i=0}^{z-1}, \{ \tilde p_{i}\}_{i=0}^{z-1}\big)$-block preserving, then
    \begin{equation}\label{eq:PPT-kills-offdiag}
        \Phi( p_i \cB(\cH) p_j) = \{0\}
    \end{equation}
    for all $i\neq j$.
\end{lemma}
\begin{proof}
    Let us prove the contrapositive. Assume for some $i\neq j$, $\Phi(p_i \cB(\cH) p_j) \neq \{0\}$; without loss of generality, take $i=0$ and $j=1$. Then let $\ket{0}\bra{1}\in p_0 \cB(\cH) p_1$ such that $\Phi(\ket{0}\bra{1})\neq 0$. Then also $\Phi(\ket{1}\bra{0})\neq 0$.
    Consider $\Omega_{01}$, the Bell-type state associated to $\ket{00}+\ket{11}$. The block-preserving assumption yields
    \[
        \Phi(\ket{i}\bra{j}) =\Phi(p_i\ket{i}\bra{j}p_j) = \tilde p_i \,\Phi(\ket{i}\bra{j}) \tilde p_j
    \]
    for each $i,j\in \{0,1\}$.
    Then, neglecting rows and columns of all zeros, $(\Phi\otimes\id)(\Omega_{01})$ can be written as
    \[
        (\Phi\otimes \id)(\Omega_{01}) = \sum_{i,j=0}^1 \Phi(\ket{i}\bra{j})\otimes\ket{i}\bra{j}= \left(\begin{array}{@{}c|c@{}}
                \begin{matrix}
                    \Phi(\ket{0}\bra{0}) & 0 \\
                    0                    & 0
                \end{matrix}
                                           & \begin{matrix}
                    0 & \Phi(\ket{1}\bra{0}) \\
                    0 & 0
                \end{matrix} \\
                \hline
                \begin{matrix}
                    0                    & 0 \\
                    \Phi(\ket{0}\bra{1}) & 0
                \end{matrix} &
                \begin{matrix}
                    0 & 0                    \\
                    0 & \Phi(\ket{1}\bra{1})
                \end{matrix}
            \end{array}\right)
    \]
    as $\{\tilde p_0,\tilde p_1\}$ blocks in the  $\{\ket{0},\ket{1}\}$ basis.
    Now, taking the partial transpose on the first system,
    \[
        (\mathcal{T}\otimes \id )\circ (\Phi\otimes \id)(\Omega_{01}) = \left(\begin{array}{@{}c|c@{}}
                \begin{matrix}
                    \Phi(\ket{0}\bra{0}) & 0 \\
                    0                    & 0
                \end{matrix}
                                           & \begin{matrix}
                    0                    & 0 \\
                    \Phi(\ket{1}\bra{0}) & 0
                \end{matrix} \\
                \hline
                \begin{matrix}
                    0 & \Phi(\ket{0}\bra{1}) \\
                    0 & 0
                \end{matrix} &
                \begin{matrix}
                    0 & 0                    \\
                    0 & \Phi(\ket{1}\bra{1})
                \end{matrix}
            \end{array}\right).
    \]
    The eigenvalues of this matrix are the eigenvalues of $  \Phi(\ket{0}\bra{0})$, together with the eigenvalues of $\Phi(\ket{1}\bra{1})$, and the eigenvalues of the block matrix
    \[
        X =\begin{pmatrix}
            0                    & \Phi(\ket{1}\bra{0}) \\
            \Phi(\ket{0}\bra{1}) & 0
        \end{pmatrix}_.
    \]
    Since $X$ is non-zero, self-adjoint, and traceless, it must have both strictly positive and strictly negative eigenvalues. Thus, $(\mathcal{T}\otimes \id )\circ(\Phi\otimes \id)(\Omega_{01}) $ has negative eigenvalues, so $(\Phi\otimes \id)(\Omega_{01})$ is not PPT.
\end{proof}

\begin{theorem}\label{theoremEB}
    Let $\Phi$ be an irreducible $\operatorname{CPTP}$ map, $k\geq 1$, and let us adopt the notation of \Cref{prop:decomp-irred}. Assume $\Phi$ is not primitive (i.e $z\geq 2$). Then
    \begin{equation} \label{eq:irred-nec_for_EB}
        \Phi^k\in \PPT(\cH) \implies \Phi^k(p_i \cB(\cH) p_j) = \{0\} \quad \forall\, i\neq j.
    \end{equation}
    On the other hand, if
    \begin{equation}\label{eq:assume_PhiQ_kills_offdiag}
        \Phi^k(p_i \cB(\cH) p_j) = \{0\} \quad \forall\, i\neq j
    \end{equation}
    and additionally, for each $j$ such that $\rank p_j \geq 2$,
    \begin{equation} \label{eq:irred-JQ-block-EB}
        \left\|J(\Phi_Q^k|_{p_j \cB(\cH)p_j})\right\|_2\leq z\lambda_{\min}(\sigma|_{p_j\cH})
    \end{equation}
    then $\Phi^k\in \EB(\cH)$.

    In the case $z=d_\cH$, we may write $p_j=\ket{j}\bra{j}$ for $j=0,\dotsc,z-1$. In this case,  $\Phi_Q^k(\ket{i}\bra{j})=0$ for all $i\neq j$, if and only if $\Phi^k\in \EB(\cH)$.
\end{theorem}

\begin{remark}
    Under the assumption \eqref{eq:assume_PhiQ_kills_offdiag},
    \begin{align}\label{EB}
        \|J(\Phi_Q^k)\|_2  \leq z\lambda_\text{min}(\sigma)
    \end{align}
    implies \eqref{eq:irred-JQ-block-EB}, as
    \[
        \|J(\Phi_Q^k|_{p_j \cB(\cH)p_j})\|_2 \leq \sum_{n=0}^{z-1}\|J(\Phi_Q^k|_{p_n \cB(\cH)p_n})\|_2 = \|J(\Phi_Q^k)\|_2  \leq z\lambda_\text{min}(\sigma) \leq  z\lambda_\text{min}(\sigma|_{p_j}).
    \]
\end{remark}

\begin{proof}
    \eqref{eq:irred-nec_for_EB} follows immediately from the fact that if $\Phi$ is irreducible, then $\Phi^k$ is $\big(\{p_i\}_{i=0}^{z-1}, \{ p_{i-k}\}_{i=0}^{z-1}\big)$-block preserving, and \Cref{lem:block preserving-PPT}.

    Next, assume \eqref{eq:assume_PhiQ_kills_offdiag} holds. Let $\{\ket{i}\}_{i=0}^{d_\cH-1}$ be an orthonormal basis of $\cH$ such that there are disjoint index sets $\{I_n\}_{n=0}^{z-1}$ such that for each $n$, $\ket{i}\in p_n \cH$ for $i\in I_n$. Taking the Choi matrix in the $\ket{i}$ basis,
    \begin{align*}
        J(\Phi^k) & = \sum_{i,j}\Phi^k(\ket{i}\bra{j})\otimes \ket{i}\bra{j}                        \\
                  & = \sum_{n=0}^{z-1}\sum_{i,j\in I_n}\Phi^k(\ket{i}\bra{j})\otimes \ket{i}\bra{j}
    \end{align*}
    using that $\Phi^k(\ket{i}\bra{j})=0$ whenever $i$ and $j$ do not share an index set $I_n$, which follows from \eqref{eq:assume_PhiQ_kills_offdiag}. Then (see e.g. \eqref{eqJPhiP}):
    \begin{align*}
        J(\Phi^k) & = \sum_{n=0}^{z-1}\sum_{i,j\in I_n}\left(\delta_{i,j}\,z\, p_{n-k}\,\sigma + \Phi^k_Q(\ket{i}\bra{j})\right)\otimes \ket{i}\bra{j}  \\
                  & = \sum_{n=0}^{z-1}\left(z \,p_{n-k}\,\sigma \otimes p_n + \sum_{i,j\in I_n}\Phi^k_Q(\ket{i}\bra{j})\otimes \ket{i}\bra{j}\right)\,.
    \end{align*}
    We note that $\sum_{i,j\in I_n}\Phi^k_Q(\ket{i}\bra{j})\otimes \ket{i}\bra{j}\in p_{n-k}\otimes p_{n}\cB(\cH\otimes \cH)p_{n-k}\otimes p_{n}$, since both $\Phi$ and $\Phi_P$ map $p_n\cB(\cH)p_n$ to $p_{n-1}\cB(\cH)p_{n-1}$ (see \cref{eq:irred_on_p_blocks,eq:PhikP_formula}).
    Since by assumption
    \[
        \|J(\Phi^k_Q|_{p_n\cB(\cH)p_n})\|_2\leq z\lambda_\text{min}(\sigma|_{p_n\cH})
    \]
    then \eqref{eq:robusteness_of_sep_bound} applied to the Hilbert space $p_{n-k} \cH\,\otimes\, p_n\cH$ gives that $z \sigma \otimes \one |_{p_{n-k}\cH\,\otimes\, p_n\cH} + J(\Phi^k_Q|_{p_n\cB(\cH)p_n})$ is separable on that space. We may embed this state in $\cB(\cH\otimes \cH)$ (without changing the tensor product structure) yielding that
    \[
        p_{n-k}\otimes p_n \left(z \sigma \otimes \one  + J(\Phi^k_Q)\right)p_{n-k}\otimes p_n
    \]
    is a non-full-rank separable state on $\cB(\cH\otimes \cH)$. Summing over $n$ then yields the fact that $J(\Phi^k)$ is separable, so $\Phi^k\in \EB(\cH)$.

    For the case $z=d$, we simply note that $\rank p_j = 1$ for all $j$, and hence the statement follows from the above two results.\qed
\end{proof}

Since the limit points of $\{J(\Phi^n)\}_{n=1}^\infty$ are separable but arbitrarily close to entangled states (as shown in the remarks following \Cref{prop:decomp-irred}) the question arises of whether or not there are quantum channels that are both irreducible and in $\operatorname{EB}_\infty(\cH)$. In the case that $\Phi$ has maximal period $z=d_\cH$, \Cref{d-periph-ES} resolves this affirmatively as long as $\Phi_Q$ is not nilpotent. In the case when the period is much less than the dimension; say $z=2 <d_\cH$, then the underlying argument (relying on the reshuffling criterion via \eqref{eq:1norm_not_EB}) provides little help: $\|\Phi^n\|_1 = z + o(n) < d_\cH$ for large $n$. However, using \Cref{theoremEB}, we can design $\EB_\infty(\cH)$ irreducible channels rather easily, as shown in the following example.
\begin{example}
    Let us construct an irreducible quantum channel $\Phi$ via \Cref{prop:decomp-irred}. with period $z=2$. We choose any full rank state $\sigma$ as the invariant state, and any pair of orthogonal projections $\{p_1,p_2\}$ which commute with $\sigma$ as the Perron-Frobenius projections. Let $d_j=\rank p_j$ for $j=0,1$, and let $\{\ket{e_i}\}_{i=0}^{d_0-1}$ be an eigenbasis of $\sigma p_0$ and likewise $\{\ket{f_i}\}_{i=0}^{d_1-1}$ be an eigenbasis of $\sigma p_1$. For some $\lambda\in \R$ with $|\lambda|<1$, define $\Phi_Q$ by
    \begin{gather*}
        \Phi_Q(\ket{e_i}\bra{e_j}) = \Phi_Q(\ket{f_i}\bra{f_j}) = 0, \\
        \Phi_Q(\ket{e_i}\bra{f_j})^* =
        \Phi_Q(\ket{f_i}\bra{e_j}) = \delta_{i,0}\delta_{j,0}\bar\lambda \ket{e_0}\bra{f_0}
    \end{gather*}
    for each $i,j$. Then $\Phi_Q(\sigma p_j)=\Phi_Q^*(p_j)=0$, $\spr(\Phi_Q)=|\lambda|<1$. Since
    \[
        J(\Phi_Q) = \lambda \ket{f_0}\bra{e_0} \otimes \ket{f_0}\bra{e_0}  +\bar\lambda \ket{e_0}\bra{f_0}\otimes  \ket{f_0}\bra{e_0} \,,
    \]
    we have $\|J(\Phi_Q)\|_2 = \sqrt{2}\,|\lambda|$. Thus, choosing $\lambda$ to satisfy $0<|\lambda| < \min(1,  \sqrt{2} \,\lambda_\text{min}(\sigma))$, we have  $J(\Phi_Q)\geq -J(\Phi_P)$ by \eqref{eq:Phi_Q-suff-cond-for-CPB}, and  $\Phi$ is an irreducible CPTP map of period 2. Moreover,
    \[
        \Phi_Q^n(\ket{e_0}\bra{f_0}) = \begin{cases}
            \lambda^n \ket{f_0}\bra{e_0} & n\text{ odd}  \\
            \lambda^n \ket{e_0}\bra{f_0} & n\text{ even}
        \end{cases}
    \]
    which is non-zero for any $n$. Thus, $\Phi\in \EB_\infty(\cH)$ by \Cref{theoremEB}. Additionally, it was proved in Theorem 21 of \cite{lami_entanglement-saving_2016} that if the number of zero eigenvalues of $\Phi$ is strictly less than $2(d_\cH-1)$, the fact that $\Phi$ has at least two peripheral eigenvalues implies that it is entanglement saving. However, in the present example,
        $\Phi$ has four non-zero eigenvalues (two peripheral, and $\pm |\lambda|$), and therefore $d_\cH^2-4$ zero eigenvalues. Thus, for $d_\cH\geq 3$, Theorem~11 of \cite{lami_entanglement-saving_2016} does not apply to $\Phi$.
\end{example}
\medskip
\Cref{theoremEB} implies the following corollary, which is extended to the non-irreducible case in the next section. Estimates on the entanglement-breaking index, which we recall is the first $n\in\NN$ such that $\Phi^n$ is entanglement-breaking, are provided in \Cref{EBTdiscrete}.
\begin{corollary} \label{cor:irred-PPT-EEB}
    Any irreducible $\operatorname{PPT}$ channel $\Phi:\cB(\cH)\to \cB(\cH)$ is eventually entanglement breaking.
\end{corollary}
\begin{proof}
    Since the channel is PPT, we have \eqref{eq:assume_PhiQ_kills_offdiag}. Then setting $\ell:=\spr(\Phi_Q)<1$, by Gelfand's formula we have that for any $\eps\in(0,1-\ell)$, there exists $n_0>0$ such that for all $k\geq n_0$,
    \[
        \|J(\Phi_Q^k)\|_2=\|\Phi_Q^k\|_2 \leq (\ell +\eps)^k.
    \]
    Thus, for $k$ large enough, $\|J(\Phi_Q^k)\|_2  \leq z\lambda_\text{min}(\sigma)$ and \eqref{eq:irred-JQ-block-EB} holds. Hence, $\Phi^k \in \EB(\cH)$.
    \qed
\end{proof}

\subsection{Beyond irreducibility}
A non-irreducible channel can be decomposed into \emph{irreducible components}, on which it acts irreducibly. More specifically, we may decompose the identity $\one$ of $\cH$ into maximal subspaces with corresponding orthogonal projections $D, P_1,\dots,P_n$ s.t. $\Phi$ restricted to $P_i\cB(\cH)P_i$ is irreducible, and $D\cH$ is orthogonal to the support of every invariant state of $\Phi$ \cite{irreddecomp}. In particular, $D=0$ is equivalent to $\Phi$ being a faithful quantum channel (that is, possessing a full-rank invariant state). In general, $\Phi$ may act non-trivially on $P_i \cB(\cH)P_j$ for $i\neq j$. The following proposition shows that this is not the case for PPT maps $\Phi$.

\begin{theorem}\label{cor:pptfulleeb}
    Any faithful $\operatorname{PPT}$ channel $\Phi:\cB(\cH)\to \cB(\cH)$ is the direct sum of irreducible $\operatorname{PPT}$ quantum channels, and therefore is eventually entanglement breaking.
\end{theorem}

\begin{remark}
    This result extends that of \cite{kennedy2017composition} where it was shown that PPT maps are AEB. It also completes Theorem 4.4 of \cite{rahaman2018eventually} where it was shown that any unital PPT channel is EEB.
    Additionally, \Cref{prop:PPT-EB_time} provides a quantitative version of this result.
\end{remark}

\begin{proof}
    As $\Phi$ has an invariant state of full rank, it follows from~\cite{irreddecomp} that we may decompose the identity $\one$ of $\cH$ into maximal subspaces with corresponding orthogonal projections $P_1,\dots,P_n$ s.t. $\Phi$ restricted to $P_i\cB(\cH)P_i$ is irreducible. We will call $P_i\cB(\cH)P_i$ a maximal irreducible component. From \cite[Proposition 5.4]{irreddecomp}, we have that
    \[
        \Phi(P_i \cD(\cH)P_j) \subset P_i \cB(\cH)P_j
    \]
    and hence linearity yields
    \[
        \Phi(P_i \cB(\cH)P_j) \subset P_i \cB(\cH)P_j.
    \]
    Thus, $\Phi$ is $\big(\{P_i\}_{i=1}^n,\{P_i\}_{i=1}^n\big)$-block preserving, and by \Cref{lem:block preserving-PPT},
    \[
        \Phi(P_i \cB(\cH)P_j) = \{0\}
    \]
    for each $i\neq j$. Hence, $\Phi=\bigoplus_{i=1}^n \Phi_i$  for $\Phi_i =\Phi|_{P_i \cB(\cH)P_i}$. Each $\Phi_i$ is irreducible and PPT and hence EEB by \Cref{cor:irred-PPT-EEB}. Thus, $\Phi$ is EEB as well.\qed
\end{proof}

Combining \Cref{cor:pptfulleeb} with that observation that for an irreducible channel $\Phi$ of period $z$, the channel $\Phi^z$ is the direct sum of primitive channels leads to the following structural result.
\begin{theorem} \label{prop:charEEB-discrete}
    Let $\Phi:\cB(\cH)\to \cB(\cH)$ be a faithful quantum channel. The following are equivalent.
    \begin{enumerate}
        \item  $\Phi\in \EEB(\cH)$,
        \item $\Phi^{d_\cH^2 Z}$ is the direct sum of primitive channels,  where $Z$ is the least common multiple of the periods of $\Phi$ restricted to each of its irreducible components,
        \item $\Phi^{n}$ is the direct sum of primitive channels for some $n\leq d_\cH^2 \exp(1.04\,d_\cH)$.
    \end{enumerate}
\end{theorem}
\begin{remark}
    This provides a characterization of $\EEB(\cH)$ in the discrete time case (albeit only for faithful channels), analogous to \Cref{propprimitivecontinuoustime} in the continuous time case.
\end{remark}
\begin{proof}
    Since primitive channels are eventually entanglement breaking\footnote{As argued in the proof of \Cref{thm:EEB_dense} and as shown quantitatively by \Cref{tEB_upperdiscrete}}, the direct sum of primitive channels $\Phi=\bigoplus_i \Phi_i$ is also eventually entanglement breaking, since
    \[
        J(\Phi^n) = \bigoplus_i J(\Phi_i^n)
    \]
    is separable once each $J(\Phi_i^n)$ is separable. Thus, (3) implies (1).

    To prove (2) implies (3), we upper bound $Z$. Since each period is at most $d_\cH$, $Z$ is at most the least common multiple of the natural numbers $1,\dotsc,d_\cH$. By \cite[Theorem 12]{rosser_approximate_1962}, we thus have
    \[
        Z \leq \exp(1.03883\,d_\cH).
    \]

    It remains to show (1) implies (2).
    Let $\Phi \in \EEB(\cH)$ be a faithful quantum channel. Without loss of generality, assume $\Phi$ is not irreducible. Recall from the proof of \Cref{cor:pptfulleeb} that there exist $P_1, \dotsc, P_k\in \cB(\cH)$ distinct orthogonal projections such that $\Phi(P_i\cB(\cH)P_i)\subset\Phi(P_i\cB(\cH)P_i)$ with $\sum_i P_i = \one$, which moreover satisfy
    \[
        \Phi(P_i \cB(\cH)P_j)\subset \Phi(P_i\cB(\cH)P_j)
    \]
    for all $i$ and $j$, and $\Phi_i = \Phi|_{P_i\cB(\cH)P_i}$ is irreducible, with some period $z_i$, and Perron-Frobenius projections $\{p_k^{(i)}\}_{k=0}^{z_i-1}$. Let $Z$ be the least common multiple of the periods $z_i$. Then
    \begin{equation} \label{eq:proof-cov}
        \Phi^Z ( p^{(j)}_i P_j \cB(\cH) P_k p^{(k)}_\ell) \subset  p^{(j)}_i P_j \cB(\cH) P_k p^{(k)}_\ell
    \end{equation}
    for each $i,j,k,\ell$. Moreover, since $\Phi$ is eventually entanglement breaking, $\Phi^{nZ}$ is PPT for some $n\in \mathbb{N}$. Hence, by \Cref{lem:block preserving-PPT} and \eqref{eq:proof-cov},
    \[
        \Phi^{nZ} ( p^{(j)}_i P_j \cB(\cH) P_k p^{(k)}_\ell) = \{0\}
    \]
    unless $j=k$ and $i=\ell$. For $j\neq k$ and $i\neq \ell$, $\Phi^Z|_{ p^{(j)}_i P_j \cB(\cH) P_k p^{(k)}_\ell}$ is therefore a nilpotent linear endomorphism (on the vector space $ p^{(j)}_i P_j \cB(\cH) P_k p^{(k)}_\ell$). Since $\dim( p^{(j)}_i P_j \cB(\cH) P_k p^{(k)}_\ell) \leq d_\cH^2$, we have
    \[
        \Phi^{d^2 Z} ( p^{(j)}_i P_j \cB(\cH) P_k p^{(k)}_\ell) = \{0\}
    \]
    Hence, $\Phi^{d^2 Z} = \bigoplus_{i=1}^n \bigoplus_{j=0}^{z_i-1} \Phi_{ij}$   is the direct sum of primitive quantum channels  $\Phi_{ij} := \left.\Phi^{d^2Z}\right|_{ \cB( p_j P_i \cH)}$.
    \qed
\end{proof}

\subsection{Criteria for AEB channels to be irreducible}\label{convergencedynamics}

In \Cref{sec:irred}, we saw that all irreducible channels are asymptotically entanglement breaking. In this section, we derive a weak converse for this claim and give criteria for when AEB channels are irreducible. Recall the decomposition of the phase space of a quantum channel given in~\eqref{eq:Phi-on-decohere-decomp}.
Below, we show how this result relates to irreducible quantum channels, following the same notation as in~\eqref{eq:Phi-on-decohere-decomp}:
\begin{proposition} \label{prop:char-irred-asymp}
    The following are equivalent:
    \begin{enumerate}
        \item There exists a decomposition \eqref{eq:decohere-decomp} such that $\Phi$ satisfies \eqref{eq:Phi-on-decohere-decomp} with $\cK_0=\{0\}$ and $d_i =1$ for each $i$, and $\pi$ is a $K$-cycle
        \item $\Phi$ is irreducible.
    \end{enumerate}
    Therefore, any asymptotically entanglement breaking channel with corresponding $K$-cyclic permutation and $\cK_0=\{0\}$ is irreducible.
\end{proposition}
\begin{proof}
    First, assume $\Phi$ is irreducible, and adopt the notation of \Cref{prop:decomp-irred}. Then
    \[
        \tilde N(\Phi) =  \linspan \{ u^n \sigma : n = 0, \dotsc,z-1\}\,.\,
    \]
    Let $X\in \tilde \cN(\Phi)$ be given by $X =\sum_{j=0}^{z-1} \lambda_j u^j \sigma$ for some $\lambda_j \in \bC$. Since $\sum_{n=0}^{z-1}p_n = \one$, we have
    \begin{equation} \label{eq:proof-X-irred}
        X = \sum_{n,j=0}^{z-1} \lambda_j  p_n u^j \sigma = \sum_{n,j=0}^{z-1} \lambda_j \theta^{nj} p_n  \sigma = \bigoplus_{n=0}^{z-1} \left(\sum_{j=0}^{z-1} \lambda_j \theta^{nj} \right) \sigma|_{p_n}
    \end{equation}
    where $\sigma|_{p_n}$ is $\sigma$ restricted to the subspace $p_n \cH$, and the direct sum decomposition is with respect to the decomposition $\cH = \bigoplus_{n=0}^{z-1} p_n\cH$.  Note moreover, $\sum_{j=0}^{z-1} \lambda_j \theta^{nj}$ is the $n$th coefficient of the discrete Fourier transform $F_z:\bC^z\to \bC^z$ of the vector $\lambda:=(\lambda_j)_{j=0}^{z-1}$. Since the Fourier transform is invertible, as $\lambda$ ranges over $\bC^z$, the vector of Fourier coefficients range over $\bC^z$ as well. We therefore find
    \[
        \tilde \cN(\Phi) = \bigoplus_{n=0}^{z-1} \bC \,  \sigma|_{p_n}
    \]
    which is a decomposition of the form given by \eqref{eq:decohere-decomp} with $\cK_0=\{0\}$ and $d_n =1$ for each $n$. Moreover, for $X\in \tilde N(\Phi)$ decomposed as $X = \bigoplus_{n=0}^{z-1}  \gamma_n   \sigma|_{p_n} = \sum_{n=0}^{z-1} \gamma_n p_n \sigma p_n$, we have
    \[
        \Phi(X) = \sum_{n=0}^{z-1} \gamma_n \Phi(p_n \sigma p_n) = \sum_{n=0}^{z-1} \gamma_n p_{n-1} \Phi( \sigma)p_{n-1}= \sum_{n=0}^{z-1} \gamma_n p_{n-1} \sigma p_{n-1}
    \]
    by \eqref{eq:irred_on_p_blocks}. Thus, \eqref{eq:Phi-on-decohere-decomp} holds, where $\pi$ is the cyclic permutation $k\to k+1$.

    On the other hand, assume we are given such a decomposition,
    \[
        \tilde \cN(\Phi) = \bigoplus_{i=0}^{K-1} \bC \tau_i
    \]
    for some states $\tau_i \in \cD_+(\cK_i)$ where $\cH = \bigoplus_{i=0}^{K-1} \cK_i$, such that\footnote{If $\pi$ is a cycle, we may reorder the index of the direct sum so that $\pi$ maps $i$ to $i+1$.}
    \[
        \Phi(X) = \bigoplus_{i=0}^{K-1} \gamma_{i+1} \tau_i,\quad \text{ for } \quad X =  \bigoplus_{i=0}^{K-1} \gamma_{i} \tau_i.
    \]
    Define $\sigma := \bigoplus_{i=0}^{K-1} \frac{1}{z}\tau_i$. Then by \eqref{eq:Phi-on-decohere-decomp},
    \[
        \Phi(\sigma) = \bigoplus_{i=0}^{K-1}\frac{1}{z}\tau_i = \sigma
    \]
    as suggested by the name. On the other hand, assume $X\in \cB(\cH)$ is also invariant under $\Phi$. Then $X \in \tilde N(\Phi)$, so $X = \bigoplus_{i=0}^K\lambda_i \tau_i$ for some $\lambda_i \in \bC$. But then
    \[
        \Phi(X) = \bigoplus_{i=0}^K\lambda_{i+1} \tau_i = X \implies \lambda_i = \lambda_{i+1} \, \forall i.
    \]
    Thus, $X$ is proportional to $\sigma$, and $\Phi$ must have a unique invariant state. Since $\sigma$ is additionally faithful (by construction), we conclude that $\Phi$ is irreducible (cf. \Cref{sec:preliminaries} or Theorem 6.4 of \cite{wolf2012quantum}, which shows that a positive map is irreducible if and only if its spectral radius is a non-degenerate eigenvalue with positive-definite left and right eigenvectors).
    \qed
\end{proof}
\medskip

\subsection{Reduction to the trace-preserving case} \label{sec:CPnonTP}

When considering quantum channels as a form of time-evolution, trace-preservation is a natural assumption as an analog of conserving total probability. When considering maps resulting from the inverse Choi isomorphism $J\inv$ applied to bipartite states, however, the map $J\inv(\rho_{A \tilde A})$ is trace-preserving if and only if the first marginal is completely mixed: $\rho_{ A} = \frac{\one_A}{d_A}$. Thus, a priori, results about quantum channels only provide results about a restricted class of bipartite states. However, a large class of CP maps are (up to normalization) similar to CPTP maps. Moreover, this similarity transformation preserves many properties.

In the following, $\spr(\Phi)$ denotes the \emph{spectral radius} of a map $\Phi$, defined as $\spr(\Phi) = \max_{\lambda} |\lambda|$, where $\lambda$ ranges over the eigenvalues of $\Phi$. The spectral radius of a quantum channel is 1.

\begin{proposition} \label{prop:similarity_transform}
    Let $\Phi$ be a $\operatorname{CP}$ map such that for some $X>0$, $\Phi^*(X) = \spr(\Phi)X$ and for some $\sigma > 0$, $\Phi(\sigma) = \spr(\Phi)\sigma$. Then
    \begin{equation} \label{eq:def_Petz}
        \cP_X(\Phi) := \frac{1}{\spr(\Phi)} \Gamma_X \circ \Phi \circ \Gamma_X\inv, \qquad \text{where} \quad \Gamma_X(Y) := X^{1/2}Y X^{1/2}
    \end{equation}
    has the following properties:
    \begin{enumerate}
        \item \label{it:faithfulCPTP} $\cP_X(\Phi)$ is a faithful quantum channel
        \item \label{it:powers} $\cP_X (\Phi^n) = \cP_X(\Phi)^n$, where powers denote repeated composition
        \item \label{it:EB} $\cP_X(\Phi)$ is $\operatorname{EB}$ iff $\Phi$ is $\operatorname{EB}$
        \item \label{it:PPT} $\cP_X(\Phi)$ is $\operatorname{PPT}$ iff $\Phi$ is $\operatorname{PPT}$
    \end{enumerate}
\end{proposition}
\begin{remarks}
    ~\begin{itemize}
        \item This transformation is not novel; the map $\Gamma_{\sigma} \circ \Phi^* \circ \Gamma\inv_{\Phi(\sigma)}$ is known as the Petz recovery map \cite{Pet03}, and is usually considered in the case of CPTP maps $\Phi$. Replacing $\Phi$ with $\Phi^*$ and $\sigma$ with $X$ yields \eqref{eq:def_Petz} in the case that $\Phi^*(X) = \spr(\Phi)X$. This transformation has also been considered in \cite[Appendix A]{HJPR2} in order to study the spectral properties of deformed CP  maps as a function of the deformation. See \cite[Theorem 3.2]{wolf2012quantum} for a similar but slightly different approach (using $\Phi^*(\one)$ instead of $X$).
        \item By the Perron-Frobenius theory (see \cite[Theorem 2.5]{evans_spectral_1978}), if $\Phi$ is a positive map, then the spectral radius $\spr(\Phi)$ is an eigenvalue of $\Phi$, and $\Phi$ admits a positive semi-definite eigenvector $\sigma$, which can be normalized to have unit trace. Since $\spr(\Phi) = \spr(\Phi^*)$, the same logic applied to $\Phi^*$ yields $X\geq 0$ such that $\Phi^*(X) = \spr(\Phi)X$. Hence, the assumption of \Cref{prop:similarity_transform} is that both of these eigenvectors are full rank.
        \item This transformation cannot be applied in general to a CPTP map $\Phi$ in order to obtain a CPTP unital map $\tilde \Phi$, since trace-preservingness will be lost. An intuitive way to see this is that a similarity transformation corresponds to a choice of (non-orthogonal) basis, and by fixing $\Phi^*(\one)=\one$, we choose a particular basis for $\Phi$ in which the dual eigenvector for $\spr(\Phi)$ is represented by the identity matrix. Thus, in general, we cannot simultaneously choose a basis to fix a representation for the eigenvector of $\spr(\Phi)$.
    \end{itemize}
\end{remarks}
\begin{proof}
    \begin{enumerate}
        \item $\cP_X(\Phi)$ is the positive multiple of the composition of CP maps and hence is CP. Since
              \[
                  \cP_X(\Phi)^* = \frac{1}{\spr(\Phi)} \Gamma_X\inv \circ \Phi^* \circ \Gamma_X
              \]
              we have that $\cP_X(\Phi)^*$ is unital:
              \[
                  \cP_X(\Phi)^*(\one) = \frac{1}{\spr(\Phi)}\Gamma_X\inv \circ \Phi^* (X) = \Gamma_X\inv(X) = \one
              \]
              and hence $\cP_X(\Phi)$ is TP. Lastly, $\rho := \frac{\Gamma_X(\sigma)}{\tr[\ldots]}$ is a full-rank invariant state for $\cP_X(\Phi)$: up to normalization,
              \[
                  \cP_X(\Phi)(\rho) =  \frac{1}{\spr(\Phi)} \Gamma_X \circ \Phi \circ \Gamma_X\inv \circ \Gamma_X(\sigma) = \frac{1}{\spr(\Phi)} \Gamma_X \circ \Phi (\sigma) = \Gamma_X(\sigma) = \rho
              \]
              and hence $\cP_X(\Phi)$ is faithful.

        \item We have
              \begin{align*}
                  \cP_X(\Phi^n) & = \frac{1}{\spr(\Phi^n)} \Gamma_X \circ \Phi^n \circ \Gamma_X\inv                                                                                                     \\
                                & = \frac{1}{\spr(\Phi)^n} \Gamma_X \circ \Phi^{n-1} \circ \Gamma_X\inv \circ \Gamma_X \circ \Phi \circ \Gamma_X\inv                                                    \\
                                & = \ldots                                                                                                                                                              \\
                                & = \frac{1}{\spr(\Phi)^n} \Gamma_X \circ \Phi \circ \Gamma_X\inv \circ \Gamma_X \circ \Phi \circ \Gamma_X\inv \circ \dotsm\circ \Gamma_X \circ \Phi \circ \Gamma_X\inv \\
                                & = \cP_X(\Phi)^n.
              \end{align*}
        \item If $\Phi$ is EB, then $\Phi\otimes \id (\cB(\cH_A \otimes \cH_B)_+ ) \subseteq \SEP(A:B)$. Since $X$ is invertible, $\Gamma_X\inv \otimes \id (\cB(\cH_A \otimes \cH_B)_+) = \cB(\cH_A \otimes \cH_B)_+$. Hence $\Phi\circ \Gamma_X\inv \otimes \id$ is EB. Since $\SEP$ is invariant under local positive operations, $\cP_X(\Phi)$ is EB too. The converse follows the same way.
        \item The same proof holds with $\SEP(A:B)$ replaced by $\PPT(A:B)$ and EB replaced by PPT.
    \end{enumerate}
\end{proof}

This transformation allows us to establish our most general result on the entanglement breakingness of PPT maps.

\begin{theorem} \label{thm:CP_PPT_EB_discrete}
    Let $\Phi$ be a $\operatorname{CP}$ map such that for some $X>0$, $\Phi^*(X) = \spr(\Phi)X$ and for some $\sigma > 0$, $\Phi(\sigma) = \spr(\Phi)\sigma$. Then if $\Phi$ is $\operatorname{PPT}$, it is $\operatorname{EEB}$.
\end{theorem}
\begin{proof}
    By (\ref{it:faithfulCPTP}), $\cP_X(\Phi)$ is faithful, CPTP, and by (\ref{it:PPT}) PPT. Hence, for some $n$, $\cP_X(\Phi)^n = \cP_X(\Phi^n)$ is EB, using (\ref{it:powers}) and \Cref{cor:pptfulleeb}. Then $\Phi^n$ is EB by (\ref{it:EB}).
\end{proof}

\section{Finite time properties of discrete time evolutions} \label{sec:finite-time-discrete}
\subsection{Discrete time decoherence-free Poincar\'{e} inequality}\label{strongmixing}
To derive quantitative bounds on entanglement-breaking times for faithful discrete-time evolutions it suffices to restrict to primitive evolutions due to \Cref{prop:charEEB-discrete}, since the entanglement-breaking time of a direct sum of channels is the maximum of the entanglement-breaking times of the component channels. Consider a primitive channel $\Phi$. Since $(\Phi^n\otimes \id)(\Omega) \xrightarrow{n\to\infty} \sigma \otimes \frac{\one}{d_\cH}$, where $\sigma$ is the unique invariant state of $\Phi$, we can use \Cref{lemma1} to show that any state sufficiently close to $\sigma \otimes \frac{\one}{d_\cH}$ is separable. Thus, it remains to establish a mixing time for $\Phi\otimes \id$. The argument given in \eqref{eq:diamond-norm-bound} shows that in fact it suffices to obtain a mixing time for $\Phi$ itself.

However, in the discrete-time case, mixing time bounds even for primitive evolutions have some subtleties. Consider the qubit quantum channel $\Phi$ defined in the Pauli basis as:
\begin{equation*}
    \Phi(\one):=\one, \quad\Phi(\sigma_x):=0,\quad \Phi(\sigma_y):=0, \quad\text{and}\quad \Phi(\sigma_z):=\sigma_x.
\end{equation*}
It is easy to see $\Phi$ defines a valid unital quantum channel by a direct inspection of the Choi matrix. Moreover,  $\Phi^2(X)=\tr(X)\frac{\one}{2}$, which implies that $\Phi$ is primitive (by e.g. \cite[Theorem 6.7]{wolf2012quantum}). One can also easily check that:
\begin{equation*}
    \Phi^*(\one)=\one, \quad\Phi^*(\sigma_x)=\sigma_z,\quad \Phi^*(\sigma_y)=0,\quad \text{and} \quad \Phi^*(\sigma_z)=0.
\end{equation*}
We see that both $\Phi^*\Phi(\one)=\one$ and $\Phi^*\Phi(\sigma_z)=\sigma_z$, implying the second largest singular value (with multiplicity) of $\Phi$ is $1$ and results such as~\cite[Theorem 9]{TKRV10} based on deriving mixing times from the singular values yield trivial estimates for the convergence of the map. This example is revisited in \Cref{rem:algnotcoincideexample}.

In this section, we remedy this by deriving mixing-time bounds which are always non-trivial. In fact, we will remove the primitivity assumption and work directly with general faithful quantum channels. To the best of our knowledge, the only framework that so far allows for convergence statements in the non-primitive case is that of~\cite{Szehr2015}. In that work, the underlying constants that govern the convergence of the quantum channel are a function of the complete spectrum of the channel. This approach does not provide a simple variational formulation of the constants that govern the evolution (such as \eqref{spectralgap}), which is usually vital for applications.
Here, we provide a new bound in the spirit of \cite{TKRV10} by exploiting the relation between the singular values of $\Phi$ and the Poincar\'{e} constant of the generator of a well-chosen continuous time quantum Markov semigroup. Furthermore, we give necessary and sufficient conditions for the distance between the image of the quantum channel and its limit to be strictly contractive in terms of properties of this semigroup.
Thus, we believe that these techniques will find applications outside of our context of estimating when a channel becomes entanglement-breaking.

In \Cref{sec:preliminaries}, we briefly described the general asymptotic structure of quantum channels. In the Heisenberg picture, i.e.~for completely positive, unital channels $\Phi^*$, the situation is similar: assuming that $\cK_0=\{0\}$, the range of the projector $P^*$, which is called the \textit{decoherence-free subalgebra} of $\{\Phi^n\}_{n\in\NN}$ and denoted by $\cN(\Phi^*)$, has the form
\begin{align}\label{Dfalgebra}
    \cN(\Phi^*):=\bigoplus_{i = 1}^K\,\cB(\cH_i)\otimes \mathbb{I}_{\cK_i}\, \quad P^*(X) = \sum_{i=1}^K \tr_{\cK_i}[ p_i X p_i ( \one_{\cH_i} \otimes \tau_i )] \otimes \one_{\cK_i}
\end{align}
where $p_i$ is the projection onto the $i$th term in the direct sum, and each $\tau_i \in \cD(\cK_i)$ is a full rank state (see e.g. \cite[Theorem 34]{lami_entanglement-saving_2016}).
This occurs when $\Phi$ is \textit{faithful} (it possesses a full-rank invariant state).
Recalling \eqref{eq:decohere-decomp}, under the faithfulness assumption, we likewise have
\begin{equation*}
    \tilde{\cN}(\Phi):=\bigoplus_{i=1}^K\,\cB(\cH_i)\otimes \tau_i\,,~~~P(\rho)=\sum_{i=1}^K\tr_{\cK_i}(p_i\rho p_i)\otimes \tau_i\,,
\end{equation*}
in the Schr\"odinger picture,
where $\tilde{\cN}(\Phi)$ is the phase subspace of $\Phi$ and $P$ is the projection onto $\tilde \cN(\Phi)$.
For any such discrete time quantum Markov semigroup, the Jordan decomposition \eqref{Jordandecomp} immediately yields the following convergence result: for any $X\in\cB(\cH)$,
\begin{equation}\label{eq:asymptotic-conv-to-Pstar}
    (\Phi^*)^n(X)-(\Phi^*)^n\circ P^*(X)\to 0~~\text{ as }~~n\to\infty\,.
\end{equation}
Here, the map $P^*$ plays the role of the conditional expectation $E_\cN^*$ associated to quantum Markov semigroup in continuous time discussed in \Cref{sec:preliminaries}. Note, we work in the Heisenberg picture here because $P^*$ projects onto  $\cN(\Phi^*)$ which is an algebra, while $P$ projects onto $\tilde \cN(\Phi)$ which is not.

The state $\sigma_{\tr}:=d_{\cH}^{-1}P(\mathbb{I}) = \frac{1}{d_\cH} \sum_{i=1}^K d_{\cK_i}I_{\cH_i} \otimes \tau_i$ plays a privileged role in the following, and commutes with all of $\cN(\Phi^*)$. We consider the inner product $\braket{X,Y}_{\sigma_{\tr}} = \tr[\Gamma_{\sigma_{\tr}}(X^\dagger) Y]$, where $\Gamma_{\sigma_{\tr}}$ is the superoperator defined as $\Gamma_{\sigma_{\tr}}(X)=\sigma_{\tr}^{1/2}X\sigma_{\tr}^{1/2}$. Note that the adjoint of a Hermitian-preserving linear map $\Psi: \cB(\cH) \to \cB(\cH)$ with respect to $\braket{\cdot,\cdot}_{\sigma_{\tr}}$ is given by $\Gamma_{\sigma_{\tr}}\inv \circ \Psi^* \circ \Gamma_{\sigma_{\tr}}$ where $\Psi^*$ is the adjoint of $\Psi$ with respect to the Hilbert-Schmidt inner product. We then consider $\hat{\Phi}:=\Gamma_{\sigma_{\tr}}^{-1}\circ \Phi\circ \Gamma_{\sigma_{\tr}}$ which is the adjoint of $\Phi^*$ with respect to $\braket{\cdot,\cdot}_{\sigma_{\tr}}$, along with $\hat{P}:=\Gamma_{\sigma_{\tr}}^{-1}\circ P\circ \Gamma_{\sigma_{\tr}}$ which is the adjoint of $P^*$ with respect to  $\braket{\cdot,\cdot}_{\sigma_{\tr}}$. Note that $\hat \Phi$ and $\hat P$ commute. Moreover, for $X \in \cB(\cH)$,
\begin{equation*}
    \hat P(X) = \sum_{i=1}^K \tr_{\cK_i} (  p_i (\mathbb{I}_{\cH_i} \otimes \tau_i^{1/2})X(\mathbb{I}_{\cH_i} \otimes \tau_i^{1/2}) p_i) \otimes \mathbb{I}_{\cK_i} = P^*(X).
\end{equation*}
In other words, $P^*$ is self-adjoint with respect to $\braket{\cdot,\cdot}_{\sigma_{\tr}}$. Note also that $\hat \Phi$ is unital.

Next, the following holds for any $\rho\in\cD(\cH)$ and $X=\sigma_{\tr}^{-1/2}\rho \sigma_{\tr}^{-1/2}$:
\begin{align}\label{fromstatetoobs}
    \|\Phi(\rho)-\Phi\circ P(\rho)\|_1=\|\hat{\Phi}(X)- \hat{\Phi}\circ\hat{P}(X)\|_{1,\sigma_{\tr}}\le \|\hat{\Phi}(X)- \hat{\Phi}\circ\hat{P}(X)\|_{2,\sigma_{\tr}}\,.
\end{align}

The relationship \eqref{fromstatetoobs} allows the convergence in the Schr\"odinger picture to be controlled by the convergence of $\hat \Phi$. This has the following advantages: the norm $\|\cdot\|_{2,\sigma_{\tr}}$ is induced by an inner product (unlike $\|\cdot\|_1$), and $\hat P$ is a conditional expectation as shown by the following lemma (unlike $P$).

\begin{lemma}\label{lemmatechn}
    \item[(i)] $\hat{P}$ is adjoint preserving: for any $X\in\cB(\cH)$, $\hat{P}(X^\dagger)=\hat{P}(X)^\dagger$;
    \item[(ii)] for any $X\in\cB(\cH)$ and $Y,Z\in\cN(\hat{\Phi})$, $\hat{P}(YXZ)=Y\hat{P}(X)Z$;
    \item[(iii)] for any $X\in\cB(\cH)$, $\tr(\sigma_{\tr}X)=\tr(\sigma_{\tr}\hat{P}(X))$;
    \item[(iv)] $\hat{P}$ is self-adjoint with respect to $\sigma_{\tr}$. Moreover, for any $X,Y\in\cB(\cH)$.
    \begin{align*}
        \langle X,\,\hat{P}(Y)\rangle_{\sigma_{\tr}}=   \langle \hat{P}(X),\,Y\rangle_{\sigma_{\tr}}=   \langle \hat{P}(X),\,\hat{P}(Y)\rangle_{\sigma_{\tr}}
    \end{align*}
\end{lemma}

\begin{proof}
    \begin{enumerate}[(i)]
        \item follows from the fact that $P$ itself is adjoint-preserving. This can be seen from the fact that $P = \lim_{N\to\infty} \frac{1}{N}\sum_{n=1}^N \Phi^n$ is the Cesaro mean of $\Phi$, and hence inherits the trace-preserving and complete positivity from $\Phi$ itself (see e.g.~Proposition 6.3 of \cite{wolf2012quantum}).
        \item For $Y,Z\in \cN(\hat\Phi)$, we have the decompositions $Y = \sum_i Y_i \otimes \one_{\cK_i}$ and $Z = \sum_i Z_i \otimes \one_{\cK_i}$. Then,
              \begin{align*}
                  \hat P(YXZ) & = \sum_i \tr_{\cK_i} ( (\mathbb{I}_{\cH_i} \otimes \tau_i) p_iY X Z p_i) \otimes \mathbb{I}_{\cK_i}                                                                                       \\
                              & = \sum_i \tr_{\cK_i} ( (\mathbb{I}_{\cH_i} \otimes \tau_i)  (Y_i \otimes \mathbb{I}_{\cK_i}) p_iXp_i (Z_i \otimes \mathbb{I}_{\cK_i})) \otimes \mathbb{I}_{\cK_i}                         \\
                              & = \sum_i [\tr_{\cK_i} (Y_i (\mathbb{I}_{\cH_i} \otimes \tau_i) p_i X  p_i )Z_i] \otimes \mathbb{I}_{\cK_i}                                                                                \\
                              & = \sum_j Y_j \otimes \mathbb{I}_{\cK_j}\left(\sum_i \tr_{\cK_i} ( (\mathbb{I}_{\cH_i} \otimes \tau_i) p_i X p_i) \otimes \mathbb{I}_{\cK_i} \right) \sum_j Z_j \otimes \mathbb{I}_{\cK_j} \\
                              & =Y \hat P(X) Z.
              \end{align*}

        \item follows from a simple computation: since $\Phi$ is trace preserving, so is $P$ (as discussed in the first point), and hence
              \begin{align*}
                  \tr(\sigma_{\tr}\hat{P}(X))=\tr (P\circ \Gamma_{\sigma_{\tr}}(X))=\tr (\Gamma_{\sigma_{\tr}}(X))=\tr(\sigma_{\tr} X)
              \end{align*}
        \item is a consequence of (i)--(iii):
              \begin{align*}
                  \langle X,\,\hat{P}(Y)\rangle_{\sigma_{\tr}} & =\tr(\sigma_{\tr}^{1/2}X^\dagger\sigma_{\tr}^{1/2}\hat{P}(Y))          \\
                                                               & =\tr(\sigma_{\tr}X^\dagger\hat{P}(Y))                                  \\
                                                               & =\tr(\sigma_{\tr}\hat{P}(X^\dagger\hat{P}(Y)))                         \\
                                                               & =\tr(\sigma_{\tr}\hat{P}(X^\dagger)\hat{P}(Y))                         \\
                                                               & =\tr(\sigma_{\tr}\hat{P}(X)^\dagger\hat{P}(Y))                         \\
                                                               & =\tr(\sigma_{\tr}^{1/2}\hat{P}(X)^\dagger\sigma_{\tr}^{1/2}\hat{P}(Y)) \\
                                                               & =\langle \hat{P}(X),\,\hat{P}(Y)\rangle_{\sigma_{\tr}}
              \end{align*}
              where we used the commutativity of $\sigma_{\tr}$ with $\cN(\hat{\Phi})$ in the second and sixth lines, (i) in the fifth line, (ii) in the fourth line and (iii) in the third line. The first identity in (iv) follows by symmetry. \qed
    \end{enumerate}
\end{proof}
\medskip
Next, we show how to obtain the strong decoherence constants of a quantum Markov chain $\{\Phi^n\}_{n=1}^\infty$ in terms of the Poincar\'e constant of a semigroup defined via $\Phi$. However, in the discrete time case the situation is a bit more subtle, as discrete-time semigroups are not necessarily strictly contractive. That is, there might be $X\not=\hat{P}(X)$ such that
\begin{align*}
    \|\hat{\Phi}(X)-\hat{\Phi}\circ \hat{P}(X)\|_{2,\sigma_{\tr}}=\|X-\hat{P}(X)\|_{2,\sigma_{\tr}}.
\end{align*}
Thus, before showing how to obtain functional inequalities that characterize the strict contraction of the channel, we first discuss some necessary conditions for the contraction to occur.
\paragraph{Cycles in the phase algebra:} Here, we show that quantum channels might not be contractive for times not proportional to the order of the permutation in the phase algebra.
More precisely, recall that for $X=P(X)=\bigoplus_i X_i\otimes \tau_i\oplus0_{\cK_0}$, we have
\begin{equation}
    \Phi(X) = \bigoplus_{i=1}^K\, U_i X_{\pi(i)} U_i^\dagger \otimes \tau_i\oplus 0_{\cK_0}.
\end{equation}
for some permutation $\pi$. We will now construct quantum channels which are not contractive for times smaller than the largest cycle in $\pi$.
To this end, define in a fixed orthonormal basis $\{|i\rangle\}_{i=0}^{d-1}$ of $\CC^d$ the matrix $A\in\M_d$ as
\begin{align*}
    A:=\sum\limits_{i,j=0}^{d-2}\ketbra{i}{j}+\ketbra{d-1}{d-1}+\epsilon\left(\ketbra{d-1}{d-2}+\ketbra{d-2}{d-1}\right).
\end{align*}
for some $0<\epsilon<1$ and let $\pi\in S_d$ be the $d-$cycle $(0\,1\,2\,\cdots\,d-1)$. $A$ is clearly a positive matrix with $1$-s on its diagonal. Next, we define the quantum channel $\Phi_{\pi,A}(X):=U_\pi\lb A\circ X\rb U_\pi^\dagger$ where $U_\pi:\,|i\rangle\mapsto|\pi(i)\rangle$ is a unitary representation of $\pi$ and $\circ$ denotes the Hadamard product.

For this quantum channel, $P=\hat{P}$ can be easily shown to be the projection onto the algebra of matrices that are diagonal in the basis $\{|i\rangle\}_{i=0}^{d-1}$. Moreover, $\Phi_{\pi,A}$ just acts as a permutation on this algebra.
Now, let $T=\sum_{i,j=0}^1\ketbra{i}{j}$. We have for $1\leq k< d-2$ and $m\in\mathbb{N}$:
\begin{align*}
    \Phi^{md+k}_{A,\pi}(T)=\ketbra{k}{k}+\ketbra{k+1}{k+1}+\epsilon^m\left(\ketbra{k+1}{k}+\ketbra{k}{k+1}\right),\quad \hat{P}(T)= \ketbra{0}{0}+\ketbra{1}{1}.
\end{align*}
Clearly, since $\sigma_{\tr}=d^{-1}\,\mathbb{I}_{\CC^d}$, $$\|\Phi^{md+k-1}_{A,\pi}(T)-P\circ \Phi^{md+k}_{A,\pi}(T)\|_{2,\sigma_{\tr}}=\|\Phi^{md+k}_{A,\pi}(T)-P\circ \Phi^{md+k+1}_{A,\pi}(T)\|_{2,\sigma_{\tr}}\,,$$ which shows that these channels are not contractive at intermediate steps.

\paragraph{Contraction properties through decoherence-free algebras:} Surprisingly, there exists a link between the strict contraction of a quantum channel $\Phi$ and whether its decoherence-free subalgebra coincides with the decoherence-free subalgebra of the quantum dynamical semigroup with generator $\cL^*:=\Phi^*\circ\hat{\Phi}-\id$. First, we show the following inclusion of algebras:

\begin{proposition}
    Given a quantum channel $\Phi:\cB(\cH)\to\cB(\cH)$ whose action on its phase space is purely unitary (i.e. contains no cycles), consider the continuous time quantum Markov semigroup $(\Phi_t)_{t\ge 0}$ whose corresponding generator is given (in the Heisenberg picture) by $\cL^*=\Phi^*\circ\hat{\Phi}-\id$. Then $\cN(\hat{\Phi})\subset \cN((\Phi_t^*)_{t\ge 0})$.
\end{proposition}

\begin{proof} Assume that $\Phi^*\circ\hat{\Phi}(X)\not=X$, but $X\in \cN( \hat \Phi )$. As discussed before, $\hat{\Phi}$ acts unitarily on $X$. Thus,
    \begin{align}\label{equ:unitaryonx}
        \langle \hat{\Phi}(X),\,\hat{\Phi}(X)\rangle_{\sigma_{\tr}}= \langle X,X\rangle_{\sigma_{\tr}}.
    \end{align}
    Now, note that $\Phi^*\circ \,\hat{\Phi}$ is a positive operator of norm $1$ w.r.t.~ $\langle \cdot,\cdot\rangle_{\sigma_{\tr}}$\footnote{By this we do not mean that it maps positive matrices to positive matrices, but rather that is a positive semidefinite operator of operator norm $1$ when seen as a linear operator between the Hilbert space $\M_d$ with scalar product given by $\langle \cdot,\cdot\rangle_{\sigma_{\tr}}$}. Thus, we may decompose $X$ as $X=a_1X_1+a_2X_2$, where $X_1$ and $X_2$ are orthogonal matrices w.r.t. the weighted scalar product and of unit norm, where $\Phi^*\circ \,\hat{\Phi}(X_1)=X_1$ and $\Phi^*\circ\,\hat{\Phi}(X_2)=X'_2$ with $\|X_2'\|_{2,\sigma_{\tr}}<1$ and such that $X'_2$ is orthogonal to $X_1$. This can be achieved by picking $X_1$ in the eigenspace corresponding to $1$ and $X_2$ in the orthogonal eigenspace where $\Phi^*\circ \,\hat{\Phi}$ is strictly contractive, i.e. corresponding to eigenvalues $<1$. It is then easy to see that
    \begin{align*}
        \langle \Phi^*\circ\,\hat{\Phi}(X),X\rangle_{\sigma_{\tr}}=a_1^2+a_2^2\langle X'_2,X_2\rangle_{\sigma_{\tr}}<a_1^2+a_2^2=\langle X,X\rangle_{\sigma_{\tr}},
    \end{align*}
    which contradicts equation~\eqref{equ:unitaryonx}.

    \qed
\end{proof}

\begin{remark}\label{rem:algnotcoincideexample}
    Note that in general we have that $\cN( \hat \Phi )\subset\cN(( \Phi_t^* )_{t\ge 0})$ is a strict inclusion, even if the action of the channel $\Phi$ on the phase space does not contain any cycle. As an example, consider the qubit quantum channel $\Phi$ defined at the start of the section. Recall
    \begin{align*}
        \Phi(\one)=\one, \,\Phi(\sigma_x)=0,\, \Phi(\sigma_y)=0,\, \Phi(\sigma_z)=\sigma_x \\
        \Phi^*(\one)=\one, \,\Phi^*(\sigma_x)=\sigma_z,\, \Phi^*(\sigma_y)=0,\, \Phi^*(\sigma_z)=0,
    \end{align*}
    and that $\Phi$ is a primitive unital quantum channel such that $\Phi^2(X)=\tr(X)\frac{\one}{2}$.
    Since $\sigma_{\tr}=\mathbb{I}/2$, $\hat{\Phi}=\Phi$ and $\Phi^*\circ\,\Phi(\sigma_z)=\sigma_z$, we have $\sigma_z\in\cN( (\Phi_t^*)_{t\ge0} )$ but $\sigma_z\not\in\cN( \hat \Phi )=\CC\mathbb{I}$. Not surprisingly, it can be checked that the
    state $\frac{1}{2}\left(\one+\sigma_z\right)$ does not contract under this channel.
\end{remark}

\begin{proposition}
    Given a quantum channel $\Phi:\cB(\cH)\to\cB(\cH)$ whose action on its phase space is purely unitary (i.e. contains no cycles), consider the continuous time quantum Markov semigroup $(\Phi_t)_{t\ge 0}$ whose corresponding generator is given (in the Heisenberg picture) by $\cL^*=\Phi^*\circ\hat{\Phi}-\id$ and the state $\sigma_{\tr}=d_{\cH}^{-1}\,P(\mathbb{I}_{\cH})$. If for all $X\notin\cN(\hat{\Phi})$
    \begin{align}\label{discpoinc}
        \|\hat{\Phi}(X)-\hat{\Phi}\circ \hat{P}(X)\|_{2,\sigma_{\tr}}^2<\|X-\hat{P}(X)\|_{2,\sigma_{\tr}}^2,
    \end{align}
    then $\cN(( \Phi_t^*)_{t\ge 0} ) = \cN( \hat \Phi )$.
\end{proposition}
\begin{proof}
    Suppose that there exists $X$ such that $X\in\cN( (\Phi_t^*)_{t\ge 0} )$ but $X\not\in \cN( \hat \Phi )$. This is equivalent to
    \begin{align*}
        \Phi^*\circ\hat{\Phi}\,(X)=X,\quad \hat{P}(X)\not=X.
    \end{align*}
    We will now show that
    \begin{align}\label{equ:equalitynorms}
        \|\hat{\Phi}(X)-\hat{\Phi}\circ \hat{P}(X)\|_{2,\sigma_{\tr}}^2=\|X-\hat{P}(X)\|_{2,\sigma_{\tr}}^2.
    \end{align}
    First, note that
    \begin{align*}
        \|\hat{\Phi}(X) & -\hat{\Phi}\circ \hat{P}(X)\|_{2,\sigma_{\tr}}^2 \\
                        & =
        \langle \hat{\Phi}(X),\,\hat{\Phi}(X)\rangle_{\sigma_{\tr}}+\langle \hat{\Phi}\circ\hat{P}(X),\,\hat{\Phi}\circ\hat{P}(X)\rangle_{\sigma_{\tr}}
        -2\operatorname{Re}\langle \hat{\Phi}(X),\,\hat{\Phi}\circ\hat{P}(X)\rangle_{\sigma_{\tr}}.
    \end{align*}
    From $\Phi^*\circ\,\hat{\Phi}(X)=X$, it follows that $\langle \hat{\Phi}(X),\,\hat{\Phi}(X)\rangle_{\sigma_{\tr}}=\langle X,X\rangle_{\sigma_{\tr}}$ and $\langle \hat{\Phi}(X),\,\hat{\Phi}\,\circ\,\hat{P}(X)\rangle_{\sigma_{\tr}}=\langle X,\,\hat{P}(X)\rangle_{\sigma_{\tr}}$. Similarly, we have $\langle \hat{\Phi}\circ\hat{P}(X),\,\hat{\Phi}\circ\hat{P}(X)\rangle_{\sigma_{\tr}}=\langle \hat{P}(X),\,\hat{P}(X)\rangle_{\sigma_{\tr}}$, since $\hat{\Phi}$ acts unitarily on the image of $\hat{P}$ with a corresponding unitary in $\cN(\Phi)$ which commutes with $\sigma_{\tr}$. It is then easy to see that~\eqref{equ:equalitynorms} follows and $\hat{P}$ is not a strict contraction.

    \qed
\end{proof}

The previous discussion suggests that, in general, it is necessary to wait until the algebras coincide and the quantum channel has no cycle to ensure strict contraction.
The next theorem, which generalizes Theorem 9 of \cite{TKRV10} to the non-primitive case, shows that this is also sufficient.
\begin{theorem} \label{prop:discrete-Poincare}
    Given a discrete time quantum Markov semigroup $\{\Phi^n\}_{n\in\NN}$ whose transition map $\Phi$ acts purely unitarily on its phase space (i.e. contains no cycles), consider the continuous time quantum Markov semigroup $(\Phi_t)_{t\ge 0}$ whose corresponding generator is given (in the Heisenberg picture) by $\cL^*=\Phi^*\circ\hat{\Phi}-\id$ and the state $\sigma_{\tr}:=d_\cH^{-1}\,P(\mathbb{I})$. Then
    \begin{itemize}
        \item[(i)]   If  $\cN(( \Phi_t^*)_{t\ge 0} ) = \cN( \hat \Phi )$, then:
              \begin{align}\label{discpoinc1}
                  \|\hat{\Phi}(X)-\hat{\Phi}\circ \hat{P}(X)\|_{2,\sigma_{\tr}}^2\le (1-\lambda(\cL^*))\|X-\hat{P}(X)\|_{2,\sigma_{\tr}}^2
              \end{align}
              Moreover, for any $\rho\in\cD(\cH)$ and any $n\in\NN$:
              \begin{align}\label{sdpoincdisc}
                  \| \Phi^n(\rho)-\Phi^n\circ P(\rho)\|_{1}\le\,\sqrt{\|  \sigma_{\tr}^{-1}\|_\infty} \,(1-\lambda(\cL^*))^{\frac{n}{2}}\,.
              \end{align}
              where $1-\lambda(\cL^*) < 1$ corresponds to the second largest singular value of $\Phi$, counted without multiplicities.
        \item[(ii)]  More generally, let $(\Phi^{(k)}_{t})_{t\ge 0}$ be the continuous time quantum Markov semigroup whose corresponding generator is given (in the Heisenberg picture) by $\cL_k^*=\lb\Phi^*\rb^k\circ(\hat{\Phi})^k-\id$. Then, for any $k$ larger than the size $m$ of the largest Jordan block of $\Phi$, $     \cN((( \Phi^{(k)}_t)^*)_{t\ge 0} ) = \cN( \hat \Phi )$ and
              \begin{align}\label{discpoinck}
                  \|\hat{\Phi}^k(X)-\hat{\Phi}^k\circ \hat{P}(X)\|_{2,\sigma_{\tr}}^2\le (1-\lambda(\cL_k^*))\|X-\hat{P}(X)\|_{2,\sigma_{\tr}}^2
              \end{align}
              Then, for any $\rho\in\cD(\cH)$ and any $n\in\NN$:
              \begin{align}\label{sdpoincdisck}
                  \| \Phi^{nk}(\rho)-\Phi^{nk}\circ P(\rho)\|_{1}\le\,\sqrt{\|  \sigma_{\tr}^{-1}\|_\infty} \,(1-\lambda(\cL_k^*))^{\frac{n}{2}}\,.
              \end{align}
              Moreover, there are quantum channels for which the minimal integer $k$ such that the decoherence-free algebras coincide is given by $m$.
    \end{itemize}
    Finally, the condition of non-cyclicity of the action of $\Phi$ on its phase space can be removed by choosing $k$ to be proportional to the least common denominator of the cycles of the semigroup.
\end{theorem}

\begin{proof}
    (i) That $1-\lambda(\cL^*)$ is the second largest singular value of $\Phi$ follows from the definition of the semigroup and of its Poincar\' e constant. This is because $\cL^*$ is, by definition, a s.a. operator with spectrum in $[-2,0]$. In this case, the conditional expectation w.r.t. $\cL^*$ is just the projection onto the eigenspace corresponding to the eigenvalue $1$ of $\Phi^*\circ\hat{\Phi}$, which just correspond to right singular vectors of the singular value $1$ of $\Phi^*$. An application of the min-max theorem for eigenvalues then yields the claim.
    Then, notice that
    \begin{align*}
        \|\hat{\Phi}(X)-\hat{\Phi}\circ \hat{P}(X)\|_{2,\sigma_{\tr}}^2 & =\|\hat{\Phi}(X)\|_{2,\sigma_{\tr}}^2+\|\hat{\Phi}\circ\hat{P}(X)\|_{2,\sigma_{\tr}}^2-2\operatorname{Re}\langle \hat{\Phi}(X),\,\hat{\Phi}\circ\hat{P}(X)\rangle_{\sigma_{\tr}} \\
                                                                        & =\|\hat{\Phi}(X)\|_{2,\sigma_{\tr}}^2-\|\hat{\Phi}\circ\hat{P}(X)\|_{2,\sigma_{\tr}}^2\,,
    \end{align*}
    where we simply used $(iv)$ of \Cref{lemmatechn} and the commutativity of $\hat{P}$ with $\hat{\Phi}$ for the second identity. One can similarly prove that $\|X-\hat{P}(X)\|_{2,\sigma_{\tr}}^2=\|X\|_{2,\sigma_{\tr}}^2-\|\hat{P}(X)\|_{2,\sigma_{\tr}}^2$. Next, the conditional expectation $E_\cN^*$ associated with the continuous time quantum Markov semigroup $(\Phi_t^*)_{t\ge 0}$ of generator $\cL^*$ is equal to $\hat{P}$. To see this, note that $\sigma_{\tr}$ (defined by the discrete time QMS) is invariant under the semigroup, since $\cL(\sigma_{\tr}) =  (\hat \Phi)^* \circ \Phi(\sigma_{\tr}) - \sigma_{\tr} = (\hat \Phi)^*(\sigma_{\tr}) - \sigma_{\tr} =0$, and therefore is invariant under $E_\cN^*$. So both $\hat P$ and $E_\cN^*$ are conditional expectations onto the same algebra $\cN((\Phi_t^*)_{t\geq 0}) = \cN(\hat \Phi)$ which both preserve $\sigma_{\tr}$. Next, \cite[Theorem 9]{carbone_decoherence_2013} shows that there exists a unique conditional expectation onto $\cN(\hat \Phi)$ which preserves a given faithful invariant state, and consequently we have $\hat P  = E_\cN^*$.
    Next,
    $\lambda(\cL^*)$ is the largest positive number $\lambda$ such that
    \begin{align}\label{poincarecontodisc}
        \lambda\,\|X-\hat{P}(X)\|_{2,\sigma_{\tr}}^2\le -\langle X,\,(\Phi^*\circ\,\hat{\Phi}-\id)(X)\rangle_{\sigma_{\tr}}\,.
    \end{align}
    By a simple computation, one shows that the right hand side of the above inequality is equal to $\|X\|_{2,\sigma_{\tr}}^2-\|\hat{\Phi}(X)\|^2_{2,\sigma_{\tr}}$. Similarly, the norm on the left hand side of \Cref{poincarecontodisc} is equal to $\|X\|_{2,\sigma_{\tr}}^2-\|\hat{P}(X)\|_{2,\sigma_{\tr}}^2$.
    Hence, (\ref{poincarecontodisc}) is equivalent to
    \begin{align*}
        \lambda(\cL)\,( \|X\|_{2,\sigma_{\tr}}^2- \|\hat{\Phi}\circ \hat{P}(X)\|_{2,\sigma_{\tr}}^2  )\le \|X\|_{2,\sigma_{\tr}}^2-\| \hat{\Phi}(X)\|_{2,\sigma_{\tr}}^2,
    \end{align*}
    which is itself equivalent to (\ref{discpoinc1}). A simple iteration procedure of (\ref{poincarecontodisc}) together with (\ref{fromstatetoobs}) leads to (\ref{sdpoincdisc}) after observing that
    \begin{align*}
        \|X-\hat{P}(X)\|_{2,\sigma_{\tr}}^2 & =\|X-\tr(\sigma_{\tr}X)\,\mathbb{I}_\cH\|_{2,\sigma_{\tr}}^2-\| \hat{P}(X)-\tr(\sigma_{\tr}X)\,\mathbb{I}_\cH \|_{2,\sigma_{\tr}}^2 \\
                                            & \le \|X-\tr(\sigma_{\tr}X)\,\mathbb{I}_\cH\|_{2,\sigma_{\tr}}^2                                                                     \\
                                            & =\| X\|_{2,\sigma_{\tr}}^2-1                                                                                                        \\
                                            & \le \tr(\,\rho\,\sigma_{\tr}^{-1/2}\,\rho\,\sigma_{\tr}^{-1/2})                                                                     \\
                                            & \le \|\sigma_{\tr}^{-1}\|_\infty\,.
    \end{align*}

    (ii) We apply a Schur decomposition of the quantum channel $\hat{\Phi}=U^*TU$, where $U$ is unitary w.r.t. to  $\langle \cdot,\,\cdot\rangle_{\sigma_{\tr}}$ and $T$ is upper triangular. W.l.o.g. assume that the eigenvalues corresponding to the peripheral spectrum are on the diagonals $T_{i,i}$ for $1\leq i\leq r$, where $r$ is the size of the peripheral spectrum. As the action of the channel is unitary on that subspace, the first $r$ columns and rows of  the matrix $T$ must be orthonormal. But their diagonal entries already have norm $1$, which shows that we may decompose the matrix $T$ into a direct sum $T=D_{phase}\oplus N$, where $D_{phase}$ is diagonal and acts on the phase space of the channel, while $N$ is upper triangular and has eigenvalues $<1$ in modulus. Now, note that the condition $ \cN((( \Phi^{(k)}_t)^*)_{t\ge 0} ) = \cN( \hat \Phi )$ is equivalent to the span of right singular vectors corresponding to the singular value $1$ of $\hat{\Phi}^k$ being equal to the phase space of $\Phi$.
    It then follows from the decomposition of $T$ that the two algebras coinciding corresponds to the matrix $N^k$ only having singular values strictly smaller than $1$. This is equivalent to $\|N^k\|_\infty<1$, where by $\|.\|_\infty$ we mean the operator norm of the matrix. Clearly, for $k\geq m$ the matrix $N^k$ is diagonalizable and, thus, the operator norm is bounded by the largest eigenvalue of  $N^k$ in modulus. As $N$ has spectral radius $<1$, it follows that $\|N^k\|_\infty<1$.
    Finally, observe that the example given in remark~\ref{rem:algnotcoincideexample} saturates this bound. (\ref{discpoinck}) and (\ref{sdpoincdisck}) follow similarly to (i).
    \qed

\end{proof}

\medskip

Note that Theorem~\ref{prop:discrete-Poincare} always ensures a strict contraction of the quantum channel, unlike previous results in the literature, as the second largest singular value of a quantum channel ignoring multiplicities is always strictly smaller than $1$.
Moreover, in the case of quantum channels  $\Phi$ satisfying detailed balance, it follows from the last proposition that $\Phi^2$ is strictly contractive. To see why this is the case, not that these channels are always diagonalizable and have a real spectrum. Thus, the only possible eigenvalues for the peripheral spectrum are $1$ and $-1$ and we can only have cycles of length $2$.

\begin{remark}
    As in the continuous case, it is expected that is possible to further improve these convergence bounds by considering discrete versions of relative entropy convergence.
    In the discrete case, strict contraction in relative entropy is related to the notion of a \textit{strong data processing inequality}. In certain situations, the corresponding strong data processing constant can be related to the so-called logarithmic Sobolev inequality of order 2 (see \cite{Raginsky2013,Miclo97,MHF16,MullerHermes2018sandwichedrenyi}), although all the results known in the quantum case only cover primitive evolutions. Also in this setting it is possible to connect the  contraction of the quantum channel to the contraction of the quantum semigroup in continuous-time considered here.
\end{remark}

\subsection{Entanglement breaking times for discrete time evolutions}\label{EBTdiscrete}

Similarly to the continuous time case, we define the following entanglement-loss times: the \textit{entanglement-breaking time} (or \emph{entanglement-breaking index}, introduced in~\cite{Lami2015eb})
\begin{equation*}
    n_{\EB}\p \{\Phi^n\}_{n\in\NN})\overset{\operatorname{def}}{=}\min\left\{n\in \NN: \Phi^n   \in \EB\p\cH)\right\}\,.
\end{equation*}
Similarly, given a discrete time quantum Markov semigroup $\{\Gamma^n\}_{n=1}^\infty$ over a bipartite Hilbert space $\cH_A\otimes\cH_B$, we define the \textit{entanglement annihilation time} $n_{\text{EA}}(\Gamma)$ as follows
\begin{align*}
    n_{\text{EA}}( \{\Gamma^n\}_{n\in\NN})\overset{\operatorname{def}}{=}\min\left\{n\in\NN:\Gamma^n\in\operatorname{EA}(\cH_A,\cH_B)\right\}.
\end{align*}
In the case when $\Gamma^n=\Phi^n\otimes\Phi^n$, $\Phi^n:\,\cB(\cH)\to\cB(\cH)$, this time is called the $2$-\textit{local entanglement annihilation time}, and is denoted by
\begin{align*}
    n_{\text{LEA}_2}(\{\Phi^n\}_{n\in\NN})\overset{\operatorname{def}}{=}\min\left\{n\in\NN:\Phi^n\in\operatorname{LEA}_2(\cH)\right\}\,.
\end{align*}
In analogy with \Cref{strongdeco}, we say that a discrete time quantum Markov semigroup $\{\Phi^n\}_{n\in \NN}$ on $\cB(\cH)$ satisfies the so-called \textit{discrete strong decoherence property} if there exist $k\in\NN$ and constants $\tilde{K}>0$ and $\tilde{\gamma}>1$, possibly depending on $d_\cH$, such that for any initial state $\rho\in\cD(\cH)$ and any $n\in\NN$:
\begin{align}\label{SDPdiscrete}\tag{$\operatorname{dSD}$}
    \|\Phi^{nk}(\rho-P(\rho))\|_1\le \,\tilde{K}\,\,\tilde{\gamma}^{-n}\,.
\end{align}
This property was shown to hold for faithful channels in \Cref{sdpoincdisck} with $k$ the size of the largest Jordan block of $\Phi$,  $\tilde{K}=\sqrt{\|\sigma_{\tr}^{-1}\|_{\infty}}$ and $\tilde{\gamma}=(1-\lambda(\cL_k^*))^{-\frac{1}{2}}$. It is now straightforward to adapt the results from the last sections to obtain bounds on when discrete time quantum Markov semigroups become entanglement breaking.
\begin{theorem}\label{tEB_upperdiscrete}
    For a quantum channel $\Phi:\cB(\cH) \to \cB(\cH)$ that is primitive with full-rank invariant state $\sigma$, we have
    \begin{align*}
        n_{\operatorname{EB}}(\{\Phi^n\}_{n\in\NN})<\frac{3\,k\,\log (d_{\cH}\|\sigma^{-1}\|_\infty)}{-\log(1-\lambda(\cL_k^*))}\,,
    \end{align*}
    where $k$ is the size of the largest Jordan block of $\Phi$ and $\lambda(\cL^*_k)$ is the spectral gap of $\cL_k^*:=(\Phi^*)^k\circ\,(\hat{\Phi})^k-\id$.
\end{theorem}
\begin{proof}
    The proof follows the same reasoning as in \Cref{tEB_upper}, simply using \Cref{prop:discrete-Poincare} instead of \eqref{poincareDF} and noting that primitive quantum channels do not have cycles.
    \qed
\end{proof}
\begin{remark} \label{rem:prim-eb-bound}
    ~\begin{enumerate}
        \item   In the system studied in \Cref{ex:RWA}, for $|\gamma|<1$, the map $\Phi$ is primitive, with $\sigma = \rho_{\beta^*}$, and $\|\sigma\inv\|_\infty = 1+g$, while $1- \lambda(\cL^*) = |\gamma|^2$. Moreover, $\Phi$ has four distinct eigenvalues, and in particular, every Jordan block is of size 1. Hence,
              \[
                  n_{\operatorname{EB}}(\{\Phi^n\}_{n\in\NN})<\frac{3\,\log (2(1+g))}{-2\log|\gamma|} \leq \frac{3 \log (3)}{-2\log|\gamma|}\,,
              \]
              which matches the $\log((|\gamma|\inv))\inv$ scaling of the exact entanglement breaking time computed in  \Cref{ex:RWA}.
        \item The bound given by \Cref{tEB_upperdiscrete} together with \Cref{prop:charEEB-discrete} yields a method to check whether a faithful channel $\Phi$ is EEB, and if so, to compute a bound on the entanglement-breaking time. One first computes $Z$, the least common multiple of the periods of $\Phi$ restricted to its irreducible components, and checks if $\Phi^{d_\cH^2Z}$ is the direct sum of primitive channels, $\Phi^{d_\cH^2Z} = \bigoplus_i \Phi_i$. If not, $\Phi\not\in \EEB(\cH)$. Otherwise, \Cref{tEB_upperdiscrete} yields an upper bound $n_i$ on the entanglement-breaking time for each $\Phi_i$. Then $n_\text{EB}(\Phi) \leq d_\cH^2Z \max_i n_i$ is  a bound on the entanglement-breaking time of $\Phi$.
    \end{enumerate}

\end{remark}
\medskip
We now derive some lower bounds on the time it takes to a bipartite channel to become $2$-locally entanglement annihilating. Next proposition is a simple consequence of \Cref{cor:n-shot-LEA2}. Note that the lower bound given here is also a lower bound on the time it takes
for a discrete time quantum Markov semigroup to become entanglement breaking (cf. \Cref{EBinLEA2}).
\begin{proposition}[Lower bound for $n_{\operatorname{LEA}_2}$] \label{tANN_lower}

    For a quantum channel $\Phi:\cB(\cH) \to \cB(\cH)$ with $\det(\Phi)\not=0$ which is reversible with respect to a faithful state $\sigma$,
    \begin{align*}
        n_{\operatorname{LEA}_2}(\{\Phi^n\}_{n\in\NN}) \geq \frac{\log(d_\cH \frac{\lambda_{\min}(\sigma)}{\|\sigma\|_\infty})}{\log(\|(\Phi^*)^{-1}\|_{2,\sigma\to2,\sigma})}.
    \end{align*}
\end{proposition}
\begin{proof}
    Note that we have for any linear operator $\Lambda$ that
    \begin{align*}
        \|\Lambda\|_{2\to 2}=   \|\Lambda^*\|_{2\to2}\leq\left(\frac{\|\sigma\|_\infty}{\lambda_{\min}(\sigma)}\right)\|\Lambda^*\|_{2,\sigma\to2,\sigma}
    \end{align*}
    As $\Phi^*$ is reversible with respect to $\sigma$, we have that $\|(\Phi^*)^{-k}\|_{2,\sigma\to2,\sigma}=\|(\Phi^*)^{-1}\|^k_{2,\sigma\to2,\sigma}$, which is
    just the inverse of the smallest eigenvalue of $\Phi^*$. The claim then follows from \Cref{cor:n-shot-LEA2}.
    \qed
\end{proof}
Note that a reversible continuous time quantum Markov semigroup $(\Phi_t)_{t\ge 0}$ with generators $\mathcal{L}$ is invertible and we have $\|\Phi_t^{-1}\|_{2,\sigma\to2,\sigma}=e^{\,t\|\mathcal{L}\|_{2,\sigma\to2,\sigma}}$, so that \Cref{tANN_lower} can be adapted to cover this case.
\begin{corollary} \label{tANN_lower_arb}
    Let $\Phi_i:\cB(\cH) \to \cB(\cH)$ be a quantum channel for each $i=1,\dotsc, n$ such that $\ell := \min_{i\in[n]} |\det \Phi_i|  > 0$. Note $\ell \leq 1$. If
    \begin{equation}\label{n_EA_bound}
        n < \frac{d_\cH^2 \log d_\cH}{2 \log ( \ell\inv )}
    \end{equation}
    then $\Phi_1\circ \dotsm \circ \Phi_n$ is not $\operatorname{EA}$.
\end{corollary}
\begin{proof}
    By the arithmetic-geometric mean inequality,  for any linear map $\Phi$,
    \begin{equation}
        \| J(\Phi)\|_2^2 \geq d_\cH^2 |\det \Phi|^{\frac{2}{d_\cH^2}}.
    \end{equation}
    Applying this to $\Phi_1\circ \dotsm \circ \Phi_n$, and using the multiplicativity of the determinant,
    \begin{align}
        \| J(\Phi_1\circ \dotsm \circ \Phi_n)\|_2^2 & \geq \,d_\cH^2(|\det \Phi_1| \dotsm |\det \Phi_n|)^{\frac{2}{d_\cH^2}} \\
                                                    & \geq d_\cH^2 \,\ell^{\frac{2n}{d_\cH^2}}
    \end{align}
    Using that $\|J(\Phi_1\circ \dotsm \circ \Phi_n)\|_2^2 > d_\cH$ implies $\Phi_1\circ \dotsm \circ \Phi_n$ is not EA (cf. \Cref{lem:1-shot-LEA2}), we rearrange to find the result.
\end{proof}
\begin{remark}
    In the case $\ell =1$, each $\Phi_i$ is unitary, and the composition $\Phi_1\dotsm \Phi_n$ is not EA for any $n$, which agrees with interpreting $\frac{1}{0}=\infty$ in \eqref{n_EA_bound}.
\end{remark}

We saw as a consequence of \Cref{cor:pptfulleeb} that PPT channels with a full rank state are eventually entanglement breaking. In the following proposition, we provide a quantitative version of that statement.
\begin{theorem} \label{prop:PPT-EB_time}
    Let $\Phi:\cB(\cH)\to \cB(\cH)$ be a $\operatorname{PPT}$ quantum channel with an invariant full rank state $\sigma$. Assuming $\operatorname{(}$\ref{SDPdiscrete}$\operatorname{)}$ holds for $\{\Phi^n\otimes\id\}_{n\in\NN}$,
    \begin{align*}
        n_{\operatorname{EB}}(\{\Phi^n\}_{n\in\NN})\le  \frac{\,\log\left( d_\cH\,\tilde{K}\|  (\sigma)^{-1}\|_{\infty} \right)}{\log\tilde{\gamma}}\,.
    \end{align*}
\end{theorem}
\begin{proof}
    Note that it follows from the proof of~\Cref{cor:pptfulleeb} that the Choi matrix of the channel $\Phi$ is the direct sum of the Choi matrix of primitive quantum channels.
    The proof proceeds similarly to the one of \Cref{tEB_upper} together with a use of \Cref{EB}, restricted to each one of these direct sums.
\end{proof}

\subsection{Approximate quantum Markov networks via strong mixing} \label{sec:aQMN}
Another application of our results is to determine when all outputs of a quantum Gibbs sampler are close to being quantum Markov networks.
A quantum Gibbs sampler $(\Phi^{\operatorname{Gibbs}}_t)_{t\ge 0}$ is a continuous time quantum Markov semigroup that converges towards the Gibbs state of the Hamiltonian of a given lattice spin system. In the case when the Hamiltonian $H$ consists of commuting terms, it was shown in \cite{Kastoryano2016} that at large enough temperature the spectral gap $\lambda$ can be chosen independently of the size of the lattice. In particular, this means that the following bound on the trace distance between any initial state $\rho$ and the equilibrium Gibbs state $\sigma$:
\begin{align*}
    \|\Phi^{\operatorname{Gibbs}}_t(\rho)-\sigma\|_1\le \frac{1}{2}\,\e^{-\lambda t}\,\|\sigma^{-1}\|_\infty\,.
\end{align*}
On the other hand, the existence of a modified logarithmic Sobolev constant $\alpha_1$ independent of the system size is still an important open problem.

It is a well-known fact that Gibbs states corresponding to commuting interactions are quantum Markov networks \cite{LEIFER20081899,brown2012quantum}: given any tripartition $ABC$ of the lattice such that $B$ shields $A$ away from $C$, $\sigma\equiv \sigma_{ABC}$ is a quantum Markov chain, that is, there exists a quantum channel $\mathcal{R}_{B\to BC}$, usually referred to as the recovery map, such that $\mathcal{R}_{B\to BC}\circ \tr_C(\sigma)=\sigma$. Equivalently, such states are the ones with vanishing conditional quantum mutual information: $I(A:C|B)_\sigma=0$, where given a state $\rho$, $$I(A:C|B)_\rho=H(AB)_\rho+H(BC)_\rho-H(B)_\rho-H(ABC)_\rho\,.$$

More recently, \cite{Fawzi2015} introduced the concept of an $\eps$-approximate quantum Markov chain. Such tripartite quantum states $\rho$ can be defined by requiring that their conditional quantum mutual information is small, that is, $I(A:C|B)_\rho\le \eps$, for $\eps>0$. It was shown in \cite{Fawzi2015} that this is equivalent to the existence of a map $\mathcal{R}_{B\to BC}$ such that $\rho$ is close to $\mathcal{R}_{B\to BC}\circ \tr_C(\rho)$, say in trace distance. A simple use of the data processing inequality implies that states that are close to a quantum Markov chain are $\eps$-approximate Markov chains. However, the opposite statement does not hold true in general \cite{Christandl2012,Ibinson2008}.

In this section, we are interested in deriving bounds on the time it takes a state evolving according to a quantum Gibbs sampler to become an $\eps$-approximate quantum Markov network.
More precisely, let $ABC$ be a tripartition of a lattice spin system $\Lambda$ as above, and let $(\Phi^{\operatorname{Gibbs}}_t)_{t\ge 0}$ a continuous-time quantum Gibbs sampler.  A very simple bound can be achieved by rudimentary triangle inequality: letting $\rho_t:=\Phi^{\operatorname{Gibbs}}_t(\rho)$,

\begin{align*}
    \|  \rho_t-\cR_{B\to BC}(\rho_t)\|_1 & \le \|   \rho_t-\sigma\|_1+\|\cR_{B\to BC}(\sigma)-\cR_{B\to BC}(\rho_t)\|_1 \\
                                         & \le 2 \|\rho_t-\sigma\|_1                                                    \\
                                         & \le\e^{-\lambda t}\|\sigma^{-1}\|_\infty\,.
\end{align*}
and therefore
\begin{align*}
    t^{A-B-C}_{\operatorname{aQMC}}(\eps)\le \frac{1}{\lambda}\ln\left( \frac{\|\sigma\|_\infty^{-1}}{\eps}\right)\,,
\end{align*}
where the time to become an $\eps$-quantum Markov chain with respect to a clique $A-B-C$ is defined as
\begin{align*}
    t^{A-B-C}_{\operatorname{aQMC}} & ((\Phi^{\operatorname{Gibbs}}_t)_{t\ge 0},\eps)                                                                                                                                                                                       \\
                                    & \!\!\!\!\!\!\!\!\!\!\!\!  :=\left\{ t\ge 0:\, \exists \,\mathcal{R}_{B\to BC}:\,\cB(\cH_B)\to \cB(\cB_B\otimes \cB_C)\,\text{CPTP},\,\|\rho_t-\mathcal{R}_{B\to  BC}(\rho_t)\|_1\le \eps    \,\forall\rho\in\cD(\cH_{ABC})\right\}\,.
\end{align*}

A much better bound can however be derived from the joint use of the following bound from \cite{Fawzi2015}: for any tripartite state $\rho_{ABC}\in\cD(\cH_A\otimes \cH_B\otimes \cH_C)$, there exists a CPTP map $\Gamma_{B\to BC}$ such that
\begin{align*}
    I(A:C|B)\ge -\log\,F(\rho_{ABC},\Gamma_{B\to BC}\circ \tr_C(\rho_{ABC}))\,,
\end{align*}
where $F(\rho,\sigma)=(\tr|\sqrt{\rho}\sqrt{\sigma}|)^2$ denotes the fidelity of to states $\rho$ and $\sigma$. By the Fuchs and van de Graaf inequality, this implies that:
\begin{align}\label{FawziRenner}
    I(A:C|B)_\rho\ge -\log\,\left(1-\frac{1}{4}\,\|\rho_{ABC}-\Gamma_{B\to BC}\circ\tr_C(\rho_{ABC})\|_1^2\right)\,.
\end{align}
Together with the Alicki-Fannes-Winter continuity bounds on conditional entropies, we obtain the following result:
\begin{proposition}
    Assume that $(\Phi_t^{\operatorname{Gibbs}})_{t\ge 0}$  $\operatorname{(}$\ref{SDP}$\operatorname{)}$. Then, for any partition $A-B-C$ of the lattice such that $B$ shields $A$ away from $C$:
    \begin{align*}
        t^{A-B-C}_{\operatorname{aQMC}}((\Phi^{\operatorname{Gibbs}}_t)_{t\ge 0},\eps)\le \min\left\{\frac{1}{\gamma}\,\log\left( \frac{(\log |B|+1)K+\frac{K^2}{2}}{\log\left[(1-\frac{1}{4}\eps^2)^{-1}\right]}   \right)\,,     \frac{1}{\gamma}\,\log\left( \frac{(\log |A|+1)K+\frac{K^2}{2}}{\log\left[(1-\frac{1}{4}\eps^2)^{-1}\right]}   \right)     \right\}\,.
    \end{align*}

\end{proposition}
\begin{proof}
    Denote $\delta(t):=\|\rho_t-\sigma\|_1$ and consider the following difference between conditional quantum mutual informations:
    \begin{align*}
        I(A:C|B)_{\rho_t} & =   |   I(A:C|B)_{\rho_t}-  I(A:C|B)_\sigma|                                                    \\
                          & \le |H(B|C)_{\rho_t}-H(B|C)_\sigma|+|H(B|AC)_{\rho_t}-H(B|AC)_\sigma|                           \\
                          & \le  \delta(t)\,\ln |B|+(1+\delta(t)/2)  \,h_2 \left( \frac{\delta(t)/2}{1+\delta(t)/2} \right) \\
                          & \le \delta(t)\,(\ln|B|+1)+\frac{\delta(t)^2}{2}                                                 \\
                          & \le \left[(\ln|B|+1)\,K+\frac{K^2}{2}\right]\,\e^{-\gamma \,t }
    \end{align*}
    where the second inequality follows from Lemma 2 of \cite{Winter2016}, and given $0\le p\le 1$, $h_2(p)$ denotes the binary entropy of the distribution $\{p,1-p\}$.
    The result directly follows.
    \qed
\end{proof}
\medskip
The above proposition together with the bound coming from the Poincar\'{e} inequality yields a bound where $K=\sqrt{\|\sigma^{-1}\|_\infty} $ and $\gamma=\lambda$. Having access to a modified logarithmic Sobolev inequality would provide an even stronger bound: by Theorem 4 of \cite{Ibinson2008},
\begin{align*}
    I(A:C|B)_{\rho_t}\le D(\rho_t\|\sigma)\le \log(\|\sigma^{-1}\|_\infty)\,\e^{-2\alpha_1\,t}\,.
\end{align*}
This together with (\ref{FawziRenner}) yields
\begin{align*}
    t^{A-B-C}_{\operatorname{aQMC}}((\Phi^{\operatorname{Gibbs}}_t)_{t\ge 0},\eps)\le \frac{1}{2\alpha_1}\log\left( \frac{\log(\|\sigma^{-1}\|_\infty)}{\log\left[(1-\frac{1}{4}\eps^2)^{-1}\right]}  \right)\,.
\end{align*}

\subsection{Time to become mixed unitary} \label{sec:tMU}
Functional analytical methods can also be used in order to derive other convergence results. As an example, we provide a bound on the time required by a doubly stochastic discrete time quantum Markov semigroup to come close to a random unitary channel. Birkhoff's theorem states that any doubly stochastic matrix may be written as a convex combination of permutation matrices. It is well known that the quantum analogue of this theorem does not hold in general, i.e. there are doubly stochastic quantum channels that cannot be written as a convex combination of unitary channels, that is, mixed unitary channels. We refer to e.g. ~\cite{Mendl2009} and references therein for a discussion of this problem. However, we will show that applying the channel often enough might lead it to become a mixed unitary channel.

Let $\{\Phi^n\}_{n\in\NN}$ be a discrete time quantum Markov semigroup on $\cB\p \cH)$, with invariant state $\mathbb{I}/d_\cH\in\cD_+\p \cH)$. The mixed unitary time $n_{\operatorname{MU}}(\Phi)$ of $\{\Phi^n\}_{n\in\NN}$ is defined as follows:
\begin{equation*}
    n_{\operatorname{MU}}(\{\Phi^n\}_{n\in\NN})\overset{\operatorname{def}}{=}\min\left\{n_0\in\NN|~\forall n\ge n_0,~ \Phi^n   \text{ is a mixed unitary channel}\right\}\,.
\end{equation*}
The goal of this section is to estimate this time, and we will follow an approach that is similar to the last sections on entanglement breaking times. In~\cite[Corollary 2]{watrous2008mixing}, the author shows that if
\begin{align*}
    \left\|J\lb \Phi\rb-\frac{\mathbb{I}}{d_\cH^2}\right\|_\infty\leq\frac{1}{d_\cH(d_\cH^2-1)}
\end{align*}
and $\Phi$ is a doubly stochastic channel, then it is a mixed unitary channel.
By observing that
\begin{align*}
    \left\|J\lb \Phi\rb-\frac{\mathbb{I}}{d_\cH^2}\right\|_\infty\leq  \left\|J\lb\Phi\rb-\frac{\mathbb{I}}{d_\cH^2}\right\|_p
\end{align*}
for all $p\geq1$, we may immediately adapt the bounds in the previous sections for entanglement breaking times to this scenario as well. It would be tedious to repeat all the argument given before to the mixed unitary case, so we just show how to adapt the bound in discrete time we obtain from the Poincar\'e inequality to this setting. For simplicity we will state the result we obtain for self-adjoint channels, but it should be clear how to generalize the result.
\begin{theorem}\label{tEB_upper2}
    For any primitive doubly stochastic self-adjoint quantum channel $\Phi:\cB(\cH) \to \cB(\cH)$, we have
    \begin{align*}
        n_{\operatorname{MU}}(\{\Phi^n\}_{n\in\NN})\leq-\frac{3\log (d_{\cH})}{2\log(1-\lambda(\cL^*))}\,,
    \end{align*}
    where $\lambda(\cL^*)$ is the spectral gap of $\cL^*:=\Phi^*\circ\,\hat{\Phi}-\id$.
\end{theorem}
\begin{proof}
    It follows from \Cref{prop:discrete-Poincare} that
    \begin{align}\label{equ:convmixedunitary}
        \left\| J\lb \Phi^n\rb-\frac{\mathbb{I}}{d_\cH^2}\right\|_{1}\le\,\sqrt{d_\cH}\, \,(1-\lambda(\cL^*))^{\frac{n}{2}}\,.
    \end{align}
    Setting $n$ as in the statement ensures that the R.H.S. of the equation above is at most $d_\cH^{-3}$, which implies that the quantum channel $\Phi^n$
    is mixed unitary.
    \qed
\end{proof}
\medskip
The above result supplements those of~\cite{Mendl2009}, in particular those of Section V. There the authors show several ``revivals'' of Birkhoff's theorem in the quantum case, that is, relaxed conditions under which a doubly stochastic quantum channel becomes a mixture of unitaries.
The above result shows that all primitive doubly stochastic quantum channels become mixed unitary if we apply them often enough.
This has interesting applications to environment assisted capacities, introduced in~\cite{gregorattiwerner}.
There they identify the class of channels that allow for complete correction, given a suitable feedback of classical
information from the environment, with the set of mixed unitary channels. Together with our last result, this implies the counterintuitive fact that a primitive doubly stochastic channel becomes perfectly correctable in this scenario if we apply it often enough, i.e., add more noise.

\appendix

\section{Proof of \Cref{prop:decomp-irred}}\label{proof:prop:decomp-irred}

Let us show that given $z, \{p_n\}_{n=0}^{z-1}$, $\sigma$, and $\Phi_Q$, the decomposition \eqref{eq:Phi-irred} gives an irreducible quantum channel.
Note, by the definition of $P_n$, for any $X\in\mathcal{B}(\cH)$,
\[
    P_j \circ P_k (X) = \tr[ u^{-k} X] \tr[ u^{k-j}\sigma] u^j\sigma  = \delta_{jk} P_j(X)
\]
using $\tr[ \sigma p_n ] = \frac{1}{z}$ and the formula
\begin{equation}\label{eq:roots-of-unity-sum}
    \sum_{n=0}^{z-1} \theta^{n m} = \begin{cases}
        z & m = zk \text{ for some } k\in \mathbb{Z} \\
        0 & \text{otherwise.}
    \end{cases}
\end{equation}
Since $P_0(X) = \tr[X]\sigma$, we have for $j\neq 0$, $0 = P_0\circ P_j (X) = \tr[P_j(X)] \sigma$, yielding that $P_j$ is trace-annihilating: $\tr[P_j(X)] =0$ for all $X\in \cB(\cH)$. In the same way, using assumption~\labelcref{it:Phi_Q-Pj},  $\Phi_Q$ is trace-annihilating. Thus, $\Phi = P_0 + \sum_{n=1}^{z-1} \theta^n P_n + \Phi_Q$ is trace-preserving.

Next, we prove \eqref{eq:irred-choi-P}, which will prove $\Phi$ is CP via assumption~\labelcref{it:Phi_Q-J}. For $\Phi_P := \sum_{m=0}^{z-1} \theta^m P_m$, we have $\Phi_P^k = \sum_{m=0}^{z-1} \theta^{km} P_m$. Then, for any $X\in \cB(\cH)$, we have the discrete Fourier-type computation,
\begin{align}
    \Phi_P^k(X) & = \sum_{m=0}^{z-1}\theta^{km} \tr[ u^{-m} X]u^m \sigma = \sum_{m,n,\ell=0}^{z-1} \theta^{km} \tr[\theta^{-m n}  p_n X] \theta^{\ell m} p_\ell \sigma\nonumber \\
                & =\sum_{m,n,\ell=0}^{z-1} \theta^{m(k-n+\ell)} \tr[ p_n X]  p_\ell \sigma =\sum_{n,\ell=0}^{z-1} z\delta_{\ell = n-k} \tr[ p_n X]  p_\ell \sigma\nonumber      \\
                & =z\sum_{n=0}^{z-1} \tr[ p_n X]  p_{n-k} \sigma
    \label{eq:PhikP_formula}
\end{align}
using \eqref{eq:roots-of-unity-sum}. Next, let $\{\ket{i}\}_{i=0}^{d_{\cH}-1}$ be an orthonormal basis of $\cH$ such that the first $\rank(p_0)$ elements are a basis for $p_0 \cH$, the next $\rank(p_1)$ elements are a basis for $p_2\cH$, and so on. We have $p_0 = \sum_{i=0}^{\rank(p_0)-1} \ket{i}\bra{i}$, and so forth. Thus,
\begin{align}
    J(\Phi_P^k) & =\sum_{i,j = 0}^{d-1} P^k(\ket{i}\bra{j}) \otimes \ket{i}\bra{j} = z\sum_{i,j}\sum_{n=0}^{z-1} \tr[p_n \ket{i}\bra{j}] p_{n-k} \sigma \otimes \ket{i}\bra{j}\nonumber \\
                & = z \sum_{n=0}^{z-1} \sum_{i = 0}^{d-1}\braket{i | p_n | i} p_{n-k} \sigma \otimes \ket{i}\bra{i} = z\sum_{n=0}^{z-1} \sigma p_{n-k}  \otimes p_n\label{eqJPhiP}      \\
                & =z \big(\sigma \otimes\one \big) L_{k}\nonumber.
\end{align}
In particular, $J(\Phi_P) =z \big(\sigma \otimes \one\big) L_{1}$. Thus, by assumption \labelcref{it:Phi_Q-J},
\[
    J(\Phi) = J(\Phi_P) + J(\Phi_Q) \geq J(\Phi_P) - z \big(\sigma \otimes \one\big) L_{1} =0
\]
and hence $\Phi$ is CP. Since $\Phi$ is CPTP, we can use \eqref{eq:irred-erogdic-average} to prove $\Phi$ is irreducible.  We have
\begin{align*}
    \frac{1}{M}\sum_{n=0}^{M-1} \Phi^n = P_0 + \frac{1}{M} \sum_{m=1}^{z-1} \frac{1 - \theta^{Mm}}{ 1 - \theta^m} P_m + \frac{1}{M}\sum_{n=0}^{M-1} \Phi_Q^n
\end{align*}
using the geometric series $\sum_{n=0}^{M-1} \theta^{mM} = \frac{1 - \theta^{mM}}{1 - \theta^m}$ for $m\neq 0$, which is valid as $\theta^m \neq 1$. Since $P_0[X] = \tr[X]\sigma$, it remains to show that the latter two terms vanish in the limit $M\to \infty$. In fact, since $ \sum_{m=1}^{z-1} \frac{1 - \theta^{Mm}}{ 1 - \theta^m} P_m$ is bounded in norm uniformly in $M$, the second term vanishes asymptotically. Next, since $\ell := \spr(\Phi_Q) < 1$ by assumption~\labelcref{it:Phi_Q-spr}, for $\eps = \frac{1 - \ell}{2} > 0$,  Gelfand's theorem gives that there is $n_0 > 0$ such that (in any matrix norm $\|\cdot\|$), for all $n\geq n_0$,
\[
    \| \Phi_Q^n \|\leq (\ell+\eps)^n < 1.
\]
We may write
\[
    \frac{1}{M}\sum_{n=0}^{M-1} \Phi_Q^n = \frac{1}{M}\sum_{n=0}^{n_0} \Phi_Q^n + \frac{1}{M}\sum_{n=n_0 + 1}^{M-1} \Phi_Q^n.
\]
Since $\sum_{n=0}^{n_0} \Phi_Q^n $ is bounded in norm independently of $M$, the first term vanishes asymptotically; the second term is bounded in norm by the triangle inequality and the geometric series $\sum_{n=0}^\infty (\ell+\eps)^n = \frac{1}{1 - (\ell + \eps)}$. Thus, the limit
\[
    \lim_{M\to \infty} \frac{1}{M}\sum_{n=0}^{M-1} \Phi^n = P_0
\]
holds in (any) norm. In particular, we have \eqref{eq:irred-erogdic-average}, so $\Phi$ is irreducible.
\qed

\paragraph{Acknowledgements}
E.H. would like to thank Yan Pautrat for discussions on the properties of irreducible quantum channels, and Ivan Bardet for discussions on the PPT$^2$ conjecture. E.H. is supported by the Cantab Capital Institute for the Mathematics of Information (CCIMI).
C.R. acknowledges support by the DFG cluster of excellence 2111 (Munich
Center for Quantum Science and Technology).
D.S.F. would like to thank Alexander M\"uller-Hermes for interesting discussions. D.S.F. acknowledges financial support from VILLUM FONDEN via the QMATH Centre of Excellence (Grant no. 10059) and from the QuantERA ERA-NET Cofund in Quantum Technologies implemented within the European Union's Horizon 2020 Programme (QuantAlgo project) via the Innovation Fund Denmark. We would like to thank the anonymous QIP 2019 reviewers for pushing us more to consider  the non-irreducible case, and an anonymous referee for suggesting the very simple proof of \Cref{lemma1}.

\bibliographystyle{abbrv}
\bibliography{refs}

\end{document}